\documentclass[11pt]{article}

\newcommand{\jmp}{J. Math. Phys.~}

\newcommand{\njp}{New. J. Phys.~}

\newcommand{\prl}{Phys. Rev. Lett.~}
\newcommand{\pra}{Phys. Rev. A~}
\newcommand{\pla}{Phys. Lett. A~}

\newcommand{\laa}{Lin. Alg. App.~}

\usepackage{young}
\usepackage[latin1]{inputenc}
\usepackage[linesnumbered,boxed,ruled,commentsnumbered]{algorithm2e}
\usepackage{appendix}
\usepackage{epstopdf}
\usepackage{xcolor}
\usepackage{subfigure}
\usepackage[T1]{fontenc}
\usepackage[sc]{mathpazo}
\usepackage{amsmath}
\usepackage{amssymb}
\usepackage{graphicx,color}
\usepackage{framed}
\usepackage{multirow}
\usepackage{enumerate}
\usepackage{amsthm}
\usepackage{amsfonts,mathrsfs}
\usepackage{geometry} 
\usepackage{eepic}
\usepackage{ifthen}
\usepackage[vcentermath]{youngtab}
\usepackage[unicode=true,pdfusetitle, bookmarks=true,bookmarksnumbered=false,bookmarksopen=false, breaklinks=false,pdfborder={0 0 0},backref=false,colorlinks=false] {hyperref}
\hypersetup{
colorlinks,linkcolor=myurlcolor,citecolor=myurlcolor,urlcolor=myurlcolor}
\definecolor{myurlcolor}{rgb}{0,0,0.7}

\usepackage{color}

\newcommand{\blue}{\textcolor{blue}}

\newcommand{\proj}[1]{| #1\rangle\!\langle #1 |}

\usepackage{hyperref}
\hypersetup{pdfpagemode=UseNone}

\newcommand{\tinyspace}{\mspace{1mu}}

\newcommand{\op}[1]{\operatorname{#1}}

\newcommand{\abs}[1]{\left\lvert\tinyspace #1 \tinyspace\right\rvert}

\newcommand{\norm}[1]{\left\lVert\tinyspace #1 \tinyspace\right\rVert}

\renewcommand{\det}{\operatorname{det}}
\renewcommand{\t}{{\scriptscriptstyle\mathsf{T}}}

\newcommand{\setft}[1]{\mathrm{#1}}

\newcommand{\density}[1]{\setft{D}\left(#1\right)}

\newcommand{\pos}[1]{\setft{Pos}\left(#1\right)}

\renewcommand{\vec}{\op{vec}}
\newcommand{\im}{\op{im}}
\newcommand{\rank}{\op{rank}}
\newcommand{\spn}{\op{span}}
\newcommand{\sign}{\op{sign}}

\newcommand{\supp}{{\operatorname{supp}}}

\def\GL{\mathsf{GL}}

\def\SO{\mathsf{SO}}
\def\SU{\mathsf{SU}}

\def\Ad{\mathrm{Ad}}

\def\haar{\mathrm{Haar}}

\def \dif {\mathrm{d}}
\def \diag {\mathrm{diag}}

\def \re {\mathrm{Re}}
\def \im {\mathrm{Im}}

\def\complex{\mathbb{C}}
\def\real{\mathbb{R}}

\def\I{\mathbb{1}}

\def\zero{\mathbf{0}}

\def\bslambda{\boldsymbol{\lambda}}
\def\bssigma{\boldsymbol{\sigma}}

\newenvironment{mylist}[1]{\begin{list}{}{
    \setlength{\leftmargin}{#1}
    \setlength{\rightmargin}{0mm}
    \setlength{\labelsep}{2mm}
    \setlength{\labelwidth}{8mm}
    \setlength{\itemsep}{0mm}}}
    {\end{list}}

\def\ot{\otimes}

\newcommand{\inner}[2]{\langle #1 , #2\rangle}
\newcommand{\iinner}[2]{\langle #1 | #2\rangle}
\newcommand{\out}[2]{| #1\rangle\langle #2 |}
\newcommand{\Inner}[2]{\left\langle #1 , #2\right\rangle}
\newcommand{\Innerm}[3]{\left\langle #1 \left| #2 \right| #3 \right\rangle}

\newcommand{\Herm}{\mathrm{Herm}}

\newcommand{\pa}[1]{(#1)}
\newcommand{\Pa}[1]{\left(#1\right)}

\newcommand{\Br}[1]{\left[#1\right]}
\newcommand{\set}[1]{\{#1\}}
\newcommand{\Set}[1]{\left\{#1\right\}}

\newcommand{\ket}[1]{|#1\rangle}

\DeclareMathOperator{\trace}{Tr}
\newcommand{\ptr}[2]{\trace_{#1}\pa{#2}}
\newcommand{\Ptr}[2]{\trace_{#1}\Pa{#2}}
\newcommand{\tr}[1]{\ptr{}{#1}}
\newcommand{\Tr}[1]{\Ptr{}{#1}}

\newcommand{\Abs}[1]{\left|\tinyspace#1\tinyspace\right|}

\def\cA{\mathcal{A}}\def\cB{\mathcal{B}}\def\cC{\mathcal{C}}\def\cE{\mathcal{E}}
\def\cF{\mathcal{F}}\def\cG{\mathcal{G}}

\def\cP{\mathcal{P}}\def\cQ{\mathcal{Q}}\def\cT{\mathcal{T}}

\def\bsA{\boldsymbol{A}}\def\bsC{\boldsymbol{C}}
\def\bsF{\boldsymbol{F}}\def\bsG{\boldsymbol{G}}\def\bsH{\boldsymbol{H}}
\def\bsL{\boldsymbol{L}}\def\bsM{\boldsymbol{M}}\def\bsO{\boldsymbol{O}}
\def\bsP{\boldsymbol{P}}\def\bsQ{\boldsymbol{Q}}\def\bsT{\boldsymbol{T}}
\def\bsU{\boldsymbol{U}}\def\bsV{\boldsymbol{V}}\def\bsW{\boldsymbol{W}}\def\bsX{\boldsymbol{X}}\def\bsY{\boldsymbol{Y}}
\def\bsZ{\boldsymbol{Z}}

\def\bsa{\boldsymbol{a}}\def\bsb{\boldsymbol{b}}
\def\bsf{\boldsymbol{f}}
\def\bsn{\boldsymbol{n}}
\def\bsp{\boldsymbol{p}}\def\bsr{\boldsymbol{r}}
\def\bsu{\boldsymbol{u}}\def\bsv{\boldsymbol{v}}\def\bsw{\boldsymbol{w}}\def\bsx{\boldsymbol{x}}
\def\bsz{\boldsymbol{z}}

\def\sA{\mathscr{A}}\def\sB{\mathscr{B}}
\def\sI{\mathscr{I}}

\def\sR{\mathscr{R}}

\def\X{\textsf{X}}\def\Y{\textsf{Y}}

\def\bbC{\mathbb{C}}\def\bbE{\mathbb{E}}
                  
\def\bbN{\mathbb{N}}
\def\bbQ{\mathbb{Q}}\def\bbR{\mathbb{R}}

\def\bbZ{\mathbb{Z}}

\def\sfO{\mathsf{O}}

\def\sfU{\mathsf{U}}

\newtheorem{thrm}{Theorem}[section]
\newtheorem{lem}[thrm]{Lemma}
\newtheorem{prop}[thrm]{Proposition}
\newtheorem{cor}[thrm]{Corollary}
\theoremstyle{definition}
\newtheorem{definition}[thrm]{Definition}

\numberwithin{equation}{section}

\newcounter{questionnumber}

\begin{document}

\title{\bf\large A Survey of Bargmann Invariants: Geometric Foundations and Applications}

\author{\blue{Lin Zhang}\footnote{E-mail: godyalin@163.com}\quad and\quad \blue{Bing Xie}\\
   {\it\small School of Science, Hangzhou Dianzi University, Hangzhou 310018, PR~China}}

\date{}
\maketitle

\begin{abstract}
Bargmann invariants, a class of unitary-invariant quantities arising
from the overlaps of quantum state vectors, provide a profound and
unifying framework for understanding the relative geometry of the
projective Hilbert space. This survey offers a comprehensive
overview of their theoretical characterization and practical
applications, with particular emphasis on recent progress in
determining the full structure of their admissible set. The core of
this review demonstrates how these invariants serve as a powerful
tool for characterizing the intrinsic geometry of the space of
quantum states, leading to applications in determining local unitary
equivalence and constructing a complete set of polynomial invariants
for mixed states. On the operational side, we review the cycle-test
quantum circuit for the direct estimation of Bargmann invariants
without full state tomography, and demonstrate their utility in
witnessing quantum imaginary, discriminating local unitary
equivalence, and detecting entanglement via partial-transpose
moments---with explicit complete invariant sets provided for
two-qubit systems. By connecting fundamental geometric
classification with experimentally feasible estimation protocols,
this survey establishes Bargmann invariants as indispensable probes
of the relational, noncommutative, and geometric structure of
quantum states, and identifies key open problems for multipartite
high-dimensional systems and for quantum resource theories.
\end{abstract}

\newpage\tableofcontents\newpage

\section{Introduction}

Quantum mechanics, since its inception, has revealed a profound and
persistent geometric character underlying its probabilistic
formalism. This geometry is not merely an artifact of representation
but is fundamentally encoded in the complex Hilbert space structure,
manifesting in phenomena such as the Pancharatnam-Berry phase, which
arises from the cyclic evolution of a quantum state. At the heart of
understanding this intrinsic geometric structure lies a class of
gauge-invariant quantities known as Bargmann invariants. Defined by
the cyclic overlaps of quantum state vectors, these complex numbers
transcend the arbitrary choice of phase for individual states,
offering a direct window into the relational and shape-like
properties of the quantum state space itself.

First introduced by Valentin Bargmann in his seminal analysis of ray
spaces and symmetry operations \cite{Bargmann1964}, these invariants
have evolved from a mathematical cornerstone in the theory of
unitary representations to a versatile and powerful framework for
probing the informational geometry of quantum systems. The simplest,
non-trivial Bargmann invariant---the triple product of inner
products for three quantum states---is intimately linked to the
geometric phase, providing its foundational complex antecedent.
Higher-order invariants, constructed from larger sets of states,
encode increasingly detailed information about the arrangement of
states within the projective Hilbert space, effectively serving as
coordinates for its geometric features \cite{Zhang2025PRA1,
Pratapsi2025PRA}.

This survey aims to provide a comprehensive overview of Bargmann
invariants, with a particular focus on their pivotal role in shaping
and elucidating the informational geometry of quantum states. We
will trace their journey from a key insight in the theory of
geometric phases to a modern toolkit for quantum information
science. The discussion begins by elucidating their fundamental
definition, gauge invariance, and algebraic properties. We will then
demonstrate how these invariants serve as natural instruments for
characterizing the intrinsic geometry of both pure and mixed quantum
states, including an analysis of their admissible numerical ranges
\cite{Xu2026PLA}. This geometric perspective leads to significant
applications, including a powerful framework for determining local
unitary equivalence of states and constructing a complete set of
polynomial invariants for mixed-state spaces \cite{Zhang2025PRA2},
with connections to the characterization of finite frames under
projective unitary equivalence \cite{Chien2016}.

Furthermore, this review highlights the contemporary resurgence of
interest in Bargmann invariants driven by quantum information
theory. We explore their pivotal role in developing operational
methods for directly measuring relational information
\cite{Oszmaniec2024NJP} and geometric features, most notably in
protocols for detecting quantum entanglement without resorting to
full state tomography \cite{Zhang2025PRA2}. By circumventing the
need for a complete density matrix reconstruction, such approaches
underscore the practical power of these geometric quantities. This
operational viewpoint is deeply connected to advances in
multivariate trace estimation using quantum algorithms
\cite{Quek2024,Yosef2025SIAM,Azado2025} and the study of related
quantum channels \cite{Zhang2024epjp,Xie2026}. Recent work has also
expanded their purview to new quantum resources, including the
characterization and witnessing of quantum imaginarity
\cite{Fernandes2024PRL, Li2025PRA} and studies of coherence and
contextuality \cite{Wagner2025PhD}.

By synthesizing historical context with recent theoretical and
experimental advances, this survey seeks to elevate the perception
of Bargmann invariants from mathematical curiosities to essential
instruments. They are not merely invariant quantities but are
fundamental probes of the relational, non-commutative, and geometric
fabric of quantum mechanics, offering a unifying language that
connects foundational quantum geometry to cutting-edge quantum
information processing.

This paper is organized as follows. In Section~\ref{sect:2}, we
introduce the concept of joint equivalence, which classifies sets of
quantum states that are equivalent under local unitary
transformations. This framework is essential for the classification
problems we address. Section~\ref{sect:3} presents Bargmann
invariants, the central objects of our study. We define them, review
their properties, and provide necessary background on circulant
matrices and circulant quantum channels. Section~\ref{sect:4}
focuses on circulant Gram matrices resulted in Bargmann invariants.
We characterize the set $\cB_n|_{\mathrm{circ}}$ and study its
convexity. In Section~\ref{sect:5}, we present a main theoretical
result: a complete characterization of when
$\cB_n=\cB_n|_{\mathrm{circ}}$, identifying the conditions under
which every valid set of $n$th-order Bargmann invariants admits a
circulant Gram matrix representation. Section~\ref{sect:6} offers an
alternative characterization of $\cB^\circ_n(d)$, describing the set
of achievable invariants for a given Hilbert space dimension $d$.
Section~\ref{sect:7} shifts to practical considerations, presenting
methods for estimating Bargmann invariants in quantum circuits and
describing concrete protocols for near-term devices.
Section~\ref{sect:8} demonstrates applications of Bargmann
invariants: witnessing quantum imaginarity, discriminating locally
unitary orbits, and entanglement detection. We conclude by
summarizing our findings and discussing open questions.

\section{Joint equivalence}\label{sect:2}

Consider the set $\density{\bbC^d}$ of all quantum states acting on
$\bbC^d$, i.e. the set of all density matrices of size $d$. Unit
vector $\ket{\psi}$ in $\bbC^d$ will be called \emph{wave functions}
and its ranked-one projector $\psi\equiv\proj{\psi}$ will be called
\emph{pure state}. To further develop our framework, we need the
following very basic results.

\begin{prop}[\cite{Kadison1997}]\label{prop:Kadison}
If $\Inner{\cdot}{\cdot}$ is a definite inner product on a complex
vector space $V$ and $\bsu,\bsv\in V$, the following three
conditions are equivalent:
\begin{enumerate}[(i)]
\item $\norm{\bsu+\bsv}=\norm{\bsu}+\norm{\bsv}$
\item $\Inner{\bsu}{\bsv}=\norm{\bsu}\norm{\bsv}$;
\item one of $\bsu$ and $\bsv$ is a non-negative scalar
multiple of the other.
\end{enumerate}
\end{prop}

\begin{proof}
For any scalars $a,b\in\bbC$, it follows that
$$
\norm{a\bsu+b\bsv}^2=\abs{a}^2\norm{\bsu}^2 +
\abs{b}^2\norm{\bsv}^2+2\re(\bar ab\Inner{\bsu}{\bsv}),
$$
which implies that
$$
\Pa{\norm{\bsu} + \norm{\bsv}}^2 - \norm{\bsu+\bsv}^2 =
2\norm{\bsu}\norm{\bsv}-2\re(\Inner{\bsu}{\bsv}).
$$
\begin{itemize}
\item (i)$\Longrightarrow$(ii) The last equation and the
Cauchy-Schwarz inequality give
$$
\re(\Inner{\bsu}{\bsv})=\norm{\bsu}\norm{\bsv}\geqslant
\abs{\Inner{\bsu}{\bsv}}
$$
and therefore
$\Inner{\bsu}{\bsv}=\re(\Inner{\bsu}{\bsv})=\norm{\bsu}\norm{\bsv}$.\\
\item (ii)$\Longrightarrow$(iii) If $a,b\in\bbR$, we get from the
assumption $\Inner{\bsu}{\bsv}=\norm{\bsu}\norm{\bsv}$ that
$$
\norm{a\bsu-b\bsv}^2=\Pa{a\norm{\bsu} - b\norm{\bsv}}^2.
$$
With $a=\norm{\bsv}$ and $b=\norm{\bsu}$, it follows that
$a\ket{\bsu}=b\ket{\bsv}$. Hence either
$\ket{\bsu}=\zero=0\ket{\bsv}$ or $\ket{\bsv}=\frac ab\ket{\bsu}$.\\
\item(iii)$\Longrightarrow$(i) We suppose that
$\ket{\bsu}=\lambda\ket{\bsv}$ for non-negative scalar $\lambda$.
Then
$$
\norm{\bsu}+\norm{\bsv}=(\lambda+1)\norm{\bsv}=\norm{(\lambda+1)\bsv}=\norm{\bsu+\bsv}.
$$
\end{itemize}
This completes the proof.
\end{proof}

\begin{prop}
If $\ket{\psi}$ and $\ket{\phi}$ are unit vectors, then
$\proj{\psi}=\proj{\phi}$ if and only if
$\ket{\psi}=e^{\mathrm{i}\theta}\ket{\phi}$ for some
$\theta\in\bbR$.
\end{prop}

\begin{proof}
In fact, it is easily seen from $\proj{\psi}=\proj{\phi}$ that
$$
1=\iinner{\psi}{\psi}\iinner{\psi}{\psi} =
\iinner{\psi}{\phi}\iinner{\phi}{\psi}\Longrightarrow
\abs{\iinner{\psi}{\phi}}=1,
$$
implying that $\iinner{\psi}{\phi}=e^{-\mathrm{i}\theta}$ for some
$\theta\in\bbR$. Thus
$\Inner{\psi}{e^{\mathrm{i}\theta}\phi}=1=\norm{\psi}\norm{e^{\mathrm{i}\theta}\phi}$,
which is true if and only if
$\ket{\psi}=e^{\mathrm{i}\theta}\ket{\phi}$, by the saturation
condition of Cauchy-Schwarz inequality in
Proposition~\ref{prop:Kadison}.
\end{proof}

\begin{definition}[(Projective) unitary equivalence]
For any given wave functions $\ket{\phi}$ and $\ket{\psi}$ in
$\bbC^d$,
\begin{enumerate}[(i)]
\item the so-called \emph{unitary equivalence} between $\ket{\phi}$
and $\ket{\psi}$ is that there exists a unitary $\bsU\in\sfU(d)$
such that $\ket{\phi}=\bsU\ket{\psi}$.
\item the so-called \emph{projective unitary equivalence} between $\ket{\phi}$
and $\ket{\psi}$ is that there exists a unitary $\bsU\in\sfU(d)$
such that
$$
\proj{\phi}=\bsU\proj{\psi}\bsU^\dagger\Longleftrightarrow
\ket{\phi}=e^{\mathrm{i}\theta}\bsU\ket{\psi}
$$
for some $\theta\in\bbR$.
\end{enumerate}
\end{definition}

\begin{definition}[Joint (projective) unitary equivalence]
For two $n$-tuples of vectors
$\Psi=(\ket{\psi_1},\ldots,\ket{\psi_n})$ and
$\Phi=(\ket{\phi_1},\ldots,\ket{\phi_n})$ in $\bbC^d$,
\begin{enumerate}[(i)]
\item the so-called \emph{joint unitary equivalence} between $\Psi$
and $\Phi$ is that there exists a unitary $\bsU\in\sfU(d)$, the
group of complex $d\times d$ unitary matrices, such that
$\ket{\phi_k}=\bsU\ket{\psi_k}$ for $k\in\set{1,\ldots,n}$. Denote
this fact by $\Phi=\bsU\Psi$.
\item the so-called \emph{joint projective unitary equivalence} between $\Psi$
and $\Phi$ is that there exists a unitary $\bsU\in\sfU(d)$ such that
$$
\proj{\phi_k}=\bsU\proj{\psi_k}\bsU^\dagger
$$
for $k\in\set{1,\ldots,n}$.
\end{enumerate}
\end{definition}

The notion of Gram matrix will be used in the characterization of
joint (projective) unitary equivalence for two $n$-tuples of vectors
in $\bbC^d$. Let me explain about it.

\begin{definition}[Gram matrix]
The so-called \emph{Gram matrix} for $n$-tuple of vectors
$\Psi=(\ket{\psi_1},\ldots,\ket{\psi_n})$, where $\ket{\psi_k}$'s
are in $\bbC^d$, is defined as
\begin{eqnarray}
G(\Psi) =\Pa{\begin{array}{ccccc}
                 \Inner{\psi_1}{\psi_1} & \Inner{\psi_1}{\psi_2} & \Inner{\psi_1}{\psi_3} & \cdots & \Inner{\psi_1}{\psi_n} \\
                 \Inner{\psi_2}{\psi_1} & \Inner{\psi_2}{\psi_2} & \Inner{\psi_2}{\psi_3} & \cdots & \Inner{\psi_2}{\psi_n} \\
                 \vdots & \vdots & \ddots & \ddots &\vdots \\
                 \Inner{\psi_{n-1}}{\psi_1} & \Inner{\psi_{n-1}}{\psi_2} & \Inner{\psi_{n-1}}{\psi_3} &\cdots &
                 \Inner{\psi_{n-1}}{\psi_n}\\
                 \Inner{\psi_n}{\psi_1} & \Inner{\psi_n}{\psi_2} & \Inner{\psi_n}{\psi_3} &\cdots & \Inner{\psi_n}{\psi_n}
                 \end{array}
}.
\end{eqnarray}
\end{definition}

\begin{prop}[\cite{Chien2016}]\label{prop:jequiv}
For two $n$-tuples of vectors
$\Psi=(\ket{\psi_1},\ldots,\ket{\psi_n})$ and
$\Phi=(\ket{\phi_1},\ldots,\ket{\phi_n})$ in $\bbC^d$, we have the
following statements:
\begin{enumerate}[(i)]
\item both $\Psi$ and $\Phi$ is joint unitary equivalent if and only
if $G(\Phi)=G(\Psi)$.
\item both $\Psi$ and $\Phi$ is joint projective unitary equivalent if and only
if $G(\Phi)=\bsT^\dagger G(\Psi)\bsT$ for $\bsT\in\sfU(1)^{\times
n}$.
\end{enumerate}
\end{prop}

\begin{proof}
(i) Clearly if $\Psi$ and $\Phi$ is joint unitary equivalent, then
$G(\Psi)=G(\Phi)$. Reversely, we assume that $G(\Phi)=G(\Psi)$. Let
\begin{eqnarray*}
V_1=\spn\Set{\ket{\psi_k}:k=1,\ldots,n},\quad
V_2=\spn\set{\ket{\phi_k}:k=1,\ldots,n}.
\end{eqnarray*}
It is easily seen that
$\dim(V_1)=\dim(V_2)=\rank(G(\Psi))=\rank(G(\Phi))$ due to the fact
that $G(\Phi)=G(\Psi)$. Define a mapping $\bsM:V_1\to V_2$ as
$\ket{\phi_k}=\bsM\ket{\psi_k}$, where $k=1,\ldots,N$. Then by
linear extension to the whole space $V_1$. Now for any vector
$\ket{\psi}\in V_1$, $\ket{\psi}=\sum_k\lambda_k\ket{\psi_k}$, we
get that
\begin{eqnarray*}
\bsM\ket{\psi} = \sum_k\lambda_k \bsM\ket{\psi_k} =
\sum_k\lambda_k\ket{\phi_k}:=\ket{\phi}\in V_2.
\end{eqnarray*}
Furthermore,
\begin{eqnarray*}
\Innerm{\psi}{\bsM^\dagger \bsM}{\psi}
=\Inner{\phi}{\phi}=\sum_{i,j}\bar\lambda_i\lambda_j
\Inner{\phi_i}{\phi_j}= \sum_{i,j}\bar\lambda_i\lambda_j
\Inner{\psi_i}{\psi_j}=\Inner{\psi}{\psi}.
\end{eqnarray*}
Thus $\bsM^\dagger\bsM=\I_{V_1}$. Similarly we have
$\bsM\bsM^\dagger=\I_{V_2}$.
Therefore both $\Psi$ and $\Phi$ are joint unitary equivalent.\\
(ii) If both $\Phi$ and $\Psi$ are joint projective unitary
equivalent, then exist unitary matrix $\bsU\in\sfU(d)$ and
$\theta_k\in\bbR$, where $k=1,\ldots,n$, such that
\begin{eqnarray*}
\ket{\phi_k}=e^{\mathrm{i}\theta_k}\bsU\ket{\psi_k}\Longrightarrow\Inner{\phi_i}{\phi_j}=e^{-\mathrm{i}\theta_i}e^{\mathrm{i}\theta_j}\Inner{\psi_i}{\psi_j}.
\end{eqnarray*}
Set
$\bsT=\diag(e^{\mathrm{i}\theta_1},\ldots,e^{\mathrm{i}\theta_n})$,
we have $G(\Phi)=\bsT^\dagger G(\Psi)\bsT$. Reversely, assume that
$G(\Phi)=\bsT^\dagger G(\Psi)\bsT$ for $\bsT\in\sfU(1)^{\times n}$.
Define $\Psi'=(\ket{\psi'_1},\ldots,\ket{\psi'_n})$, where
$\ket{\psi'_k}:=e^{\mathrm{i}\theta_k}\ket{\psi_k}$. Now
$G(\Phi)=\bsT^\dagger G(\Psi)\bsT$ implies that
$\Inner{\phi_i}{\phi_j}=e^{-\mathrm{i}\theta_i}e^{\mathrm{i}\theta_j}\Inner{\psi_i}{\psi_j}=\Inner{e^{\mathrm{i}\theta_i}\psi_i}{e^{\mathrm{i}\theta_j}\psi_j}$,
i.e., $\Inner{\phi_i}{\phi_j}=\Inner{\psi'_i}{\psi'_j}$. That is,
$G(\Phi)=G(\Psi')$. This indicates that there exists a unitary
$\bsU\in\sfU(d)$ such that $\Phi=\bsU\Psi'$, which is equivalent to
the fact that $\Phi$ and $\Psi$ are joint projective unitary
equivalent.
\end{proof}

In the following, we turn to discuss the joint unitary similarity of
two $N$-tuples $\Psi=(\rho_1,\ldots,\rho_N)$ and
$\Psi'=(\rho'_1,\ldots,\rho'_N)$, acting on $\bbC^d$.
\begin{definition}[Joint unitary similarity]
For given two $N$-tuples $\Psi=(\rho_1,\ldots,\rho_N)$ and
$\Psi'=(\rho'_1,\ldots,\rho'_N)$, where
$\rho_k,\rho'_k\in\density{\bbC^d}$, we say that both $\Psi$ and
$\Psi'$ are \emph{joint unitary similarity} if there exists a
unitary $\bsU\in\sfU(d)$ such that
\begin{eqnarray}
\rho'_k=\bsU\rho_k\bsU^\dagger,
\end{eqnarray}
where $k\in\set{1,\ldots,N}$, which is denoted by
$\Psi'=\bsU\Psi\bsU^\dagger$.
\end{definition}

Let $K$ be a compact group and let
\begin{eqnarray}
\Pi: K\ni g\mapsto \Pi_g\in\GL(V)
\end{eqnarray}
be a representation of $K$ in a finite dimensional \emph{real}
vector space $V$. We can assume that $\Pi_g\in\sfO(V)$ by the
compactness of $K$. The space of all real polynomials on $V$ is
denoted by $\bbR[V]$. Denote the space of real homogeneous
polynomials on $V$ by $\bbR[V]_n$. Homogeneous polynomials of degree
$n$ in $V$ are mappings of the form
$p(\bsv)=\Inner{\tilde\bsp}{\bsv^{\ot n}}$, where
$\Inner{\cdot}{\cdot}$ is the $K$-invariant inner product in $V^{\ot
n}$ (induced by the inner product on $V$) and $\tilde\bsp\in V^{\ot
n}$ is a tensor encoding the polynomial $p$. $K$-invariant
homogeneous polynomials of degree $n$ must satisfy $\Pi^{\ot
n}_g\tilde\bsp=\tilde\bsp$ for every $g\in K$. Denote the set of all
$K$-invariant polynomials by $\bbR[V]^K$.

Recall a result in invariant theory \cite{Vrana2011}: for
$\bsu,\bsv\in V$, we have $\bsv=\Pi_g\bsu$ for some $g\in K$ if and
only if for every $K$-invariant homogeneous polynomial $p_n\in
\bbR[V]^K$ of degree $n$, we have $p_n(\bsv)=p_n(\bsu)$, where
$n=1,2,\ldots$.

\begin{prop}[\cite{Procesi1976,Oszmaniec2024NJP}]
For given two $N$-tuples $\Psi=(\rho_1,\ldots,\rho_N)$ and
$\Psi'=(\rho'_1,\ldots,\rho'_N)$, where
$\rho_k,\rho'_k\in\density{\bbC^d}$, both $\Psi$ and $\Psi'$ are
joint unitary similarity if and only if for every $n\in\bbN$ and for
every sequence $i_1,i_2,\ldots,i_n$ of numbers from
$\set{1,\ldots,N}$, the corresponding Bargmann invariants of degree
$n$ agree
\begin{eqnarray}
\Tr{\rho_{i_1}\rho_{i_2}\cdots\rho_{i_n}}=\Tr{\rho'_{i_1}\rho'_{i_2}\cdots\rho'_{i_n}}.
\end{eqnarray}
\end{prop}

\begin{proof}
Let $V=\overbrace{\Herm(\bbC^d)\oplus
\Herm(\bbC^d)\oplus\cdots\oplus \Herm(\bbC^d)}^N$. Every $\X\in V$
can be identified with a tuple of linear operators $\X\simeq
(\bsX_1,\ldots,\bsX_N)\in V$, which is equivalent to
$\X=\sum^N_{k=1}\bsX_k\ot\ket{k}\in \Herm(\bbC^d)\ot\bbR^N(\simeq
V)$, where $\bsX_k\in\Herm(\bbC^d)$. Consider the joint conjugation
of the unitary group $\sfU(d)$:
\begin{eqnarray*}
\Pi_{\bsU}\X=(\Ad_{\bsU}(\bsX_1),\ldots,\Ad_{\bsU}(\bsX_N)),\quad
\bsU\in\sfU(d).
\end{eqnarray*}
Under the identification, $V\simeq \Herm(\bbC^d)\ot\bbR^N$, the
joint conjugation $\Pi_{\bsU}\simeq \Ad_{\bsU}\ot\I_N$. Now for any
$\X,\Y\in\Herm(\bbC^d)\ot\bbR^N$, its inner product is defined
\begin{eqnarray*}
\Inner{\X}{\Y} = \sum^N_{k=1}\Inner{\bsX_k}{\bsY_k}.
\end{eqnarray*}
Clearly $\Inner{\Pi_{\bsU}\X}{\Pi_{\bsU}\Y}=\Inner{\X}{\Y}$. Up to
the ordering, we have the following identifications
\begin{eqnarray*}
V^{\ot{n}}\simeq \Herm(\bbC^d)^{\ot n} \ot (\bbR^N)^{\ot n},\quad
\Pi^{\ot n}_{\bsU}\simeq \Ad^{\ot n}_{\bsU}\ot\I^{\ot n}_N.
\end{eqnarray*}
Thus $\tilde\bsp\in V^{\ot n}$ can be written as
\begin{eqnarray*}
\tilde \bsp = \sum^N_{k_1,\ldots,k_n=1}\bsP_{k_1k_2\ldots
k_n}\ot\ket{k_1k_2\ldots k_n},
\end{eqnarray*}
where $\bsP_{k_1k_2\ldots k_n}\in\Herm(\bbC^d)^{\ot
n}=\Herm((\bbC^d)^{\ot n})$. Recall that $\tilde \bsp$ defines an
invariant polynomial $p\in \bbR_n[V]^K$ for $K=\sfU(d)$ if and only
if $(\Ad^{\ot n}_{\bsU}\ot\I^{\ot n}_N)\tilde\bsp=\tilde\bsp$. That
is,
\begin{eqnarray*}
\sum^N_{k_1,\ldots,k_n=1}\Ad^{\ot n}_{\bsU}\bsP_{k_1k_2\ldots
k_n}\ot\ket{k_1k_2\ldots
k_n}=\sum^N_{k_1,\ldots,k_n=1}\bsP_{k_1k_2\ldots
k_n}\ot\ket{k_1k_2\ldots k_n},
\end{eqnarray*}
which is equivalent to $\Ad^{\ot n}_{\bsU}\bsP_{k_1k_2\ldots
k_n}=\bsP_{k_1k_2\ldots k_n}$, i.e.,
\begin{eqnarray*}
\bsU^{\ot n}\bsP_{k_1k_2\ldots k_n}=\bsP_{k_1k_2\ldots k_n}\bsU^{\ot
n}\Longleftrightarrow[\bsU^{\ot n},\bsP_{k_1k_2\ldots
k_n}]=0,\quad\forall i_1,\ldots,i_n;\forall\bsU\in\sfU(d).
\end{eqnarray*}
Thus, by Schur-Weyl duality \cite{Zhang2014}, it holds that
\begin{eqnarray}\label{eq:perrep}
\bsP_{k_1k_2\ldots k_n}\in \spn_{\bbC}\Set{\bsP_{d,n}(\pi):\pi\in
S_n}
\end{eqnarray}
where $\bsP_{d,n}(\pi)\ket{i_1\ldots i_n}=\ket{i_{\pi^{-1}(1)}\cdots
i_{\pi^{-1}(n)}}$, thus $\bsP_{k_1k_2\ldots k_n}$ can be expanded
into some linear combinations of $\bsP_{d,n}(\pi)$'s:
\begin{eqnarray*}
\bsP_{k_1k_2\ldots k_n} = \sum_j c_j\bsP_{d,n}(\pi_j).
\end{eqnarray*}
Let us consider $\tilde\bsp_\pi :=\bsP_{d,n}(\pi)\ot\ket{i_1\ldots
i_n}$ corresponding to polynomial $p_\pi$. In fact,
\begin{eqnarray*}
p_\pi(\X) = \Inner{\tilde\bsp_\pi}{\X^{\ot n}} =
\Inner{\bsP_{d,n}(\pi)\ot\ket{i_1\ldots i_n}}{\X^{\ot n}},
\end{eqnarray*}
for $\X=\sum^N_{i=1}\bsX_i\ot\ket{i}$, and thus $\X^{\ot
n}\simeq\sum^N_{j_1,\ldots,
j_n=1}\bsX_{j_1}\ot\cdots\ot\bsX_{j_n}\ot\ket{j_1\ldots j_n}$. We
get that
\begin{eqnarray*}
p_\pi(\X) =
\Tr{\bsP_{d,n}(\pi)\Pa{\bsX_{i_1}\ot\cdots\ot\bsX_{i_n}}},\quad
i_1,\ldots,i_n\in\set{1,\ldots,N}.
\end{eqnarray*}
Using the decomposition of $\pi$ into disjoint cycles, we get that
$p_\pi(\X)$ can be expressed as a product of Bargmann invariants of
degree at most $n$.

Since an \emph{arbitrary} real polynomial invariant $p\in \bbR[V]^K$
for $K=\sfU(d)$ can be expressed via a suitable linear combination
of $p_\pi$'s, we can conclude that Bargmann invariants determine the
joint unitary similarity of two tuples of Hermitian operators
$\X=(\bsX_1,\ldots,\bsX_N)$ and $\Y=(\bsY_1,\ldots,\bsY_N)$ acting
on $\bbC^d$. Restricting to quantum states, we naturally get the
desired result.
\end{proof}

\section{Bargmann invariants}\label{sect:3}

For an $n$-tuple of mixed states $\Psi=(\rho_1,\ldots,\rho_n)$,
where $\rho_k\in\density{\bbC^d}$ for $k\in\set{1,\ldots,n}$, we
have the following definition of Bargmann invariant:
\begin{definition}[Bargmann invariant]\label{def:nthBarg}
The so-called $n$th-order \emph{Bargmann invariant} for such
$n$-tuple $\Psi=(\rho_1,\ldots,\rho_n)$ is defined as
\begin{eqnarray}\label{eq:nthBarg}
\Delta_n(\Psi):=\Tr{\rho_1\cdots\rho_n}.
\end{eqnarray}
In order to be convenience, the set of all $n$th Bargmann invariants
will be denoted by $\cB^\circ_n(d)/\cB^\bullet_n(d)$ if the tuple
$\Psi$ consists of all pure/mixed states.
\end{definition}
Later, we will see that $\cB^\circ_n(d)=\cB^\bullet_n(d)$, which is
denoted by $\cB_n(d)$. Denote
\begin{eqnarray}\label{eq:Bn}
\cB^\circ_n:=\bigcup^{+\infty}_{d=2}\cB^\circ_n(d),\quad
\cB^\bullet_n:=\bigcup^{+\infty}_{d=2}\cB^\bullet_n(d),\text{ and
}\cB_n:=\bigcup^{+\infty}_{d=2}\cB_n(d).
\end{eqnarray}
The properties of $\cB^\circ_n$ are as follows:
\begin{enumerate}[(a)]
    \item $\cB^\circ_n$ is a \emph{closed} and \emph{connected} set because it is the continuous image of Cartesian product of pure state spaces.
    \item $\cB^\circ_n$ is symmetric with respect to the real axis, as
    $\Tr{\psi_n \cdots \psi_2\psi_1} = \overline{\Tr{\psi_1\psi_2 \cdots \psi_n}}$.
    \item $\cB^\circ_n \subset \Set{z \in \complex : \abs{z} \leqslant 1}$ due to the fact that $\abs{\Tr{\rho_1\rho_2\cdots\rho_n}}\leqslant1$.
\end{enumerate}
We should also remark here that, when each member of the tuple
$\Psi$ is pure state, we write $\psi_k=\proj{\psi_k}$ instead of
$\rho_k$. Now
\begin{eqnarray}
\Delta_n(\Psi)=\Tr{\psi_1\cdots\psi_n} =
\Inner{\psi_1}{\psi_2}\Inner{\psi_2}{\psi_3}\cdots
\Inner{\psi_{n-1}}{\psi_n}\Inner{\psi_n}{\psi_1},
\end{eqnarray}
which can be viewed as
\begin{eqnarray}
\Tr{\psi_1\cdots\psi_n} =\prod^d_{k=1}[G(\Psi)]_{k,k\oplus 1},
\end{eqnarray}
where $\oplus$ denotes the addition modulo $n$ throughout the whole
paper and $G(\Psi)$ is the Gram matrix determined by a some tuple of
wave functions $\ket{\psi_k}$'s, i.e.,
\begin{eqnarray}
G(\Psi) =\Pa{\begin{array}{ccccc}
                 \Inner{\psi_1}{\psi_1} & \Inner{\psi_1}{\psi_2} & \Inner{\psi_1}{\psi_3} & \cdots & \Inner{\psi_1}{\psi_n} \\
                 \Inner{\psi_2}{\psi_1} & \Inner{\psi_2}{\psi_2} & \Inner{\psi_2}{\psi_3} & \cdots & \Inner{\psi_2}{\psi_n} \\
                 \vdots & \vdots & \ddots & \ddots &\vdots \\
                 \Inner{\psi_{n-1}}{\psi_1} & \Inner{\psi_{n-1}}{\psi_2} & \Inner{\psi_{n-1}}{\psi_3} &\cdots &
                 \Inner{\psi_{n-1}}{\psi_n}\\
                 \Inner{\psi_n}{\psi_1} & \Inner{\psi_n}{\psi_2} & \Inner{\psi_n}{\psi_3} &\cdots & \Inner{\psi_n}{\psi_n}
                 \end{array}
}.
\end{eqnarray}
Denote by $\cG_n(d)$ the set of all $n\times n$ Gram matrices formed
by $n$-tuples of wave-functions (i.e., complex unit vectors) acting
on $\bbC^d$. Define
\begin{eqnarray}\label{eq:Gn}
\cG_n:=\bigcup^{+\infty}_{d=2}\cG_n(d).
\end{eqnarray}
We should be careful that all $n$th-order Bargmann invariants
$\Delta_n(\Psi)$ are gauge-invariant for tuple $\Psi$, but Gram
matrices $G(\Psi)$ are not.

We demonstrate that $\cB^\circ_n(d)$ is independent of the
underlying space dimension $d$. This follows because each element of
$\cB^\circ_n(d)$ factorizes into a product of inner products, and
inner products themselves are dimension-independent. Although this
claim appears in \cite{Zhang2025PRA1}, we provide its first
quantitative verification.
\begin{prop}\label{prop:pdfinnprod}
For two independent Haar-distributed random unit vectors $\bsu,\bsv$
in $\bbC^d$, if $z=\Inner{\bsu}{\bsv}=x+\mathrm{i}y$, where
$x,y\in\bbR$, then the joint probability density function (PDF) of
$(x,y)$ is given by
\begin{eqnarray}
f_d(z) = \tfrac1\pi (d-1)\Pa{1-\abs{z}^2}^{d-2}\chi_D(z),
\end{eqnarray}
where $D:=\Set{z\in\bbC:\abs{z}\leqslant1}$ and $\chi_D(z)$ is the
indicator function of the set $D$, i.e., $\chi_D(z)=1$ if $z\in D$,
and $\chi_D(z)=0$ if $z\notin D$. Thus both $x$ and $y$ has the same
marginal PDF $p_d$, which is given by
\begin{eqnarray}
p_d(t) = \tfrac{\Gamma(d)}{\sqrt{\pi}\Gamma
\Pa{d-\frac{1}{2}}}\Pa{1-t^2}^{d-\frac32}\chi_{[-1,1]}(t).
\end{eqnarray}
\end{prop}

\begin{proof}
Recall that for two Haar-distributed random unit vectors $\bsu,\bsv$
in $\bbC^d$, their inner product has a polar form
$\Inner{\bsu}{\bsv}=re^{\mathrm{i}\theta}=r\cos\theta+\mathrm{i}r\sin\theta$,
where the PDF of $r=\abs{\Inner{\bsu}{\bsv}}$ is given by
\cite{Hiai2000}
\begin{eqnarray}\label{eq:Pn}
P_d(r) = 2(d-1)r(1-r^2)^{d-2}\chi_{[0,1]}(r).
\end{eqnarray}
Next, we want to derive the joint PDF of both the real part and
imaginary part of $\Inner{\bsu}{\bsv}=x+\mathrm{i}y$.
\begin{enumerate}[(i)]
\item Firstly, we derive the joint PDF $f(z)\equiv f(x,y)$. Note that, from the
defining expression,
\begin{eqnarray}
f_d(x,y)= \frac1{2\pi}\int^1_0\dif
r\int^{2\pi}_0\dif\theta\delta\Pa{x-r\cos\theta}\delta\Pa{y-r\sin\theta}P_d(r),
\end{eqnarray}
where $P_d(r)$ is taken from Eq.~\eqref{eq:Pn}, using the integral
representation of Dirac delta function,
\begin{eqnarray}
\delta(h)=\frac1{2\pi}\int_\real e^{\mathrm{i}h\nu}\dif \nu,
\end{eqnarray}
we see that
\begin{eqnarray*}
f_d(x,y) &=& \frac1{(2\pi)^3}\int_{\real^2}\dif\alpha\dif\beta
e^{\mathrm{i}( x\alpha+ y\beta)}\int^1_0\dif
rP_d(r)\int^{2\pi}_0\dif\theta e^{-\mathrm{i}r(\alpha
\cos\theta+\beta\sin\theta)}\\
&=&\frac1{(2\pi)^2}\int_{\real^2}\dif\alpha\dif\beta e^{\mathrm{i}(
x\alpha+ y\beta)}\int^1_0\dif r
P_d(r)J_0\Pa{r\pa{\alpha^2+\beta^2}^{\frac12}}\\
&=& \frac1{(2\pi)^2}\int^1_0\dif r
P_d(r)\int_{\real^2}\dif\alpha\dif\beta e^{\mathrm{i}( x\alpha+
y\beta)} J_0\Pa{r\pa{\alpha^2+\beta^2}^{\frac12}}\\
&=&\frac1{(2\pi)^2}\int^1_0\dif r P_d(r) \int^\infty_0\dif R R
J_0(rR)\int^{2\pi}_0e^{\mathrm{i}R(x\cos\varphi+
y\sin\varphi)}\dif\varphi\\
&=&\frac1{2\pi}\int^1_0\dif r P_d(r) \int^\infty_0\dif RR
J_0(rR)J_0(R(x^2+y^2))\\
&=&\frac1{2\pi}\int^\infty_0\dif RR
J_0(R(x^2+y^2))\int^1_0\dif r P_d(r)J_0(rR)\\
&=& \frac1{2\pi}\int^\infty_0\dif R RJ_0\Pa{R(x^2+y^2)}\,
_0F_1\Pa{;d;-\tfrac{R^2}{4}}\\
&=& \frac1\pi (d-1)(1-x^2-y^2)^{d-2}\chi_D(x,y).
\end{eqnarray*}
Here $\, _0F_1$ is the so-called \emph{confluent hypergeometric
function}, defined by
\begin{eqnarray}
\, _0F_1(;a;z) :=\sum^\infty_{k=0}\frac{z^k}{k!a(a+1)\cdots
(a+k-1)}.
\end{eqnarray}
\item Second, we obtain the marginal density $p_d(x)$ by integrating the joint PDF $f_d(x,y)$ over $y$.
Indeed, $x^2+y^2\leqslant1$ is equivalent to $x\in[-1,1],
y\in\Br{-\sqrt{1-x^2},\sqrt{1-x^2}}$, thus
\begin{eqnarray*}
p_d(x)&=&\int^{\sqrt{1-x^2}}_{-\sqrt{1-x^2}} f_d(x,y)\dif y\quad
(x\in[-1,1])\\
&=&\frac{\Gamma(d)}{\sqrt{\pi}\Gamma
\Pa{d-\frac{1}{2}}}\Pa{1-x^2}^{d-\frac32}\chi_{[-1,1]}(x).
\end{eqnarray*}
Similarly, the marginal PDF for $y$, denoted by $q_d(y)$, is
obtained by integrating $f_d(x,y)$ over $x$, mirroring the procedure
used for $p_d(x)$. Finally, we find that $p_d(t)\equiv q_d(t)$ for
$t\in[-1,1]$.
\end{enumerate}
This completes the proof.
\end{proof}
From Proposition~\ref{prop:pdfinnprod}, we see that the support of
PDF $f_d(z)$ is $\supp(f_d)=D$, the unit disk in $\bbC$. Apparently,
$\supp(f_d)=D$ is independent of the underlying space dimension $d$.

An interesting problem can be posed here: For $n$ independent
Haar-distributed random unit vectors $\ket{\psi_k}$'s in $\bbC^d$,
the joint PDF of the random Bargmann invariant
$z=\Tr{\psi_1\cdots\psi_n}$ can be investigated via the Dirac delta
function \cite{Zhang2021IJTP}:
\begin{eqnarray}
\phi_{n,d}(z):=
\int\delta(z-\Tr{\psi_1\cdots\psi_n})\prod^n_{k=1}\dif\mu_{\haar}(\psi_k)
\end{eqnarray}
where $\mu_{\haar}$ is the normalized Haar measure. The PDF
$\phi_{n,d}(z)$ satisfies $\supp(\phi_{n,d})=\cB^\circ_n(d)$. An
analytical expression for $\phi_{n,d}(z)$ would provide more
information than its support alone. This is beyond the scope of our
current discussion.

Beyond the notion of Gram matrix, the circulant matrix is also very
important notion which will be used in characterizing the boundary
curve of $\cB_n$. Let us explain it in more detail.

\subsection{Circulant matrices and its properties}

Denote by $S_n$ the set of all permutations of $n$ distinct elements
$\set{0,1,\ldots,n-1}$. For any permutation $\pi\in S_n$, define the
matrix representation of $\pi$ as
\begin{eqnarray}
\bsP_\pi : = \sum^{n-1}_{k=0}\out{\pi(k)}{k},
\end{eqnarray}
where $\set{\ket{k}:k=0,1,\ldots,n-1}$ is the computational basis of
$\bbC^n$. Apparently, $\bsP_\pi$'s are just the usual permutation
matrices. By conventions, the zeroth power of $\bsP_{\pi}$ means
that $\bsP^0_{\pi}\equiv\I_n$.
\begin{definition}[Circulant matrix]
Fix $\pi_0=(n-1,n-2,\ldots,2,1,0)\in S_n$. The so-called
\emph{circulant matrix} determined by
$\bsz=(z_0,z_1,\ldots,z_{n-1})\in \bbC^n$ is just the following one:
\begin{eqnarray}
\bsC(\bsz):=\sum^{n-1}_{k=0}z_k\bsP^k_{\pi_0}.
\end{eqnarray}
\end{definition}
Denote by $\cC_n$ the set of all $n\times n$ complex circulant
matrices. In the above definition, we can write $\bsP_{\pi_0}$
explicitly as
\begin{eqnarray}
\bsP_{\pi_0} = \Pa{\begin{array}{ccccc}
                     0 & 1 & 0 & \cdots & 0 \\
                     0 & 0 & 1 & \cdots & 0 \\
                     \vdots & \vdots & \ddots & \ddots & \vdots \\
                     0 & 0 & 0 & \ddots & 1 \\
                     1 & 0 & 0 & \cdots & 0
                   \end{array}
}.
\end{eqnarray}
Its characteristic polynomial is given by
$p_0(\lambda)=\det(\lambda\I_n-\bsP_{\pi_0})=\lambda^n-1$ whose
roots are $\set{\omega^k_n:k=0,1,\ldots,n-1}$, where
$\omega_n=\exp\Pa{\frac{2\pi\mathrm{i}}n}$.

\begin{prop}
The permutation matrix $\bsP_{\pi_0}$ can be diagonalized as
\begin{eqnarray}
\bsP_{\pi_0} = \bsF\Omega\bsF^\dagger =
\sum^{n-1}_{k=0}\omega^k_n\proj{\bsf_k},
\end{eqnarray}
where $\Omega:=\diag(\omega^0_n,\omega^1_n,\ldots,\omega^{n-1}_n)$
and $\bsF:=(\bsf_0,\bsf_1,\ldots,\bsf_{n-1})$ for
\begin{eqnarray}\label{eq:DF}
\ket{\bsf_k}:=\frac1{\sqrt{n}}\Pa{\omega^0_n,\omega^k_n,\omega^{2k}_n,\ldots,\omega^{(n-1)k}_n}^\t,
\end{eqnarray}
where $k\in\set{0,1,\ldots,n-1}$. Moreover, the circulant matrix
$\bsC(\bsz)$ for $\bsz=(z_0,z_1,\ldots,z_{n-1})^\t\in\bbC^n$ can be
diagonalized as
\begin{eqnarray}\label{eq:eigens}
\bsC(\bsz) = \bsF\Lambda_{\bsz}\bsF^\dagger,
\end{eqnarray}
where
$\Lambda_{\bsz}:=\sum^{n-1}_{k=0}z_k\Omega^k=\diag(\lambda_0(\bsz),\lambda_1(\bsz),\ldots,\lambda_{n-1}(\bsz))$
for $\lambda_j(\bsz) = \sum^{n-1}_{k=0}z_k\omega^{jk}_n$.
\end{prop}

\begin{proof}
(i) Since the matrix $\bsP_{\pi_0}$ is the permutation matrix
corresponding to an $n$-cycle. It is unitary, and its eigenvalues
are the $n$th roots of unity: $\lambda_k=\omega^k_n$ for
$k\in\set{0,1,\ldots,n-1}$. The corresponding normalized
eigenvectors are
\begin{eqnarray*}
\ket{\bsf_k}:=\frac1{\sqrt{n}}\Pa{\omega^0_n,\omega^k_n,\omega^{2k}_n,\ldots,\omega^{(n-1)k}_n}^\t,
\end{eqnarray*}
These eigenvectors form an orthonormal basis of $\bbC^n$. The
spectral decomposition of $\bsP_{\pi_0}$ is given by
$\bsP_{\pi_0}=\sum^{n-1}_{k=0}\omega^k_n\proj{\bsf_k}$. In matrix
form, $\bsP_{\pi_0} = \bsF\Omega\bsF^\dagger$, where $\bsF$ has
columns $\bsf_k$ and
$\Omega=\diag(\lambda_0,\lambda_1,\ldots,\lambda_{n-1})$.

(ii) By the definition of the circulant matrix, we get that
$\bsC(\bsz)=\sum^{n-1}_{k=0}z_k\bsP^k_{\pi_0}$. Because
$\bsP_{\pi_0}=\bsF\Omega\bsF^\dagger$, we see that
$\bsP^k_{\pi_0}=\bsF\Omega^k\bsF^\dagger$. This implies that
\begin{eqnarray*}
\bsC(\bsz)=\sum^{n-1}_{k=0}z_k\bsP^k_{\pi_0} =
\bsF\Pa{\sum^{n-1}_{k=0}z_k\Omega^k}\bsF^\dagger=\bsF\Lambda_{\bsz}\bsF^\dagger,
\end{eqnarray*}
where
\begin{eqnarray*}
\Lambda_{\bsz}:=\sum^{n-1}_{k=0}z_k\Omega^k =
\diag(\lambda_0(\bsz),\lambda_1(\bsz),\ldots,\lambda_{n-1}(\bsz))
\end{eqnarray*}
for $\lambda_j(\bsz) = \sum^{n-1}_{k=0}z_k\omega^{jk}_n$.
\end{proof}

\begin{prop}\label{prop:GnCn}
For each $\bsC(\bsz)\in\cC_n$ for
$\bsz=(z_0,z_1,\ldots,z_{n-1})\in\bbC^n$, it holds that
\begin{enumerate}[(i)]
\item $\bsC(\bsz)\in\Herm(\bbC^n)$, the set of all Hermitian matrices acting on $\bbC^n$, if and only if $z_0\in\bbR$ and
$\bar z_k=z_{n-k}$ for $k\in\set{1,\ldots,n-1}$.
\item $\bsC(\bsz)\in\pos{\bbC^n}$, the set of all positive semi-definite matrices in $\Herm(\bbC^n)$, if and only if $\bsF\bsz\in\bbR^n_{\geqslant0}$, where
$\bsz=(z_0,z_1,\ldots,z_{n-1})^\t\in\bbC^n$ satisfying
$z_0\in\bbR_{\geqslant0}$ and $\bar z_k=z_{n-k}$ for
$k\in\set{1,\ldots,n-1}$.
\item $\bsC(\bsz)\in\cG_n$ if and only if $\bsF\bsz\in\bbR^n_{\geqslant0}$, where
$\bsz=(z_0,z_1,\ldots,z_{n-1})^\t\in\bbC^n$ satisfying $z_0=1$ and
$\bar z_k=z_{n-k}$ for $k\in\set{1,\ldots,n-1}$.
\end{enumerate}
\end{prop}

\begin{proof}
\begin{enumerate}[(i)]
\item It is trivially.
\item From Eq.~\eqref{eq:eigens}, we can read all eigenvalues of
$\bsC(\bsz)$ as follows:
\begin{eqnarray}
\Pa{\begin{array}{c}
      \lambda_0(\bsz) \\
      \lambda_1(\bsz) \\
      \lambda_2(\bsz) \\
      \vdots \\
      \lambda_{n-1}(\bsz)
    \end{array}
}=\Pa{\begin{array}{ccccc}
        1 & 1 & 1 & \cdots & 1 \\
        1 & \omega_n & \omega^2_n & \cdots & \omega^{n-1}_n \\
        1 & \omega^2_n & \omega^4_n & \cdots & \omega^{2(n-1)}_n \\
        \vdots & \vdots & \vdots & \ddots & \vdots \\
        1 & \omega^{n-1}_n & \omega^{2(n-1)}_n & \cdots & \omega^{(n-1)^2}_n
      \end{array}
}\Pa{\begin{array}{c}
      z_0 \\
      z_1 \\
      z_2 \\
      \vdots \\
      z_{n-1}
    \end{array}
},
\end{eqnarray}
which can be rewritten simply as $\bslambda(\bsz)=\sqrt{n}\bsF\bsz$.
Then $\bsC(\bsz)\in\pos{\bbC^n}$ if and only if all eigenvalues
$\lambda_j(\bsz)\in\bbR_{\geqslant0}$ for each
$j\in\set{0,1,\ldots,n-1}$, which is equivalently to
$\bslambda(\bsz)\in\bbR^n_{\geqslant0}$. That is,
$\bsF\bsz\in\bbR^n_{\geqslant0}$.
\item The proof is trivially.
\end{enumerate}
This completes the proof.
\end{proof}

In what follows, we introduce a subclass of mixed-permutation
channels, studied in \cite{Zhang2024epjp}. Such channels are called
the circulant quantum channels \cite{Xie2026}, which is intimately
to the above circulant matrices.

\subsection{Circulant quantum channels and its properties}

\begin{definition}[Circulant quantum channel]
Fix $\pi_0=(n-1,n-2,\ldots,2,1,0)\in S_n$. The so-called
\emph{circulant quantum channel} is defined as
\begin{eqnarray}
\Phi(\bsX):=\frac1n\sum^{n-1}_{k=0}\bsP^k_{\pi_0}\bsX\bsP^{-k}_{\pi_0}.
\end{eqnarray}
\end{definition}
An interesting property of this quantum channel is that its fixed
points are precisely the circulant matrices \cite{Xie2026}. This
arises from the covariance of the channel under the cyclic shift
generated by the matrix $\bsP_{\pi_0}$, implying that any fixed
point must commute with $\bsP_{\pi_0}$--- a defining characteristic
of circulant matrices.

\begin{prop}[\cite{Xie2026}]
The circulant quantum channel $\Phi$ can be reformulated into two
forms:
\begin{eqnarray}
\Phi(\bsX)&=&\sum^{n-1}_{k=0}\proj{\bsf_k}\bsX\proj{\bsf_k}\\
&=&
\frac1n\sum^{n-1}_{k=0}\Inner{\bsP^k_{\pi_0}}{\bsX}\bsP^k_{\pi_0}.
\end{eqnarray}
Moreover, the set of fixed points of $\Phi$ is precisely the set of
circulant matrices. In addition, $\Phi$ is an entanglement-breaking
channel\footnote{A channel is entanglement-breaking if and only if
its Choi presentation is separable.}.
\end{prop}

\begin{proof}
By spectral decomposition of $\bsP_{\pi_0}$:
$\bsP_{\pi_0}=\sum^{n-1}_{k=0}\omega^k_n\proj{\bsf_k}$. It follows
that $\bsP_{\pi_0}\ket{\bsf_i}=\omega^i_n\ket{\bsf_i}$ and
$\bsP^{-k}_{\pi_0}\ket{\bsf_i}=\omega^{-ki}_n\ket{\bsf_i}$. Then
\begin{eqnarray*}
\Innerm{\bsf_i}{\Phi(\bsX)}{\bsf_j}
&=&\frac1n\sum^{n-1}_{k=0}\Innerm{\bsf_i}{\bsP^k_{\pi_0}\bsX\bsP^{-k}_{\pi_0}}{\bsf_j}
=\frac1n\sum^{n-1}_{k=0}\Innerm{\bsP^{-k}_{\pi_0}\bsf_i}{\bsX}{\bsP^{-k}_{\pi_0}\bsf_j}\\
&=&\frac1n\sum^{n-1}_{k=0}\overline{\omega^{-ki}_n}\omega^{-kj}_n\Innerm{\bsf_i}{\bsX}{\bsf_j}
=
\Pa{\frac1n\sum^{n-1}_{k=0}\omega^{k(i-j)}_n}\Innerm{\bsf_i}{\bsX}{\bsf_j}\\
&=&\delta_{ij}\Innerm{\bsf_i}{\bsX}{\bsf_j}.
\end{eqnarray*}
Furthermore,
\begin{eqnarray*}
\Phi(\bsX)
&=&\sum^{n-1}_{i,j=0}\proj{\bsf_i}\Phi(\bsX)\proj{\bsf_j}=\sum^{n-1}_{i,j=0}
\Innerm{\bsf_i}{\Phi(\bsX)}{\bsf_j}\out{\bsf_i}{\bsf_j}\\
&=&\sum^{n-1}_{i,j=0}\delta_{ij}\Innerm{\bsf_i}{\bsX}{\bsf_j}\out{\bsf_i}{\bsf_j}
=\sum^{n-1}_{i=0}\Innerm{\bsf_i}{\bsX}{\bsf_i}\out{\bsf_i}{\bsf_i}\\
&=&\sum^{n-1}_{k=0}\proj{\bsf_k}\bsX\proj{\bsf_k}.
\end{eqnarray*}
Based on this expression, $\Phi(\bsX)$ is a circulant matrix for
each $\bsX$ acting on $\bbC^n$, which implies that the set of fixed
points of $\Phi$ is precisely the set of circulant matrices.
Finally, Choi representation of $\Phi$ can be calculated as
\begin{eqnarray*}
J(\Phi)&=&\frac1n\sum^{n-1}_{k=0}\vec(\proj{\bsf_k})\vec(\proj{\bsf_k})^\dagger\\
&=&\frac1n\sum^{n-1}_{k=0}\proj{\bsf_k}\ot\overline{\proj{\bsf_k}},
\end{eqnarray*}
which is separable. That is, $\Phi$ is an entanglement-breaking
channel.
\end{proof}

With this result, we can easily realize circulant quantum channel by
chosen von Neumann measurement $\set{\proj{\bsf_j}}^{n-1}_{j=0}$
along the Fourier basis $\set{\ket{\bsf_i}:i=0,1,\ldots,n-1}$ from
Eq.~\eqref{eq:DF}.

\section{Circulant Gram matrices}\label{sect:4}

Now we introduce an important subset of $\cB_n$ as follows:
\begin{eqnarray}\label{eq:circBarg}
\cB_n|_{\mathrm{circ}}:=\Set{\prod^n_{k=1}h_{k,k\oplus1}\Big|(h_{ij})\in
\cG_n\cap\cC_n}.
\end{eqnarray}
Apparently $\cB_n|_{\mathrm{circ}}\subseteq \cB_n$. In what follows,
we characterize the set $\cB_n|_{\mathrm{circ}}$.

\subsection{Characterization of the set $\cB_n|_{\mathrm{circ}}$}

From the definition of $\cB_n|_{\mathrm{circ}}$ in
Eq.~\eqref{eq:circBarg}, we see that
\begin{eqnarray}
\cB_n|_{\mathrm{circ}}=\Set{z^n_1\mid
\bsC(\bsz)=\sum^{n-1}_{k=0}z_k\bsP^k_{\pi_0}\in\cG_n}=\set{z^n_1\mid
z_1\in \cQ_n},
\end{eqnarray}
where
\begin{eqnarray}
\cQ_n:=\Set{z_1\mid \bsC(\bsz)\in\cG_n}.
\end{eqnarray}
To characterize the set $\cB_n|_{\mathrm{circ}}$, we begin by
characterizing the set $\cQ_n$. We will prove that $\cQ_n$ is
geometrically identified with a particular regular $n$-gon inscribed
in the unit circle, which we denote by $\cP_n$ and define as the
polygon whose vertices are the $n$th roots of unity.

\begin{thrm}[\cite{Li2025PRA}]
For each integer $n\geqslant3$, denote
$\omega_n=\exp\Pa{\frac{2\pi\mathrm{i}}n}$, it holds that
\begin{eqnarray}\label{eq:n-gon}
\cQ_n=\cP_n.
\end{eqnarray}
Based on this identification, we get that
\begin{eqnarray}
\cB_n|_{\mathrm{circ}} &=& \Set{z^n\mid z\in \cP_n},\\
\partial \cB_n|_{\mathrm{circ}}
&=&\Set{(t+(1-t)\omega_n)^n\mid t\in[0,1]}.\label{eq:bcurve}
\end{eqnarray}
\end{thrm}

\begin{proof}
The proof of the equality $\cQ_n=\cP_n$ is completed in four steps:
\begin{enumerate}[(1)]
\item $1\in \cQ_n$.
\item $\omega_n z_1\in \cQ_n$ whenever $z_1\in \cQ_n$.
\item the set $\cQ_n$ is convex.
\item For each $z_1=x+\mathrm{i}y\in \cQ_n$, identified with
$(x,y)$, in Cartesian coordinates, the point satisfies the
inequality defining the closed half-plane below the straight line
$L_n$ through two points $(1,0)$ and
$(\cos\frac{2\pi}n,\sin\frac{2\pi}n)$, corresponding to $1$ and
$\omega_n$, respectively. The equation of $L_n$ is given by
\begin{eqnarray}
\tfrac{y-0}{x-1} = \tfrac{0-\sin\frac{2\pi}n}{1-\cos\frac{2\pi}n} =
-\cot\tfrac\pi n\Longleftrightarrow y=(\cot\tfrac\pi n)(1-x).
\end{eqnarray}
Note that the side of $\cP_n$ connecting $1$ and $\omega_n$ lies on
the line $L_n$.
\end{enumerate}
Now we can use the four items above to show $\cQ_n=\cP_n$.
\begin{itemize}
\item $\cP_n\subseteq\cQ_n$. Apparently
$\cP_n=\mathrm{ConvHull}\Set{1,\omega_n,\ldots,\omega^{n-1}_n}$.
From the above items (1) and (2), we can get that
$\Set{1,\omega_n,\ldots,\omega^{n-1}_n}\subset\cQ_n$. By the item
(3), i.e., the convexity of $\cQ_n$, we get that
\begin{eqnarray*}
\cP_n=\mathrm{ConvHull}\Set{1,\omega_n,\ldots,\omega^{n-1}_n}\subseteq
\mathrm{ConvHull}\set{\cQ_n}=\cQ_n.
\end{eqnarray*}
\item $\cQ_n\subseteq\cP_n$. It is easily seen that
$\cQ_n=\omega^k_n\cQ_n$ for $k\in\set{0,1,\ldots,n-1}$. It suffices
to show it for $k=1$. Indeed, by the item (2), for any
$z_1\in\cQ_n$, we have $\omega_nz_1\in\cQ_n$, thus
$\omega^{n-1}_nz_1\in\cQ_n$, which amounts to say that
$\omega^{-1}_nz_1\in\cQ_n$ because $\omega^n_n=1$. That is,
$z_1\in\omega_n\cQ_n$. Thus $\cQ_n\subseteq \omega_n\cQ_n$.
Furthermore,
\begin{eqnarray*}
\cQ_n\subseteq \omega_n\cQ_n\subseteq
\omega^2_n\cQ_n\subseteq\cdots\subseteq\omega^{n-1}_n\cQ_n\subseteq
\omega^n_n\cQ_n=\cQ_n.
\end{eqnarray*}
We know that the line $L_n$ partitions the plane into an upper and a
lower half-plane, defined by the points lying above and below $L_n$,
respectively. Define $L^+_n$ and $L^-_n$ to be the upper and lower
half-planes (relative to $L_n$). In fact,
\begin{eqnarray*}
L^-_n=\Set{x+\mathrm{i}y\in\bbC\mid (x,y)\in\bbR^2,x\cos\tfrac\pi
n+y\sin\tfrac\pi n\leqslant \cos\tfrac\pi n}.
\end{eqnarray*}
By the item (4), $\cQ_n\subset L^-_n$. Then $\omega^k_n\cQ_n\subset
\omega^k_nL^-_n$ for $k\in\set{0,1,\ldots,n-1}$, and thus
\begin{eqnarray*}
\cQ_n=\bigcap^{n-1}_{k=0}\omega^k_n\cQ_n\subset
\bigcap^{n-1}_{k=0}\omega^k_nL^-_n=\cP_n.
\end{eqnarray*}
\end{itemize}
In summary, $\cQ_n=\cP_n$. Let us proceed to prove all items
mentioned above.
\begin{itemize}
\item \textbf{Step 1: $1\in \cQ_n$.}

Consider a specific matrix $\bsC(\bsz_0)$, where
$\bsz_0=(1,1,\ldots,1)\in\bbC^n$. Apparently
$\bsC(\bsz_0)\in\cG_n\cap\cC_n$. This implies that $1\in \cQ_n$.
\item \textbf{Step 2: $\omega_n z_1\in \cQ_n$ whenever $z_1\in \cQ_n$.}

Note that $\bsC(\bsz)\in\cG_n$ if and only if
$\bsF\bsz\in\bbR^n_{\geqslant0}$, where
$\bsz=(z_0,z_1,\ldots,z_{n-1})^\t\in\bbC^n$ satisfying $z_0=1$ and
$\bar z_k=z_{n-k}$ for $k\in\set{1,\ldots,n-1}$ by
Proposition~\ref{prop:GnCn}:
\begin{eqnarray*}
\cQ_n = \Set{z_1\in\bbC\mid z_0=1,\bar z_k=z_{n-k}\text{ for
}k\in\set{1,\ldots,n-1}, \bsF\bsz\in\bbR^n_{\geqslant0}}.
\end{eqnarray*}
Now if $z_1\in \cQ_n$, then there exists
$\bsz=(z_0,z_1,\ldots,z_{n-1})^\t\in \bbC^n$ satisfying that $z_0=1$
and $\bar z_k=z_{n-k}$ for $k\in\set{1,\ldots,n-1}$ and
$\bsF\bsz\in\bbR^n_{\geqslant0}$. Define
\begin{eqnarray}
\bsz_\omega:=\Pa{z_0\omega^0_n,z_1\omega^1_n,\ldots,z_{n-1}\omega^{n-1}_n}^\t=\Omega\bsz.
\end{eqnarray}
Apparently $z_0\omega^0_n=1$ and
$\overline{z_k\omega^k_n}=z_{n-k}\omega^{n-k}_n$ for
$k\in\set{1,\ldots,n-1}$ due to the fact that $z_0=1$ and $\bar
z_k=z_{n-k}$ for $k\in\set{1,\ldots,n-1}$. The proof of
$\omega_nz_1\in \cQ_n$ can be completed whenever we show that
$\bsF\bsz_\omega\in\bbR^n_{\geqslant0}$. In fact,
\begin{eqnarray}
\bsF\bsz_\omega = \bsF\Pa{\Omega\bsz}=\Pa{\bsF\Omega\bsF^\dagger}
(\bsF\bsz)=\bsP_{\pi_0}(\bsF\bsz),
\end{eqnarray}
which components are the permutated ones of
$\bsF\bsz\in\bbR^n_{\geqslant0}$. Therefore $\bsF\bsz_\omega
\in\bbR^n_{\geqslant0}$.
\item \textbf{Step 3: The set $\cQ_n$ is convex.}

For any elements $u_1,v_1\in \cQ_n$ and $s\in[0,1]$, there exists
$\bsu=(u_0,u_1,\ldots,u_{n-1}),\bsv=(v_0,v_1,\ldots,v_{n-1})$ in
$\bbC^n$ satisfying $u_0=1,\bar u_k=u_{n-k}$ for
$k\in\set{1,\ldots,n-1}$ and $\bsF\bsu\in \bbR^n_{\geqslant0}$; and
$v_0=1,\bar v_k=v_{n-k}$ for $k\in\set{1,\ldots,n-1}$ and
$\bsF\bsv\in \bbR^n_{\geqslant0}$, respectively. Consider the convex
combination of $\bsu$ and $\bsv$ with weight $s\in[0,1]$:
\begin{eqnarray*}
\bsw:=s\bsu+(1-s)\bsv=\Pa{su_0+(1-s)v_0,su_1+(1-s)v_1,\ldots,su_{n-1}+(1-s)v_{n-1}}^\t.
\end{eqnarray*}
Thus $w_0=su_0+(1-s)v_0=1$ and
\begin{eqnarray*}
\bar w_k &=& \overline{su_k+(1-s)v_k} = s\bar u_k+(1-s)\bar v_k\\
&=& s u_{n-k} + (1-s)v_{n-k} = w_{n-k}.
\end{eqnarray*}
Moreover,
\begin{eqnarray*}
\bsF\bsw =s\bsF\bsu + (1-s)\bsF\bsv\in\bbR^n_{\geqslant0}\text{ for
}\bsF\bsu\in\bbR^n_{\geqslant0}\text{ and
}\bsF\bsv\in\bbR^n_{\geqslant0}.
\end{eqnarray*}
Therefore $su_1+(1-s)v_1=w_1\in \cQ_n$. That is, the set $\cQ_n$ is
convex.
\item \textbf{Step 4: For each $z_1=x+\mathrm{i}y\in \cQ_n$, the point $(x,y)$ must satisfy the inequality: $y\leqslant(\cot\frac\pi n)(1-x)$ or $x\cos\frac\pi n + y\sin\frac\pi n \leqslant\cos\frac\pi
n$.}

For any $z_1=x+\mathrm{i}y\in \cQ_n$, there exists
$\bsz=(z_0,z_1,\ldots,z_{n-1})^\t\in\bbC^n$ such
$\bsF\bsz\in\bbR^n_{\geqslant0}$ satisfying $z_0=1$ and $\bar
z_k=z_{n-k}$ for $k\in\set{1,\ldots,n-1}$. Let
$\bsb=(b_0,b_1,\ldots,b_{n-1})^\t$, where $b_0=2\cos\frac\pi n$ and
$b_1=\bar b_{n-1}=-\cos\frac\pi n+\mathrm{i}\sin\frac\pi n$, and
$b_k=0$ for $k\in\set{2,\ldots, n-2}$. Note that $\bar z_1=z_{n-1}$
\begin{eqnarray*}
\bsb^\t\bsz &=& b_0z_0+b_1z_1+\cdots+b_{n-1}z_{n-1} =
b_0z_0+b_1z_1+b_{n-1}z_{n-1}\\
&=& b_0+b_1z_1+ \bar b_1\bar z_1 = b_0+2\re[b_1z_1]\\
&=& 2\cos\tfrac\pi n  + 2\re\Br{\Pa{-\cos\tfrac\pi n
+\mathrm{i}\sin\tfrac\pi n}z_1} \\
&=&2\Br{\cos\tfrac\pi n - \Pa{x\cos\tfrac\pi n + y\sin\tfrac\pi n}},
\end{eqnarray*}
where $\re\Br{\Pa{-\cos\tfrac\pi n +\mathrm{i}\sin\tfrac\pi
n}z_1}=-\Pa{x\cos\tfrac\pi n + y\sin\tfrac\pi n}$. Thus
\begin{eqnarray*}
&&2[\cos\tfrac\pi n - \Pa{x\cos\tfrac\pi n + y\sin\tfrac\pi n}]
=\bsb^\t\bsz =(\bar\bsF\bsb)^\t(\bsF\bsz).
\end{eqnarray*}
Our goal is to prove that
\begin{eqnarray*}
x\cos\tfrac\pi n + y\sin\tfrac\pi n \leqslant\cos\tfrac\pi n.
\end{eqnarray*}
This is equivalent to prove that
$\bsb^\t\bsz=(\bar\bsF\bsb)^\t(\bsF\bsz)\geqslant0$. To this end, it
suffices to show that $\bsx:=\bar\bsF\bsb\in\bbR^n_{\geqslant0}$
since we have already $\bsF\bsz\in\bbR^n_{\geqslant0}$. In what
follows, we show that $\bsx=\bar\bsF\bsb\in \bbR^n_{\geqslant0}$.
Indeed,
\begin{eqnarray*}
\bsx&=&b_0\ket{\bar\bsf_0} + b_1\ket{\bar\bsf_1} +
b_{n-1}\ket{\bar\bsf_{n-1}}\\
&=&b_0\ket{\bar\bsf_0} + b_1\ket{\bar\bsf_1} + \bar
b_1\ket{\bar\bsf_{n-1}},
\end{eqnarray*}
where $\ket{\bsf_k}$'s are the $k$th column vectors of $\bsF$,
implying that $\bsx=(x_0,x_1,\ldots,x_{n-1})$ is identified by
\begin{eqnarray*}
x_k &=& \frac1{\sqrt{n}}\Pa{b_0 + b_1\bar\omega^k_n + \bar
b_1\bar\omega^{k(n-1)}_n} = \frac1{\sqrt{n}}\Pa{b_0 +
b_1\bar\omega^k_n + \bar
b_1\omega^k_n}\\
&=& \frac2{\sqrt{n}}\Pa{\cos\tfrac\pi n -
\cos\tfrac{(2k+1)\pi}n}=\frac4{\sqrt{n}}\sin\tfrac{(k+1)\pi}n\sin\tfrac{k\pi}n,
\end{eqnarray*}
where $k\in\set{0,1,\ldots,n-1}$. For $n\geqslant3$, we have
$\sin\tfrac{k\pi}n\sin\tfrac{(k+1)\pi}n\geqslant0$ for
$k\in\set{0,1,\ldots,n-1}$. Therefore
$\bsx=\bar\bsF\bsb\in\bbR^n_{\geqslant0}$. Based on this
observation, we see that
$$
2\Br{\cos\tfrac\pi n - \Pa{x\cos\tfrac\pi n + y\sin\tfrac\pi
n}}=(\bar\bsF\bsb)^\t(\bsF\bsz)=\bsx^\t(\bsF\bsz)\geqslant0,
$$
implying that $x\cos\tfrac\pi n + y\sin\tfrac\pi n\leqslant
\cos\tfrac\pi n$.
\end{itemize}
From the definition of $\cB_n|_{\mathrm{circ}}$, together
$\cQ_n=\cP_n$, it follows that
\begin{eqnarray*}
\cB_n|_{\mathrm{circ}} = \Set{z^n\mid z\in \cQ_n} = \Set{z^n\mid
z\in \cP_n}.
\end{eqnarray*}
Moreover,
\begin{eqnarray*}
\partial\cB_n|_{\mathrm{circ}} &=& \Set{z^n\mid
z\in \partial\cP_n} = \Set{z^n\mid z=t+(1-t)\omega_n}\\
&=& \Set{(t+(1-t)\omega_n)^n\mid t\in[0,1]}.
\end{eqnarray*}
The proof is complete.
\end{proof}

\subsection{Convexity of the set $\cB_n|_{\mathrm{circ}}$}

In this subsection, the convexity of $\cB_n|_{\mathrm{circ}}$ will
be established immediately.

\begin{thrm}[\cite{Zhang2025PRA1}]\label{th:conv}
The set $\cB_n|_{\mathrm{circ}}$ is convex in $\bbC$.
\end{thrm}

\begin{proof}
The boundary curve $\partial\cB_n|_{\mathrm{circ}}$ in
Eq.~\eqref{eq:bcurve} can be described as
\begin{eqnarray}
r_n(\theta)e^{\mathrm{i}\theta}=\cos^n\Pa{\tfrac\pi
n}\sec^n\Pa{\tfrac{\theta-\pi}n}e^{\mathrm{i}\theta},\quad\theta\in[0,2\pi].
\end{eqnarray}
Indeed, let
\begin{eqnarray}\label{eq:polarequation}
t:=f_n(\theta)=\tfrac12\Br{1-\cot\Pa{\tfrac\pi
n}\tan\Pa{\tfrac{\theta-\pi}n}},\quad \theta\in[0,2\pi].
\end{eqnarray}
It is easily seen that the function $f_n:[0,2\pi]\to[0,1]$ is a
monotonically decreasing and continuous function, moreover, $f_n$ is
one-one and onto \cite{Chen2026}. The graph of the boundary curve is
symmetric with respect to the real axis. Then we show that the graph
of this curve $r_n(\theta)=\cos^n\Pa{\frac\pi
n}\sec^n\Pa{\frac{\theta-\pi}n}$ in Cartesian Coordinate System is
\emph{concave} on the open interval $(0,\pi)$. This implies that the
region below this curve must be a convex region. By the symmetry of
this curve with respect to the real axis, $r_n(\theta)$ is
\emph{convex} on the open interval $(\pi,2\pi)$. Both together form
the set $\cB_n|_{\mathrm{circ}}$. Let us rewrite it as the
parametric equation in parameter $\theta$:
\begin{eqnarray*}
\begin{cases}
x(\theta) =r_n(\theta)\cos\theta=\cos^n\Pa{\frac\pi n}\sec^n\Pa{\frac{\theta-\pi}n}\cos\theta,\\
y(\theta) =r_n(\theta)\sin\theta=\cos^n\Pa{\frac\pi
n}\sec^n\Pa{\frac{\theta-\pi}n}\sin\theta.
\end{cases}
\end{eqnarray*}
In fact, for $n\geqslant4$ and $\theta\in(0,\pi)$,
\begin{eqnarray*}
\frac{\dif^2 y}{\dif x^2} = -\tfrac{n-1}n\sec^n\Pa{\tfrac\pi
n}\cos^{n+1}\Pa{\tfrac{\theta-\pi}n}\csc^3\Pa{\tfrac{\pi+(n-1)\theta}n},
\end{eqnarray*}
which is negative because all factors on the right hand side are
positive up to the minus sign. In summary, $\frac{\dif^2 y}{\dif
x^2}<0$ on $(0,\pi)$, implying that the curve
$r_n(\theta)=\cos^n\Pa{\frac\pi n}\sec^n\Pa{\frac{\theta-\pi}n}$ is
concave on $[0,\pi]$. Furthermore, it is convex on $[\pi,2\pi]$ by
the symmetry with respect to the real axis. Therefore the region
$\cB_n|_{\mathrm{circ}}$ enclosed by the curve
$r_n(\theta)=\cos^n\Pa{\frac\pi n}\sec^n\Pa{\frac{\theta-\pi}n}$,
where $\theta\in[0,2\pi]$, is convex.
\end{proof}

\section{Characterization of the equality $\cB_n=\cB_n|_{\mathrm{circ}}$}
\label{sect:5}

In what follows, we will show that the reverse containment is also
true: $\cB_n\subseteq \cB_n|_{\mathrm{circ}}$. To that end, we need
the following result:

\begin{prop}[\cite{Xu2026PLA}]\label{prop:phaseinvariant}
Given any $n$-tuple of wave functions
$\Psi=(\ket{\psi_1},\ldots,\ket{\psi_n})$ with
$\Tr{\psi_1\cdots\psi_n}\neq0$, where $\psi_k\equiv\proj{\psi_k}$,
there exists a $n$-tuple of wave functions
$\widetilde\Psi=(\ket{\tilde\psi_1},\ldots,\ket{\tilde\psi_n})$ such
that the Gram matrix
$G(\widetilde\Psi)=(\inner{\tilde\psi_i}{\tilde\psi_j})_{n\times n}$
is a circular matrix with
\begin{enumerate}[(i)]
\item
$\inner{\tilde\psi_1}{\tilde\psi_2}=\inner{\tilde\psi_2}{\tilde\psi_3}=\cdots=\inner{\tilde\psi_{n-1}}{\tilde\psi_n}=\inner{\tilde\psi_n}{\tilde\psi_1}$,
\item $\arg \Tr{\psi_1\cdots\psi_n}=\arg
\Tr{\tilde\psi_1\cdots\tilde\psi_n}$,
\item $\Abs{\Tr{\psi_1\cdots\psi_n}}\leqslant
\Abs{\Tr{\tilde\psi_1\cdots\tilde\psi_n}}$, where $\arg
z\in[0,2\pi)$ is the principle argument of the complex number $z$.
\end{enumerate}
\end{prop}

\begin{proof}
Let $\Tr{\psi_1\cdots\psi_n}=re^{\mathrm{i}\theta}$ with
$r=\abs{\Tr{\psi_1\cdots\psi_n}}>0$ and $\theta\in[0,2\pi)$. Let
$\Inner{\psi_k}{\psi_{k\oplus1}}=r_ke^{\mathrm{i}\theta_k}$, where
$r_k=\abs{\Inner{\psi_k}{\psi_{k\oplus1}}}>0$ and
$\theta_k\in[0,2\pi)$ for $k\in\set{1,2,\ldots,n}$. Thus
$r=r_1r_2\cdots r_n$.\\
\textbf{Step 1: Based on $\Psi$, we define a new tuple of wave
functions $\Psi'=(\ket{\psi'_1},\ldots,\ket{\psi'_n})$ by choosing a
diagonal unitary $\bsT\in\sfU(1)^{\times n}$ such that
$G(\Psi')=\bsT^\dagger\bsG(\Psi)\bsT$}. In fact, by
Proposition~\ref{prop:jequiv}, $\Psi'$ is joint unitary equivalent
to $\Psi$ if and only if there exists a diagonal unitary
$\bsT\in\sfU(1)^{\times n}$ such that $G(\Psi')=\bsT^\dagger
G(\Psi)\bsT$. Thus we can choose suitable
$\bsT=\diag(e^{\mathrm{i}\alpha_1},\ldots,e^{\mathrm{i}\alpha_n})$,
and construct $\Psi'=(\ket{\psi'_1},\ldots,\ket{\psi'_n})$ by
defining $\ket{\psi'_k} :=e^{\mathrm{i}\alpha_k}\ket{\psi_k}$, where
$k\in\set{1,\ldots,n}$. We require the new tuple $\Psi'$ to satisfy
the condition that all consecutive inner products have the same
phase:
\begin{eqnarray*}
\arg\Inner{\psi'_1}{\psi'_2}=\arg\Inner{\psi'_2}{\psi'_3}=\cdots=\arg\Inner{\psi'_{n-1}}{\psi'_n}=\arg\Inner{\psi'_n}{\psi'_1}=\frac\theta
n
\end{eqnarray*}
which leads to the following constraints on the phase angles
$\set{\alpha_k\mid k=1,\ldots,n}\subset\bbR$:
$$
\alpha_j=\alpha_1+(j-1)\frac\theta n-\sum^{j-1}_{k=1}\theta_k,\quad
j\in\set{2,\ldots,n}.
$$
In what follows, we explain about why we can do this! That is, we
can show that the existence of $\alpha_k$'s satisfying the
constraints. Indeed, we rewrite
$\Tr{\psi_1\cdots\psi_n}=re^{\mathrm{i}\theta}$ in two ways:
\begin{eqnarray*}
\begin{cases}
re^{\mathrm{i}\theta}=(r_1e^{\mathrm{i}\frac\theta n})\cdots (r_n
e^{\mathrm{i}\frac\theta n}),\\
re^{\mathrm{i}\theta}=(r_1e^{\mathrm{i}\theta_1})\cdots
(r_ne^{\mathrm{i}\theta_n}).
\end{cases}
\end{eqnarray*}
What we want is that $r_ke^{\mathrm{i}\frac\theta n}$ will be
identified with
$r_ke^{\mathrm{i}\theta_k}=\Inner{\psi_k}{\psi_{k\oplus1}}$ \emph{up
to a phase factor}. This leads to the definition $\ket{\psi'_k}
:=e^{\mathrm{i}\alpha_k}\ket{\psi_k}$ for some $\alpha_k\in\bbR$
\emph{to be determined}. In such definition, we hope that
$r_ke^{\mathrm{i}\frac\theta n}=\inner{\psi'_k}{\psi'_{k\oplus1}}$.
Assuming this, we get that
\begin{eqnarray*}
r_ke^{\mathrm{i}\frac\theta n}&=&\inner{\psi'_k}{\psi'_{k\oplus1}}=
e^{\mathrm{i}(\alpha_{k\oplus1}-\alpha_k)}\inner{\psi_k}{\psi_{k\oplus1}}\\
&=&
e^{\mathrm{i}(\alpha_{k\oplus1}-\alpha_k)}r_ke^{\mathrm{i}\theta_k}=r_ke^{\mathrm{i}(\alpha_{k\oplus1}-\alpha_k+\theta_k)}.
\end{eqnarray*}
Under the assumption about the existence of
$\set{\alpha_k}^n_{k=1}\subset\bbR$, we derive the constraints: The
requirement concerning unknown constants $\alpha_k$'s are determined
by
\begin{eqnarray*}
\frac\theta n=\alpha_{k\oplus1}-\alpha_k+\theta_k
\Longleftrightarrow \alpha_{k\oplus1}-\alpha_k=\frac\theta
n-\theta_k,
\end{eqnarray*}
implying that
\begin{eqnarray*}
\alpha_{j\oplus1}=\alpha_1+(j-1)\frac\theta
n-\sum^{j-1}_{k=1}\theta_k,\quad j\in\set{2,\ldots,n}.
\end{eqnarray*}
Note that $\alpha_1$ is chosen freely. Let
\begin{eqnarray*}
\gamma_j\equiv
\gamma_j(\theta,\theta_1,\ldots,\theta_n):=(j-1)\frac\theta
n-\sum^{j-1}_{k=1}\theta_k,\quad j\in\set{2,\ldots,n}.
\end{eqnarray*}
Thus
$\bsT=e^{\mathrm{i}\alpha_1}\diag\Pa{1,e^{\mathrm{i}\gamma_2},\ldots,e^{\mathrm{i}\gamma_n}}$.
Thus $\Psi'$ can be defined reasonably by using such constants
$\alpha_k$'s satisfying the mentioned constraints. In this
situation, $G(\Psi')=\bsT^\dagger G(\Psi)\bsT$ for
$\bsT=\diag(e^{\mathrm{i}\alpha_1},\ldots,e^{\mathrm{i}\alpha_n})\in\sfU(1)^{\times
n}$.\\
\textbf{Step 2: Based on $\Psi'$, we define a new tuple of wave
functions
$\tilde\Psi=(\ket{\tilde\psi_1},\ldots,\ket{\tilde\psi_n})$ such
that $G(\tilde\Psi)$ is a circulant matrix}. By acting the
\emph{circulant quantum channel} $\Phi$ on the Gram matrix
$G(\Psi')$ for the constructed tuple of wave functions $\Psi'$ in
step 1, we have that $\zero\leqslant\Phi(G(\Psi'))\in\cC_n$. Thus
there exists some tuple of wave functions
$\widetilde\Psi=(\ket{\tilde\psi_1},\ldots,\ket{\tilde\psi_n})$\footnote{We
cannot guarantee that $\tilde\Psi$ and $\Psi'$ live in the same
underlying space.} such that $G(\widetilde\Psi) = \Phi(G(\Psi'))$,
and
\begin{eqnarray*}
\inner{\tilde\psi_j}{\tilde\psi_{j\oplus1}}&=&[G(\widetilde\Psi)]_{j,j\oplus1}=[\Phi(G(\Psi'))]_{j,j\oplus1}\\
&=&\frac1n\sum^n_{k=1}\inner{\psi'_k}{\psi'_{k\oplus1}}=\Pa{\frac1n\sum^n_{k=1}r_k}e^{\mathrm{i}\frac\theta
n},\quad\forall j\in\set{1,2,\ldots,n}.
\end{eqnarray*}
By the Arithmetic-Geometric Mean (AM-GM) Inequality, we have
\begin{eqnarray*}
&&\Abs{\Tr{\psi_1\cdots\psi_n}} =
\Abs{\prod^n_{k=1}\Inner{\psi_k}{\psi_{k\oplus1}}} =
\prod^n_{k=1}r_k\\
&&\leqslant \Pa{\frac1n\sum^n_{k=1}r_k}^n =
\Abs{\prod^n_{k=1}\Inner{\tilde\psi_k}{\tilde\psi_{k\oplus1}}}=
\Abs{\Tr{\tilde\psi_1\cdots\tilde\psi_n}}.
\end{eqnarray*}
In addition, we also see that
\begin{eqnarray*}
\begin{cases}
\Tr{\psi_1\cdots\psi_n}
=\Pa{\prod^n_{k=1}r_k}e^{\mathrm{i}\theta},\\
\Tr{\tilde\psi_1\cdots\tilde\psi_n}
=\Pa{\frac1n\sum^n_{k=1}r_k}^ne^{\mathrm{i}\theta}.
\end{cases}
\end{eqnarray*}
This completes the proof.
\end{proof}

We note that the technique used in the proof of
Proposition~\ref{prop:phaseinvariant} is a variation of one was
previously employed in studies of the numerical ranges of weighted
cyclic matrices \cite{Chien2023LAA,Gau2024LAA}. Specifically, the
so-called $n\times n$ \emph{weighted cyclic matrix} is of the
following shape:
\begin{eqnarray}
\bsA=\Pa{\begin{array}{cccc}
           0 & a_1 &  & \\
            & 0 & \ddots & \\
            &  & \ddots & a_{n-1} \\
           a_n &  &  & 0
         \end{array}
}\quad (a_k\in\bbC \text{ for all } k\in\set{1,\ldots,n}).
\end{eqnarray}
Let $a_k=\abs{a_k}e^{\mathrm{i}\theta_k}$ for $\theta_k\in\bbR$ and
\begin{eqnarray}
\bsA'=\Pa{\begin{array}{cccc}
           0 & \abs{a_1} &  & \\
            & 0 & \ddots & \\
            &  & \ddots & \abs{a_{n-1}} \\
           \abs{a_n} &  &  & 0
         \end{array}
}.
\end{eqnarray}
Then we have that both $\bsA$ and
$e^{\mathrm{i}\frac1n\sum^n_{k=1}\theta_k}\bsA'$ are unitary similar
via a diagonal unitary matrix.

\begin{thrm}[\cite{Xu2026PLA,Pratapsi2025PRA}]
For each integer $n\geqslant3$, it holds that
\begin{eqnarray}
\cB_n=\cB_n|_{\mathrm{circ}}.
\end{eqnarray}
\end{thrm}

\begin{proof}
With the preparations in the preceding, we can show that
$\cB_n\subseteq\cB_n|_{\mathrm{circ}}$. Indeed, choose any
$z\in\cB_n$, if $z=0$, apparently $0\in \cB_n|_{\mathrm{circ}}$; if
$z\neq0$, which can be realized as $z=\Tr{\psi_1\cdots\psi_n}$ for
an $n$-tuple of wave functions
$\Psi=(\ket{\psi_1},\ldots,\ket{\psi_n})$. By
proposition~\ref{prop:phaseinvariant}, there exists an $n$-tuple of
wave functions $\tilde
\Psi=(\ket{\tilde\psi_1},\ldots,\ket{\tilde\psi_n})$ such that
$G(\tilde\Psi)\in\cG_n\cap\cC_n$ with
\begin{enumerate}[(i)]
\item
$\inner{\tilde\psi_1}{\tilde\psi_2}=\inner{\tilde\psi_2}{\tilde\psi_3}=\cdots=\inner{\tilde\psi_{n-1}}{\tilde\psi_n}=\inner{\tilde\psi_n}{\tilde\psi_1}$,
\item $\arg \Tr{\psi_1\cdots\psi_n}=\arg
\Tr{\tilde\psi_1\cdots\tilde\psi_n}$,
\item $\Abs{\Tr{\psi_1\cdots\psi_n}}\leqslant
\Abs{\Tr{\tilde\psi_1\cdots\tilde\psi_n}}$.
\end{enumerate}
Let $\tilde z=\Tr{\tilde\psi_1\cdots\tilde\psi_n}$. From the above
proof, we see that $\tilde
z=\Pa{\frac1n\sum^n_{k=1}r_k}^ne^{\mathrm{i}\theta}\in\cB_n|_{\mathrm{circ}}$.
Due to the fact that $\cB_n|_{\mathrm{circ}}$ is convex by
Theorem~\ref{th:conv}, and thus star-shaped, it follows that
\begin{eqnarray}
z=\Pa{\prod^n_{k=1}r_k}e^{\mathrm{i}\theta}=\frac{\prod^n_{k=1}r_k}{\Pa{\frac1n\sum^n_{k=1}r_k}^n}\Pa{\frac1n\sum^n_{k=1}r_k}^ne^{\mathrm{i}\theta}=\frac{\prod^n_{k=1}r_k}{\Pa{\frac1n\sum^n_{k=1}r_k}^n}\tilde
z\in\cB_n|_{\mathrm{circ}}
\end{eqnarray}
because
$\frac{\prod^n_{k=1}r_k}{\Pa{\frac1n\sum^n_{k=1}r_k}^n}\in(0,1]$.
Therefore $\cB_n\subseteq\cB_n|_{\mathrm{circ}}$. It is trivially
that $\cB_n|_{\mathrm{circ}}\subseteq\cB_n$. Finally, we obtain that
$\cB_n=\cB_n|_{\mathrm{circ}}$.
\end{proof}

\begin{prop}[\cite{Zhang2025PRA1}]
The curve $r_n(\theta)e^{\mathrm{i}\theta}=\cos^n(\tfrac\pi
n)\sec^n(\tfrac{\theta-\pi}n)e^{\mathrm{i}\theta}$, where
$\theta\in[0,2\pi]$, can be attained by a family of single-parameter
qubit pure states. Based on this result, we can infer that
\begin{eqnarray*}
\cB_n=\cB_n(2).
\end{eqnarray*}
\end{prop}

\begin{proof}
For the orthonormal basis $\set{\ket{0},\ket{1}}$ of $\bbC^2$,
consider the $n$-tuple of qubit pure states
\begin{eqnarray}
\ket{\psi_{k+1}(\gamma)}:=\sin\gamma\ket{0}+\omega^k_n\cos\gamma
\ket{1},
\end{eqnarray}
called \emph{Oszmaniec-Brod-Galv\~{a}o's states}, where
$k\in\set{0,1,\ldots,n-1}$ and
$\omega_n=\exp\Pa{\tfrac{2\pi\mathrm{i}}n}$. Then,
\begin{eqnarray*}
\Inner{\psi_k(\gamma)}{\psi_{k+1}(\gamma)} =
\sin^2\gamma+\omega_n\cos^2\gamma=\Inner{\psi_n(\gamma)}{\psi_1(\gamma)}.
\end{eqnarray*}
Therefore
\begin{eqnarray*}
\Tr{\psi_1\psi_2\cdots\psi_n} &=&
(\sin^2\gamma+\omega_n\cos^2\gamma)^n = [t+(1-t)\omega_n]^n,
\end{eqnarray*}
where $\sin^2\gamma=t$. For each $\gamma$, there exists uniquely $t$
or $\theta\in[0,2\pi]$ via Eq.~\eqref{eq:polarequation} such that
\begin{eqnarray*}
\Tr{\psi_1\psi_2\cdots\psi_n} = [t+(1-t)\omega_n]^n =
\cos^n\Pa{\tfrac\pi n}\sec^n\Pa{\tfrac{\theta-\pi}n}.
\end{eqnarray*}
Based on this result, we see that $\cB_n\subset\cB_n(2)$. In
addition, it is trivially that $\cB_n(2)\subset\cB_n$. Therefore
$\cB_n=\cB_n(2)$. This completes the proof.
\end{proof}
In fact, we infer from this result that $\cB_n(d)=\cB_n(2)$. We can
summarize these results
\cite{Fernandes2024PRL,Li2025PRA,Zhang2025PRA1,Xu2026PLA,Pratapsi2025PRA}
in the preceding sections into the following:

\begin{thrm}\label{th:allproperties}
For the set of $n$th-order Bargmann invariants $\cB_n$, where
$3\leqslant n\in\bbN$, it holds that
\begin{enumerate}[(i)]
\item $\cB^\circ_n(d)=\cB^\bullet_n(d)=:\cB_n(d)$.
\item $\cB_n(d)=\cB_n(2)=:\cB_n$ for all integer $d\geqslant2$.
\item $\cB_n=\cB_n|_\mathrm{circ}$.
\item The boundary curve $\partial\cB_n=\partial\cB_n|_\mathrm{circ}$ is identified with the graph of the
polar equation $r_n(\theta)e^{\mathrm{i}\theta}=\cos^n(\frac\pi
n)\sec^n(\frac{\theta-\pi}n)e^{\mathrm{i}\theta}$, where
$\theta\in[0,2\pi]$. See the Figure~\ref{fig:bcurve}.
\item The set $\cB_n$ is a convex set in $\bbC$.
\end{enumerate}
\end{thrm}

\begin{proof}
The proof follows immediately from the preceding sections.
\end{proof}

\begin{figure}[ht]\centering
{\begin{minipage}[b]{0.6\linewidth}
\includegraphics[width=1\textwidth]{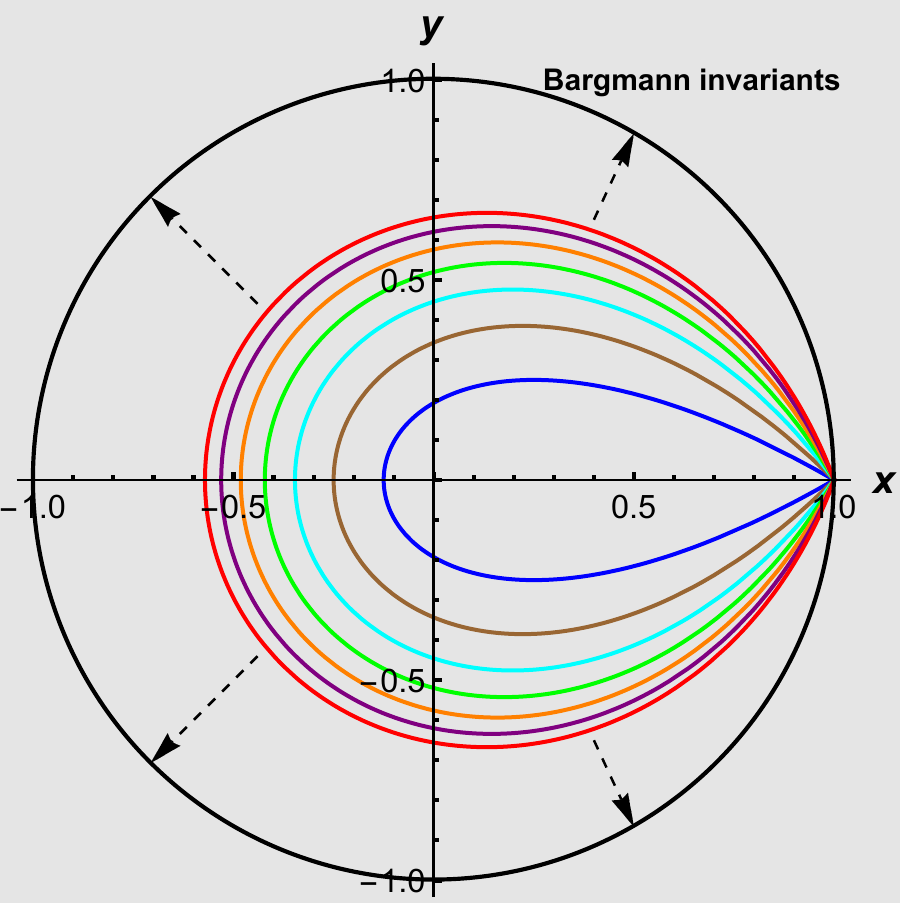}
\end{minipage}}
\caption{The graphs of boundary curves $\partial\cB_n$'s for
$n\in\set{3(\text{blue}),4(\text{brown}),5(\text{cyan}),
6(\text{green}),7(\text{orange}),8(\text{purple}),9(\text{red})}$.
The black curve is the unit circle. Here the horizontal axis means
the real part $x=\re\Tr{\psi_1\cdots\psi_n}$ and the vertical axis
stands for the imaginary part
$y=\im\Tr{\psi_1\cdots\psi_n}$.}\label{fig:bcurve}
\end{figure}

Based on the above Theorem~\ref{th:allproperties}, a longstanding
open problem \cite{Wagner2025PhD}, regarding the convexity of the
set $\cB_n$ of $n$th-order Bargmann invariants, is perfectly
resolved.

\begin{cor}[\cite{Li2025PRA,Zhang2025PRA2}]
It holds that
\begin{eqnarray}
\cB_n\subset \Br{-\cos^n(\tfrac\pi n),1}\times[-\tau_n,\tau_n],
\end{eqnarray}
where
\begin{eqnarray}
\tau_n:=\cos^n(\tfrac\pi n)\sec^{n-1}\Pa{\tfrac\pi{2(n-1)}}.
\end{eqnarray}
\end{cor}

\begin{proof}
Define by $y_n(\theta):=r_n(\theta)\sin\theta=\cos^n(\tfrac\pi
n)\sec^n(\tfrac{\theta-\pi} n)\sin\theta$ and let
$$
\tau_n:=\max_{\theta\in[0,\pi]}y_n(\theta).
$$
It suffices to identity $\tau_n$. In fact,
\begin{eqnarray*}
y'_n(\theta) = \cos^n(\tfrac\pi n)
\sec^{n+1}\Pa{\tfrac{\theta-\pi}n}\cos\Pa{\tfrac\pi n +
(1-\tfrac1n)\theta},
\end{eqnarray*}
which vanishes if and only if $\theta=\frac{n-2}{n-1}\frac\pi2$ is
the stationary point of $y_n(\theta)$ in $[0,\pi]$. Based on the
above reasoning, we get that
\begin{eqnarray*}
\tau_n = y_n\Pa{\tfrac{n-2}{n-1}\tfrac\pi2}=\cos^n(\tfrac\pi
n)\sec^{n-1}\Pa{\tfrac\pi{2(n-1)}}.
\end{eqnarray*}
We have done the proof.
\end{proof}
The quantity $\tau_n$ predicts the upper bound for the imaginary
part of any $n$th-order Bargmann invariant. This immediately raises
the question: what is the complete set of $n$-tuples
$\Psi=(\ket{\psi_1},\ldots,\ket{\psi_n})$ in $\bbC^2$ that achieve
$\im\Tr{\psi_1\ldots\psi_n}=\tau_n$? We leave this characterization,
along with the operational interpretation of the condition
$z\in\cB_{n+1}\setminus\cB_n$, as interesting open problems for
future work.

\section{An alternative characterization of $\cB^\circ_n(d)$}
\label{sect:6}

In this section, we will present another approach to the
characterization of $\partial\cB^\circ_n(d)$, where $n\in\set{3,4}$.
Before formal proof, we make an preparation. First, we recall a
notion of envelope of family of plane curves.

\begin{definition}[Envelope, \cite{Bickel2020}]
Let $\cF$ denote a family of curves in the $xy$ plane (rectangular
coordinate system). We assume that $\cF$ is a family of curves given
by $F(x,y,t)=0$, where $F$ is smooth and $t$ lies in an open
interval. The envelope of $\cF$ is the set of points $(x,y)$ so that
there is a value of $t$ with both $F(x,y,t)=0$ and
$\partial_tF(x,y,t)=0$.
\end{definition}

In what follows, we will use the notion of envelope in a polar
coordinate system, and say that $\cF$ is a family of curves given by
$F(r,\theta,t)=0$. The envelope of $\cF$ is the set of points
$(r,\theta)$ so that there is a value of $t$ with both
$F(r,\theta,t)=0$ and $\partial_tF(r,\theta,t)=0$. Denote the
numerical range of $d\times d$ complex matrix $\bsA$ by
\begin{eqnarray}
W_d(\bsA):=\Set{\Innerm{\psi}{\bsA}{\psi}\mid
\ket{\psi}\in\bbC^d,\Inner{\psi}{\psi}=1}.
\end{eqnarray}
In particular, for $\bsA=\bsu\bsv^\dagger$, where
$\bsu,\bsv\in\bbC^d$, it holds \cite{Chien2001} that
\begin{eqnarray}
W_d(\bsu\bsv^\dagger)=\Set{z\in\bbC:
\abs{z}+\abs{z-\Inner{\bsv}{\bsu}}\leqslant \norm{\bsu}\norm{\bsv}}.
\end{eqnarray}
\begin{lem}[\cite{Zhang2025PRA1}]\label{lem:n=34}
For any integer $m\geqslant2$, via $\psi_{m+1}\equiv\psi_1$, it
holds that
\begin{eqnarray}
\cB^\circ_{2m-1}(d) &=&
\bigcup_{\psi_1,\ldots,\psi_m}\Inner{\psi_1}{\psi_2}\prod^m_{j=2}W_d(\out{\psi_{j+1}}{\psi_j}),\label{eq:3}\\
\cB^\circ_{2m}(d) &=&
\bigcup_{\psi_1,\ldots,\psi_m}\prod^m_{j=1}W_d(\out{\psi_{j+1}}{\psi_j}).\label{eq:4}
\end{eqnarray}
\end{lem}

\begin{proof}
For any positive integer $n\geqslant3$, it can be written as
$n=2m-1$ or $2m$ for a positive integer $m\geqslant2$. Thus
$\cB^\circ_n(d)=\cB^\circ_{2m-1}(d)$ or $\cB^\circ_{2m}(d)$ for
$m\geqslant2$. \\
(1) Note that \small{\begin{eqnarray*} \cB^\circ_{2m-1}(d)
&=&\bigcup_{\varphi_1,\ldots,\varphi_{2(m-1)}}\Inner{\varphi_1}{\varphi_2}\Inner{\varphi_2}{\varphi_3}\Inner{\varphi_3}{\varphi_4}\cdots\Inner{\varphi_{2(m-2)}}{\varphi_{2m-3}}\Inner{\varphi_{2m-3}}{\varphi_{2(m-1)}}W_d(\out{\varphi_1}{\varphi_{2(m-1)}})\\
&=&
\bigcup_{\varphi_1,\varphi_{2j}:j=1,2,\ldots,m-1}\Inner{\varphi_1}{\varphi_2}\Pa{\prod^m_{j=3}\bigcup_{\varphi_{2j-3}}\Innerm{\varphi_{2j-3}}{\out{\varphi_{2(j-1)}}{\varphi_{2(j-2)}}}{\varphi_{2j-3}}}W_d(\out{\varphi_1}{\varphi_{2(m-1)}})\\
&=&
\bigcup_{\varphi_1,\varphi_{2j}:j=1,2,\ldots,m-1}\Inner{\varphi_1}{\varphi_2}\Pa{\prod^m_{j=3}W_d(\out{\varphi_{2(j-1)}}{\varphi_{2(j-2)}})}W_d(\out{\varphi_1}{\varphi_{2(m-1)}})
\end{eqnarray*}}
By setting
$(\varphi_1,\varphi_2,\varphi_4,\ldots,\varphi_{2(m-1)})=(\psi_1,\psi_2,\ldots,\psi_m)$
and $\psi_{m+1}\equiv\psi_1$, we get that
\begin{eqnarray*}
\cB^\circ_{2m-1}(d)=
\bigcup_{\psi_1,\ldots,\psi_m}\Inner{\psi_1}{\psi_2}\prod^m_{j=2}W_d(\out{\psi_{j+1}}{\psi_j}).
\end{eqnarray*}
(2) Again note that
\begin{eqnarray*}
\cB^\circ_{2m}(d)&=&\bigcup_{\varphi_1,\ldots,\varphi_{2m-1}}
\Inner{\varphi_1}{\varphi_2}\Inner{\varphi_2}{\varphi_3}\Inner{\varphi_3}{\varphi_4}\cdots\Inner{\varphi_{2(m-1)}}{\varphi_{2m-1}}W_d(\out{\varphi_1}{\varphi_{2m-1}})\\
&=&
\bigcup_{\varphi_j:j=1,3,\ldots,2m-1}\Pa{\prod^m_{j=2}\bigcup_{\varphi_{2(j-1)}}\Innerm{\varphi_{2(j-1)}}{\out{\varphi_{2j-1}}{\varphi_{2j-3}}}{\varphi_{2(j-1)}}}W_d(\out{\varphi_1}{\varphi_{2m-1}})\\
&=&
\bigcup_{\varphi_j:j=1,3,\ldots,2m-1}\Pa{\prod^m_{j=2}W_d(\out{\varphi_{2j-1}}{\varphi_{2j-3}})}W_d(\out{\varphi_1}{\varphi_{2m-1}})
\end{eqnarray*}
By setting
$(\varphi_1,\varphi_3,\varphi_5,\ldots,\varphi_{2m-1})=(\psi_1,\psi_2,\ldots,\psi_m)$
and $\psi_{m+1}\equiv\psi_1$, we get that
\begin{eqnarray*}
\cB^\circ_{2m}(d)=
\bigcup_{\psi_1,\ldots,\psi_m}\prod^m_{j=1}W_d(\out{\psi_{j+1}}{\psi_j}).
\end{eqnarray*}
We are done.
\end{proof}
\begin{itemize}
\item If $n=3$, then by Eq.~\eqref{eq:3} in Lemma~\ref{lem:n=34}, we get
that
\begin{eqnarray*}
\cB^\circ_3(d)
=\bigcup_{\psi_1,\psi_2}\Inner{\psi_1}{\psi_2}W_d(\out{\psi_1}{\psi_2})=\bigcup_{t\in[0,1]}
\cE_t,
\end{eqnarray*}
where we used the polar decomposition
$\Inner{\psi_1}{\psi_2}=te^{\mathrm{i}\theta}$, where $t\in[0,1]$
and $\theta\in[0,2\pi)$; and
$\cE_t:=\Set{z\in\complex:\abs{z}+\abs{z-t^2}\leqslant t}$ whose
boundary curve $\partial\cE_t$ can be put in polar form:
$r(\theta)=\frac{t(1-t^2)}{2(1-t\cos\theta)}$. Thus,
$\partial\cB^\circ_3(d)$ is the envelope of the family
$\cF=\set{\partial\cE_t}_t$ of ellipses, defined by
$F(r,\theta,t):=r(1-t\cos\theta)-\frac{t(1-t^2)}2=0$. By an envelope
algorithm, eliminating $t$ in both $F(r,\theta,t)=0$ and
$\partial_tF(r,\theta,t)=\frac12(3t^2-2r\cos\theta-1)=0$, we obtain
that the envelope of $\cF$ is implicitly determined by
$$
(8\cos^3\theta)r^3+(12\cos^2\theta-27)r^2+(6\cos\theta)r+1=0.
$$
But, only one root $r=\cos^3(\frac\pi3)\sec^3(\frac{\theta-\pi}3)$
is the desired one.
\item If $n=4$, then by Eq.~\eqref{eq:4} in Lemma~\ref{lem:n=34}, we get that
\begin{eqnarray*}
\cB^\circ_4(d) =\bigcup_{\psi_1,\psi_2}
W_d(\out{\psi_2}{\psi_1})W_d(\out{\psi_1}{\psi_2})=\bigcup_{t\in[0,1]}E^2_t
\end{eqnarray*}
where
$W_d(\out{\psi_2}{\psi_1})W_d(\out{\psi_1}{\psi_2})=E_tE_t\equiv
E^2_t$ is the Minkowski product of two subsets in $\bbC$. Here
$E_t:=\Set{z\in\bbC:\abs{z}+\abs{z-t}\leqslant1}$ and
$\Inner{\psi_1}{\psi_2}:=te^{\mathrm{i}\phi}$, where $t\in[0,1]$ and
$\phi\in[0,2\pi)$. For any fixed $t\in(0,1)$, it is easily seen that
$E^2_t =\bigcup_{z\in\partial E_t}zE_t$, where $z\in\partial E_t$
can be parameterized as $z=\frac{1-t^2}{2(1-t\cos
\alpha)}e^{\mathrm{i}\alpha}$ for $\alpha\in[0,2\pi)$. Then
\begin{eqnarray*}
E^2_t = \bigcup_{\alpha\in[0,2\pi)} \tfrac{1-t^2}{2(1-t\cos
\alpha)}e^{\mathrm{i}\alpha}E_t,
\end{eqnarray*}
where $\frac{1-t^2}{2(1-t\cos\alpha)}e^{\mathrm{i}\alpha}E_t$ is the
elliptical disk
$$
\Set{z\in\complex: \abs{z}+\abs{z-\tfrac{1-t^2}{2(1-t\cos
\alpha)}te^{\mathrm{i}\alpha}}\leqslant
\tfrac{1-t^2}{2(1-t\cos\alpha)}}
$$
whose boundary can parametrized in polar form
$z=re^{\mathrm{i}\theta}$, where $r$ and $\theta$ can be connected
as
$$
r=\frac{(1-t^2)^2}{4(1-t\cos\alpha)(1-t\cos(\alpha-\theta))}.
$$
This defines the family $\tilde\cF$ of curves resulted from $\tilde
F(r,\theta,t,\alpha):=(1-t\cos\alpha)(1-t\cos(\alpha-\theta))r-\frac{(1-t^2)^2}4=0$.
The boundary curve of $E^2_t$ is the envelope of the family
$\tilde\cF$. In fact, by envelope algorithm, eliminating $\alpha$ by
setting $\tilde F(r,\theta,t,\alpha)=0$ and
$$
\partial_\alpha\tilde
F(r,\theta,t,\alpha)=rt\Br{\sin(\alpha-\theta)+\sin\alpha-t
\sin(2\alpha-\theta)}=0,
$$
we get that $r=\frac{(1-t^2)^2}{4(1-t\cos\frac\theta2)^2}$ is the
polar equation of $\partial(E^2_t)$. Now
$\cB^\circ_4(d)=\bigcup_{t\in[0,1]}E^2_t$, where $\partial(E^2_t)$
is parameterized in polar form
$r=\frac{(1-t^2)^2}{4(1-t\cos\frac\theta2)^2}$. Once again, we
define the family $\cG$ of curves by
$G(r,\theta,t):=(1-t\cos\frac\theta2)\sqrt{r} - \frac{1-t^2}2=0$.
Then $\partial_tG(r,\theta,t)=t-\sqrt{r}\cos\frac\theta2$. Now the
boundary $\partial\cB^\circ_4(d)$ is the envelope of the family
$\cG$ of curves. It can be computed as by setting $G(r,\theta,t)=0$
and $\partial_tG(r,\theta,t)=0$ by envelope algorithm. Eliminating
$t$, we get that only
$r=\frac1{(\sin\frac\theta4+\cos\frac\theta4)^4}=\cos^4(\frac\pi4)\sec^4(\frac{\theta-\pi}4)$
is the desired envelope.
\end{itemize}
We can summarize the above discussion into the following theorem:
\begin{thrm}[See Figure~\ref{fig:bc34}]
\begin{enumerate}[(i)]
\item The boundary curve $\partial\cB^\circ_3(d)$ is the envelope of
a family of curves $r=f_3(\theta,t)$ with polar coordinate
$(r,\theta)$, defined by
\begin{eqnarray}
F_3(r,\theta,t):=r(1-t\cos\theta)-\frac{t(1-t^2)}2=0,\quad
t\in[0,1],\theta\in[0,2\pi].
\end{eqnarray}
\item The boundary curve $\partial\cB^\circ_4(d)$ is the envelope of
a family of curves $r=f_4(\theta,t)$ with polar coordinate
$(r,\theta)$, defined by
\begin{eqnarray}
F_4(r,\theta,t):=r\Pa{1-t\cos\tfrac\theta2}^2-\frac{(1-t^2)^2}4=0,\quad
t\in[0,1],\theta\in[0,4\pi].
\end{eqnarray}
\end{enumerate}
\end{thrm}

\begin{figure}[h!]
\subfigure[$n=3$]{\begin{minipage}[b]{0.5\linewidth}
\includegraphics[width=1\textwidth]{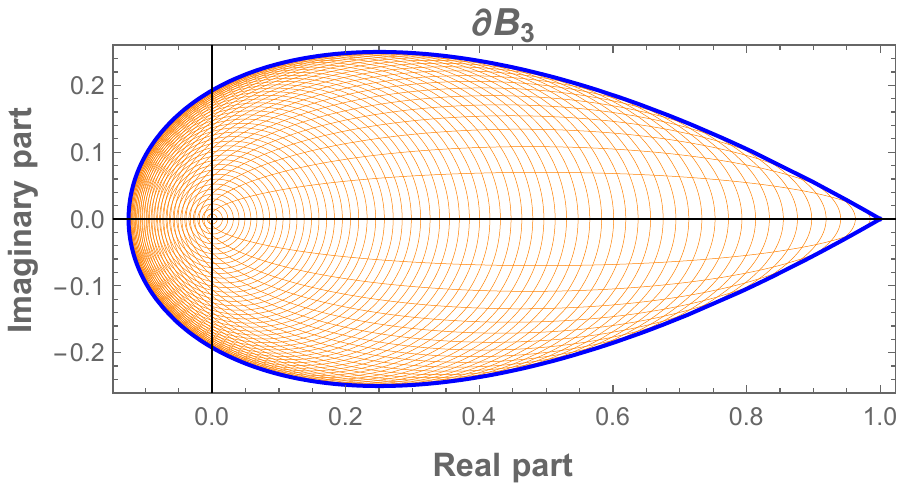}
\end{minipage}}
\subfigure[$n=4$]{\begin{minipage}[b]{0.4\linewidth}
\includegraphics[width=1\textwidth]{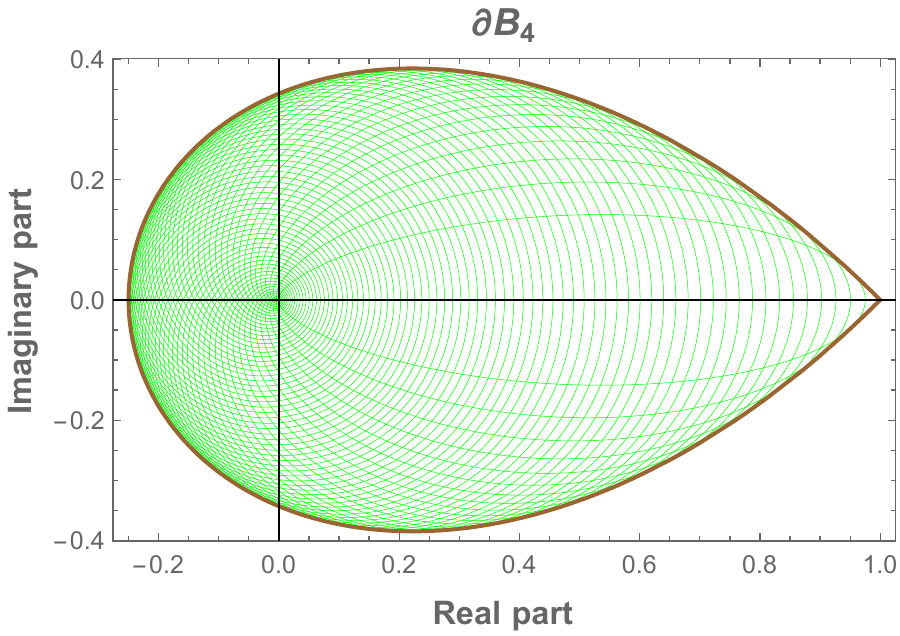}
\end{minipage}}
\caption{The boundary curve $\partial\cB^\circ_n(d)$ as an
envelope}\label{fig:bc34}
\end{figure}

Whether this envelope approach can be used to characterize
$\partial\cB^\circ_n$(d) for $n\geqslant5$ remains an open question.

We conclude this section with the following remarks. In relation to
the numerical range, the $n$th-order Bargmann invariant for the
$n$-tuple of wave functions
$\Psi=(\ket{\psi_1},\ldots,\ket{\psi_n})$ living in $\bbC^d$, given
by $\Tr{\psi_1\cdots\psi_n}$, can be rewritten as
\begin{eqnarray}
\Tr{\psi_1\cdots\psi_n} =
\Innerm{\psi_1\psi_2\cdots\psi_n}{\bsP_{d,n}(\pi_0)}{\psi_1\psi_2\cdots\psi_n},
\end{eqnarray}
where the meaning of the operator $\bsP_{d,n}(\pi_0)$ can be found
in Eq.~\eqref{eq:perrep} for the cyclic permutation
$\pi_0=(n,n-1,\ldots,2,1)\in S_n$. Based on this observation, we
find that $\cB^\circ_n(d)$ is essentially the separable numerical
range of the operator $\bsP_{d,n}(\pi_0)$ \cite{Simnacher2021PRA}.

\section{Bargmann invariant estimation in a quantum circuit}\label{sect:7}

Recently, a quantum circuit known as the cycle test
\cite{Oszmaniec2024NJP} was introduced, which enables the direct
measurement of complete sets of Bargmann invariants for both mixed
and pure quantum states. Before analyzing the measurement of
Bargmann invariants, we first define a key component of the cycle
test circuit--- the Fredkin gate.

\begin{definition}[Fredkin gate, aka Controlled-SWAP gate]
The \emph{Fredkin gate}, which is denoted by $\bsU_{\mathrm{Fred}}$,
is a unitary operator acting on $\bbC^8$, defined as follows:
\begin{eqnarray}
\bsU_{\mathrm{Fred}}\ket{0}\ket{\phi}\ket{\psi} =
\ket{0}\ket{\phi}\ket{\psi}\quad
\text{and}\quad\bsU_{\mathrm{Fred}}\ket{1}\ket{\phi}\ket{\psi} =
\ket{1}\ket{\psi}\ket{\phi}.
\end{eqnarray}
where the first qubit is the control qubit which determines whether
to swap the last two qubits. It can be also represented as
\begin{eqnarray}
\bsU_{\mathrm{Fred}}\ket{c,x,y} = \ket{c,\bar c x\oplus cy,cx\oplus
\bar cy},
\end{eqnarray}
where $\bar c:=1-c$ is the complementary bit of $c\in\set{0,1}$ and
$x,y\in\set{0,1}$. Under the computational basis,
$\bsU_{\mathrm{Fred}}$ can be represented as
\begin{eqnarray}
\bsU_{\mathrm{Fred}} = \proj{0}\ot\I_4 + \proj{1}\ot
\bsU_{\mathrm{SWAP}},
\end{eqnarray}
where the SWAP gate
$\bsU_{\mathrm{SWAP}}=\frac12\sum^3_{k=0}\sigma_k\ot\sigma_k$, where
$\sigma_0=\I_4$ and $(\sigma_1,\sigma_2,\sigma_3)$ the vector of
Pauli operators.
\end{definition}

The Fredkin gate (also called the controlled-SWAP gate) is a
three-qubit gate that swaps two target qubits conditional on the
state of a control qubit. Specifically, if the control qubit is in
state $\ket{1}$, it swaps the two target qubits; if the control is
in $\ket{0}$, it leaves them unchanged. This gate is essential in
the cycle test circuit, where multiple Fredkin gates are arranged in
a cascaded structure to measure Bargmann invariants of arbitrary
degree. Besides, we also need Hadamard gate, which is a fundamental
single-qubit gate in quantum computing. It is defined by
\begin{eqnarray}
\bsH\ket{0} =
\frac{\ket{0}+\ket{1}}{\sqrt{2}}=\ket{+}\quad\text{and}\quad
\bsH\ket{1} = \frac{\ket{0}-\ket{1}}{\sqrt{2}}=\ket{-}.
\end{eqnarray}
That is, $\bsH=\out{+}{0}+\out{-}{1}$.
\begin{itemize}
\item SWAP test to measure the two-state overlap
$\Delta_{12}=\abs{\Inner{\psi_1}{\psi_2}}^2$. The initial state is
$\ket{\Psi_{\mathrm{i}}}:=\ket{0}\ket{\psi_1}\ket{\psi_2}$, which is
transformed by $\bsH\ot\I^{\ot2}_2$ and then $\bsU_{\mathrm{Fred}}$,
followed by $\bsH\ot\I^{\ot2}_2$, thus the output state is given by
\begin{eqnarray*}
\ket{\Psi_{\mathrm{f}}}
&=&(\bsH\ot\I^{\ot2}_2)\bsU_{\mathrm{Fred}}(\bsH\ot\I^{\ot2}_2)\ket{\Psi_{\mathrm{i}}}
\\
&=&
\frac12\ket{0}\Pa{\ket{\psi_1\psi_2}+\ket{\psi_2\psi_1}}+\frac12\ket{1}\Pa{\ket{\psi_1\psi_2}-\ket{\psi_2\psi_1}}.
\end{eqnarray*}
The first qubit is measured and the probability of $0$ is denoted as
$P(0)$, given by
\begin{eqnarray*}
P(0) &=& \Tr{(\proj{0}\ot\I_4)\proj{\Psi_{\mathrm{f}}}} =
\frac{1+\abs{\Inner{\psi_1}{\psi_2}}^2}2,
\end{eqnarray*}
implying that  $\Delta_{12}=\abs{\Inner{\psi_1}{\psi_2}}^2=2P(0)-1$.
This implies that the probability of obtaining result $0$ in a
measurement allows one to indirectly obtain
$\Delta_{12}=\abs{\Inner{\psi_1}{\psi_2}}^2$.
\item CYCLE test (see Figure~\ref{fig:measurement}) to measure the $n$th Bargmann invariant $\Delta_{1\ldots n}=\Tr{\psi_1\cdots\psi_n} = \re\Tr{\psi_1\cdots\psi_n} +
\mathrm{i}\im\Tr{\psi_1\cdots\psi_n}$. The initial state
$\ket{\Psi_{\mathrm{i}}}=\ket{0}\ket{\psi_1\cdots\psi_n}$ is
transformed into the final state:
\begin{eqnarray*}
\ket{\Psi^{\bsU}_{\mathrm{f}}} &=&(\bsH\ot\I^{\ot
n}_2)\Pa{\bsU\ot\I^{\ot n}_2}\bsU_{\text{c-cycle}}(\bsH\ot\I^{\ot
n}_2)\ket{\Psi_{\mathrm{i}}},
\end{eqnarray*}
where $\bsU_{\text{c-cycle}}=\proj{0}\ot\I^{\ot n}_2 + \proj{1}\ot
\bsP_{2,n}((12\ldots n)$ for $\bsP_{2,n}((12\ldots n)$ is defined by
\begin{eqnarray}
\bsP_{2,n}((12\ldots n)\ket{i_1i_2\cdots i_n}=\ket{i_ni_1i_2\cdots
i_{n-1}},
\end{eqnarray}
where $\ket{i_k}\in\bbC^2$ with $i_k\in\set{0,1}$ for
$k=1,\ldots,n$.
\begin{figure}[ht]\centering
{\begin{minipage}[b]{0.6\linewidth}
\includegraphics[width=1\textwidth]{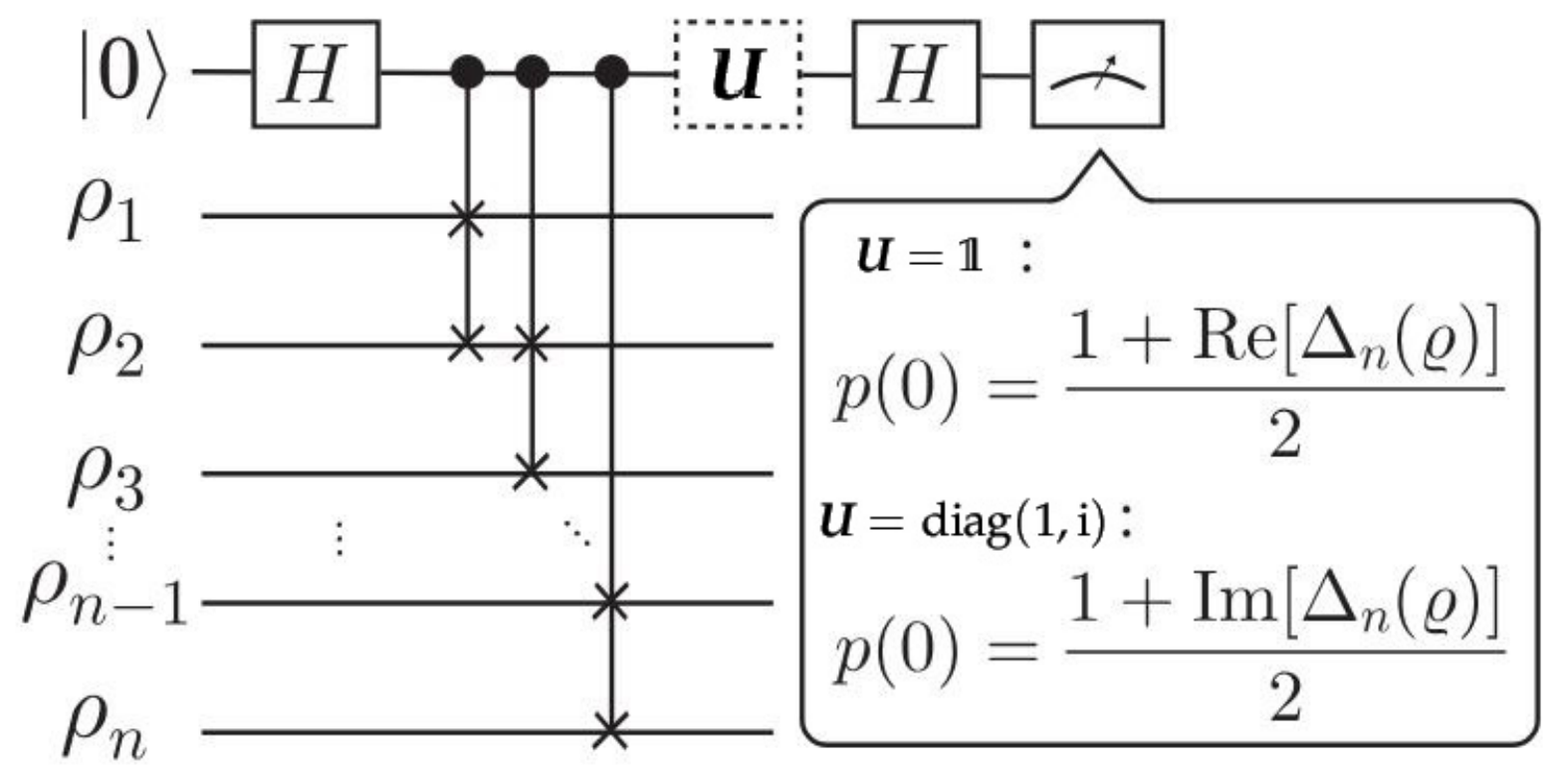}
\end{minipage}}
\caption{A quantum circuit, taken from \cite{Fernandes2024PRL}, for
measuring Bargmann invariants. Here
$\Delta_n(\varrho):=\Tr{\rho_1\cdots\rho_n}$ for
$\varrho=(\rho_1,\ldots,\rho_n)$. If $\bsU=\I$, the circuit
estimates $\re[\Delta_n(\varrho)]$ while $\bsU=\diag(1,\mathrm{i})$,
the circuit estimates
$\im[\Delta_n(\varrho)]$.}\label{fig:measurement}
\end{figure}
\begin{enumerate}[(a)]
\item If $\bsU=\I_2$, then
\begin{eqnarray*}
\ket{\Psi^{\bsU}_{\mathrm{f}}}&=&\frac12\ket{0}\Pa{\ket{\psi_1\psi_2\cdots\psi_n}+\ket{\psi_n\psi_1\psi_2\cdots\psi_{n-1}}}\\
&&+\frac12\ket{1}\Pa{\ket{\psi_1\psi_2\cdots\psi_n}-\ket{\psi_n\psi_1\psi_2\cdots\psi_{n-1}}}
\end{eqnarray*}
The first qubit is measured and the probability of $0$ is given by
\begin{eqnarray*}
P_{\bsU}(0) &=& \Tr{(\proj{0}\ot\I^{\ot
n}_2)\proj{\Psi_{\mathrm{f}}}} =
\frac{1+\re\Tr{\psi_1\cdots\psi_n}}2,
\end{eqnarray*}
implying that $\re\Tr{\psi_1\cdots\psi_n}=2P_{\bsU}(0)-1$. This
implies that the probability of obtaining result $0$ in a
measurement allows one to indirectly obtain $\Delta_{12\ldots
n}=\re\Tr{\psi_1\cdots\psi_n}$.
\item If $\bsU=\diag(1,\mathrm{i})$ the phase gate, then
\begin{eqnarray*}
\ket{\Psi^{\bsU}_{\mathrm{f}}}&=&\frac12\ket{0}\Pa{\ket{\psi_1\psi_2\cdots\psi_n}+\mathrm{i}\ket{\psi_n\psi_1\psi_2\cdots\psi_{n-1}}}\\
&&+\frac12\ket{1}\Pa{\ket{\psi_1\psi_2\cdots\psi_n}-\mathrm{i}\ket{\psi_n\psi_1\psi_2\cdots\psi_{n-1}}}
\end{eqnarray*}
The first qubit is measured and the probability of $0$ is given by
\begin{eqnarray*}
P_{\bsU}(0) &=& \Tr{(\proj{0}\ot\I^{\ot
n}_2)\proj{\Psi_{\mathrm{f}}}} =
\frac{1+\im\Tr{\psi_1\cdots\psi_n}}2,
\end{eqnarray*}
implying that $\im\Tr{\psi_1\cdots\psi_n}=2P_{\bsU}(0)-1$. This
implies that the probability of obtaining result $0$ in a
measurement allows one to indirectly obtain
$\im\Tr{\psi_1\cdots\psi_n}$.
\end{enumerate}
\end{itemize}
Thus $\Delta_{12\ldots
n}=\re\Tr{\psi_1\cdots\psi_n}+\mathrm{i}\im\Tr{\psi_1\cdots\psi_n}$
is obtained by measurements. In summary, we can list the following
algorithms about measuring the real and imaginary parts of
$n$th-order Bargmann invariants:

\begin{algorithm}[H]\caption{Estimate $\re\Tr{\rho_1\cdots\rho_n}$}
 Prepare a qubit in the $\ket{+}:=\bsH\ket{0}$ state and adjoin to it the state $\rho_1\ot\cdots\ot\rho_n$\;
 Perform a controlled cyclic permutation unitary gate, defined as
 $\bsU_{\text{c-cycle}}:=\proj{0}\ot\I^{\ot n}_2+\proj{1}\ot\bsP_{2,n}((12\ldots n))$\;
 Measure the 1st qubit in the basis
 $\set{\ket{\pm}}$, where $\ket{-}:=\bsH\ket{1}$, and record the
 outcome $X=+1$ if the 1st outcome $\ket{+}$ is observed and
 $X=-1$ if the 2nd outcome $\ket{-}$ is observed\;
 Repeat Steps 1 to 3 a number of times equal to
 $N:=O(\varepsilon^{-2}\log\delta^{-1})$ and return $\hat
 X:=\frac1N\sum^N_{i=1}X_i$, where $X_i$ is the outcome of the
 $i$-th repetition of Step 3.
\end{algorithm}

\begin{algorithm}[H]\caption{Estimate $\im\Tr{\rho_1\cdots\rho_m}$}
 Prepare a qubit in the $\ket{+}:=\bsH\ket{0}$ state and adjoin to it the state $\rho_1\ot\cdots\ot\rho_n$\;
 Perform a controlled cyclic permutation unitary gate, defined as
 $\proj{0}\ot\I^{\ot n}_2+\proj{1}\ot\bsP_{2,n}((12\ldots n))$\;
 Measure the 1st qubit in the basis
 $\set{\ket{\pm \mathrm{i}}}$, where $\ket{\pm\mathrm{i}}:=\frac{\ket{0}\pm\mathrm{i}\ket{1}}{\sqrt{2}}$, and record the
 outcome $Y=+1$ if the 1st outcome $\ket{+\mathrm{i}}$ is observed and
 $Y=-1$ if the 2nd outcome $\ket{-\mathrm{i}}$ is observed\;
 Repeat Steps 1 to 3 a number of times equal to
 $N:=O(\varepsilon^{-2}\log\delta^{-1})$ and return $\hat
 Y:=\frac1N\sum^N_{i=1}Y_i$, where $Y_i$ is the outcome of the
 $i$-th repetition of Step 3.
\end{algorithm}

\begin{prop}
It holds that $\bbE[X]=\re\Tr{\rho_1\cdots\rho_n}$ and
$\bbE[Y]=\im\Tr{\rho_1\cdots\rho_n}$.
\end{prop}

\begin{proof}
To this end, note that in the special case when all the states are
pure, i.e., $\rho_i=\proj{\psi_i}$, the input to the circuit is an
$n$-partite pure state $\rho_1\ot\cdots\ot\rho_n$, and so
\begin{eqnarray*}
\mathrm{Pr}(X=\pm1) &=& \Tr{(\proj{\pm}\ot\I^{\ot
n}_2)\bsU_{\text{c-cycle}}(\proj{+}\ot\rho_1\ot\cdots\ot\rho_n)\bsU^\dagger_{\text{c-cycle}}}\\
&=& \frac12(1\pm \re\Tr{\rho_1\cdots\rho_n}),\\
\mathrm{Pr}(Y=\pm1) &=& \Tr{(\proj{\pm\mathrm{i}}\ot\I^{\ot
n}_2)\bsU_{\text{c-cycle}}(\proj{+}\ot\rho_1\ot\cdots\ot\rho_n)\bsU^\dagger_{\text{c-cycle}}}\\
&=& \frac12(1\pm \im\Tr{\rho_1\cdots\rho_n}).
\end{eqnarray*}
Based on the above observations, we get that
\begin{eqnarray*}
\bbE[X] &=&(+1)\mathrm{Pr}(X=+1) + (-1)\mathrm{Pr}(X=-1) =
\re\Tr{\rho_1\cdots\rho_n},\\
\bbE[Y] &=&(+1)\mathrm{Pr}(Y=+1) + (-1)\mathrm{Pr}(Y=-1) =
\im\Tr{\rho_1\cdots\rho_n}.
\end{eqnarray*}
This completes the proof.
\end{proof}
Furthermore, $\hat X$ and $\hat Y$ are empirical estimates. To
guarantee that these sample averages are within $\varepsilon$ of the
true means with high probability, the sample size must satisfy
conditions provided by the Hoeffding inequality, see
Theorem~\ref{th:Hoeffding}. Specifically, for any $\varepsilon > 0$
and confidence parameter $\delta \in (0,1)$, the well-known
Hoeffding inequality provides the requisite sample complexity.
\begin{eqnarray*}
\mathrm{Pr}(\abs{\hat X - \re\Tr{\rho_1\cdots\rho_n}})\geqslant
1-\delta,\\
\mathrm{Pr}(\abs{\hat Y - \im\Tr{\rho_1\cdots\rho_n}})\geqslant
1-\delta.
\end{eqnarray*}

\begin{thrm}[Hoeffding \cite{Hoeffding}]\label{th:Hoeffding}
Suppose that we are given $n$ independent samples $R_1,\ldots,R_n$
of a bounded random variable $R$ taking values in $[a,b]$ and having
mean $\mu$. Denote the sample mean by $\bar
R_n:=\frac1n\sum^n_{k=1}R_k$. Let $\varepsilon>0$ be the desired
accuracy, and let $1-\delta$ be the desired probability, where
$\delta\in(0,1)$. Then
\begin{eqnarray*}
\mathrm{Pr}(\abs{\bar R_n-\mu}\leqslant
\varepsilon)\geqslant1-\delta
\end{eqnarray*}
provided that
$$
n\geqslant \frac{M^2}{2\varepsilon^2}\ln\Pa{\frac2\delta}
$$
where $M:=b-a$.
\end{thrm}

\begin{proof}
The proof is omitted here.
\end{proof}

\section{Applications of Bargmann invariants}\label{sect:8}

Bargmann invariants are powerful tools because they extract core,
unchanging information from quantum systems, which is useful for
both fundamental understanding (geometric phase, classification) and
practical tasks (detection, benchmarking) in quantum physics and
information science. Here we outline selected aspects of interest,
omitting full citations for brevity.

\begin{definition}[Frame graph, \cite{Chien2016}]
The so-called \emph{frame graph} of a sequence of vectors
$\set{\bsv_j}$ (or the indices $j$ themselves) to be the
(undirected) graph with
\begin{enumerate}[(i)]
\item vertices $\set{\bsv_j}$, and
\item and edge between $\bsv_i$ and $\bsv_j$, where $i\neq j$, if
and only if $\Inner{\bsv_i}{\bsv_j}\neq0$.
\end{enumerate}
A finite spanning sequence of vectors for an inner product space is
also called a \emph{finite frame}.
\end{definition}

\begin{definition}[Spanning tree]
Given a sequence of vectors $(\bsv_1,\ldots,\bsv_N)$. Let $\Gamma$
be the frame graph. The so-called \emph{spanning tree} $\cT$ of the
frame graph $\Gamma$ is a subgraph of $\Gamma$ satisfying the
following conditions:
\begin{enumerate}[(i)]
\item $\cT$ contains all the vertices in $(\bsv_1,\ldots,\bsv_N)$,
\item $\cT$ is connected, and
\item no cycle in $\cT$.
\end{enumerate}
\end{definition}

\begin{definition}[\cite{Chien2016}]
A sequence of $n(\geqslant d)$ unit vectors $\set{\bsv_j}$ in
$\bbC^d$ is \emph{equiangular} if for some $C\geqslant0$,
$$
\abs{\Inner{\bsv_i}{\bsv_j}}=C,\quad (i\neq j).
$$
The angles of a sequence of vectors $\set{\bsv_j}$ are the
$\theta_{ij}\in\bbR/(2\pi\bbZ)$ defined by
$\Inner{\bsv_i}{\bsv_j}=\abs{\Inner{\bsv_i}{\bsv_j}}e^{\mathrm{i}\theta_{ij}}$
where $\Inner{\bsv_i}{\bsv_j}\neq0$.
\end{definition}
Note that $\theta_{ji}=-\theta_{ij}$ due to
$\Inner{\bsv_j}{\bsv_i}=\overline{\Inner{\bsv_i}{\bsv_j}}$. Now we
have the following result:
\begin{prop}[\cite{Chien2016}]\label{prop:JPU}
Let $\Psi=(\ket{\psi_1},\ldots,\ket{\psi_n})$ and
$\Phi=(\ket{\phi_1},\ldots,\ket{\phi_n})$ be two $N$-tuples of
vectors in $\bbC^d$, with angles $\alpha_{ij}$ and $\beta_{ij}$.
Then $\Psi$ and $\Phi$ are joint projective unitary equivalent if
and only if the following two statements are true:
\begin{enumerate}[(i)]
\item Their Gram matrices have entries with equal moduli:
$$
\abs{\Inner{\psi_i}{\psi_j}}=\abs{\Inner{\phi_i}{\phi_j}}
$$
for all $i,j\in\set{1,\ldots,N}$.
\item Their angles are gauge equivalent in the sense: There exist
$\theta_j\in\bbR/(2\pi\bbZ)$ with
$$
\alpha_{ij}=\beta_{ij}+\theta_i-\theta_j
$$
for all $i,j\in\set{1,\ldots,N}$.
\end{enumerate}
\end{prop}

\begin{proof}
\begin{itemize}
\item[($\Longrightarrow$)] Assume that $\Psi$ and $\Phi$ are joint projective unitary equivalent.
Then $\ket{\psi_j}=c_j\bsU\ket{\phi_j}$, where $\bsU\in\sfU(d)$ is
unitary and $c_j=e^{-\mathrm{i}\theta_j}$. Then
\begin{eqnarray*}
&&e^{\mathrm{i}\alpha_{ij}}\abs{\Inner{\psi_i}{\psi_j}} =
\Inner{\psi_i}{\psi_j} = \Inner{c_i\bsU\phi_i}{c_j\bsU\phi_j} \\
&&=\bar c_ic_j\Inner{\phi_i}{\phi_j} =
e^{\mathrm{i}(\theta_i-\theta_j)}e^{\mathrm{i}\beta_{ij}}\abs{\Inner{\phi_i}{\phi_j}},
\end{eqnarray*}
which implies that
$\abs{\Inner{\psi_i}{\psi_j}}=\abs{\Inner{\phi_i}{\phi_j}}$ and
$$
\alpha_{ij}=\beta_{ij}+\theta_i-\theta_j.
$$
\item[($\Longleftarrow$)] Conversely, suppose that the above two statements are true. Let
$\ket{\tilde\phi_j}=e^{-\mathrm{i}\theta_j}\ket{\phi_j}$. Then
\begin{eqnarray*}
&&\Inner{\tilde\phi_i}{\tilde\phi_j} =
\Inner{e^{-\mathrm{i}\theta_i}\phi_i}{e^{-\mathrm{i}\theta_j}\phi_j}
=e^{\mathrm{i}(\theta_i-\theta_j)}\Inner{\phi_i}{\phi_j}\\
&&=e^{\mathrm{i}(\theta_i-\theta_j)}e^{\mathrm{i}\beta_{ij}}\abs{\Inner{\phi_i}{\phi_j}}
=
e^{\mathrm{i}\alpha_{ij}}\abs{\Inner{\psi_i}{\psi_j}}=\Inner{\psi_i}{\psi_j}.
\end{eqnarray*}
Thus $\Psi$ is joint unitary equivalent to
$\widetilde\Phi=\set{\ket{\tilde\phi_j}}$, which is joint projective
unitary equivalent to $\Phi$. Therefore $\Psi$ is joint projective
unitary equivalent to $\Phi$.
\end{itemize}
We have done the proof.
\end{proof}
We define the $n$-products of a $N$-tuple
$\Psi=(\ket{\psi_1},\ldots,\ket{\psi_N})$ to be
\begin{eqnarray}
\Delta_{i_1,\ldots,i_n}(\Psi)=\Tr{\psi_{i_1}\cdots\psi_{i_n}}
=\Inner{\psi_{i_1}}{\psi_{i_2}}\Inner{\psi_{i_2}}{\psi_{i_3}}\cdots
\Inner{\psi_{i_n}}{\psi_{i_1}},
\end{eqnarray}
where $1\leqslant i_1,\ldots,i_n\leqslant N$ for $1\leqslant
n\leqslant N$. Formally, the number of $n$-products for a $N$-tuple
is at most $\binom{N}{n}\cdot n!$.

Using this Proposition~\ref{prop:JPU}, we can derive the following
result:
\begin{thrm}[\cite{Chien2016}]\label{th:tripleprod}
For given $N$-tuples of vectors $\Psi=\Set{\ket{\psi_j}}^N_{j=1}$
and $\Phi=\Set{\ket{\phi_j}}^N_{j=1}$ in $\bbC^d$, when the frame
graphs of both $\Psi$ and $\Phi$ are complete in the sense that all
inner products are not vanished, both $\Psi$ and $\Phi$ are joint
projective unitary equivalent if and only if
\begin{eqnarray}
\Delta_{ijk}(\Psi)=\Delta_{ijk}(\Phi)
\end{eqnarray}
for all $i,j,k\in\set{1,\ldots,N}$.
\end{thrm}

\begin{proof}
$(\Longleftarrow)$ Suppose that $\Psi$ and $\Phi$ have the same
triple products, and their common frame graph is complete, then all
the triple products are nonzero. Assume that
\begin{eqnarray*}
\Inner{\psi_i}{\psi_j}\Inner{\psi_j}{\psi_k}\Inner{\psi_k}{\psi_i} =
\Inner{\phi_i}{\phi_j}\Inner{\phi_j}{\phi_k}\Inner{\phi_k}{\phi_i}\quad
(\forall i,j,k).
\end{eqnarray*}
Thus from the above equation, we infer that
\begin{itemize}
\item $\Delta_{iii}(\Psi)=\Inner{\psi_i}{\psi_i}^3=\norm{\psi_i}^6$ for all $i$;
\item $\Delta_{iij}(\Psi) =
\Inner{\psi_i}{\psi_i}\abs{\Inner{\psi_i}{\psi_j}}^2$ for all $i,j$.
\end{itemize}
We see that
\begin{eqnarray*}
\abs{\Inner{\psi_i}{\psi_j}}=\sqrt{\Delta_{iij}(\Psi)\Delta^{-\frac13}_{iii}(\Psi)}=\sqrt{\Delta_{iij}(\Phi)\Delta^{-\frac13}_{iii}(\Phi)}=\abs{\Inner{\phi_i}{\phi_j}},\quad
\forall i,j.
\end{eqnarray*}
Let $\alpha_{ij}$ be the angles of $\Psi$ (and $\beta_{ij}$ for
$\Phi$). Since the triple products have the polar form
\begin{eqnarray*}
&&\Delta_{ijk}(\Psi) =
\Inner{\psi_i}{\psi_j}\Inner{\psi_j}{\psi_k}\Inner{\psi_k}{\psi_i}\\
&&=
e^{\mathrm{i}(\alpha_{ij}+\alpha_{jk}+\alpha_{ki})}\abs{\Inner{\psi_i}{\psi_j}\Inner{\psi_j}{\psi_k}\Inner{\psi_k}{\psi_i}}.
\end{eqnarray*}
We obtain from $\Delta_{ijk}(\Psi)=\Delta_{ijk}(\Phi)$ and
$\abs{\Inner{\psi_i}{\psi_j}}=\abs{\Inner{\phi_i}{\phi_j}}$ that
$$
\alpha_{ij}+\alpha_{jk}+\alpha_{ki}=\beta_{ij}+\beta_{jk}+\beta_{ki}.
$$
Fix $k$, and rearrange this, using $\alpha_{jk}=-\alpha_{kj}$ and
$\beta_{jk}=-\beta_{kj}$, to get that
\begin{eqnarray*}
\alpha_{ij} &=& \beta_{ij} + (\beta_{ki}-\alpha_{ki}) +
(\beta_{jk}-\alpha_{jk})\\
&=& \beta_{ij} + (\beta_{ki}-\alpha_{ki})-
(\beta_{kj}-\alpha_{kj})=\beta_{ij} + \theta_i - \theta_j
\end{eqnarray*}
where $\theta_i:=\beta_{ki}-\alpha_{ki}$, i.e., the angles of $\Psi$
and $\Phi$ are gauge equivalent by Proposition~\ref{prop:JPU}.
Therefore $\Psi$ is joint unitary equivalent to $\Phi$.
\end{proof}

\begin{thrm}[\cite{Oszmaniec2024NJP}]
For given two $N$-tuples $\Psi=(\ket{\psi_1},\ldots,\ket{\psi_N})$
and $\Phi=(\ket{\phi_1},\ldots,\ket{\phi_N})$ in $\bbC^d$, both
$\Psi$ and $\Phi$ are joint projective unitary equivalent if and
only if their $n$-products are equal, i.e.,
\begin{eqnarray}
\Delta_{i_1,\ldots,i_n}(\Psi)=\Delta_{i_1,\ldots,i_n}(\Phi),
\end{eqnarray}
where $1\leqslant i_1,\ldots,i_n\leqslant N$ for $1\leqslant
n\leqslant N$.
\end{thrm}

\begin{proof}
It suffices to find a Gram matrix $G(\Psi) = \bsT^\dagger
G(\Phi)\bsT$ by using only the $n$-products of $\Phi$, where
$1\leqslant n\leqslant N$. In particular, using $n$-products for
$n=3$, we can get the modulus $\abs{\Inner{\phi_i}{\phi_j}}$ of each
entry of $G(\Phi)$ and the frame graph of $\Phi$. In fact, from the
proof of Theorem~\ref{th:tripleprod}, we see that
\begin{eqnarray*}
\abs{\Inner{\psi_i}{\psi_j}}=\sqrt{\Delta_{iij}(\Psi)\Delta^{-\frac13}_{iii}(\Psi)}=\sqrt{\Delta_{iij}(\Phi)\Delta^{-\frac13}_{iii}(\Phi)}=\abs{\Inner{\phi_i}{\phi_j}},\quad
\forall i,j.
\end{eqnarray*}
We therefore need only determine the arguments $\alpha_{ij}$ of the
(nonzero) inner products
$\Inner{\phi_i}{\phi_j}=\abs{\Inner{\phi_i}{\phi_j}}e^{\mathrm{i}\alpha_{ij}}$,
which correspond to edges of the frame graph $\Gamma$. This can be
done on each connected component $\Gamma$ of the frame graph of
$\Phi$.
\begin{itemize}
\item Find a spanning tree $\cT$ of
the frame graph $\Gamma$ with root vertex $\bsr$ (here every other
vertex has a unique parent on the path back to the root.
Conceptually, edges can be thought of as directed away from the root
or towards it, depending on the context). This can be done because a
spanning tree of a connected graph without cycles always exists!
Moreover in such situation, the spanning tree is not uniquely
determined. Starting from the root vertex $\bsr$, we can multiply
the vertices $\phi\in\Gamma\backslash\set{\bsr}$ by unit scalars so
that the arguments of the inner products
$\Inner{c_i\phi_i}{c_j\phi_j}=\abs{\Inner{\phi_i}{\phi_j}}\Pa{\bar
c_i c_je^{\mathrm{i}\alpha_{ij}}}$, where $\abs{c_i}=\abs{c_j}=1$,
corresponding to the edges of $\Gamma$ take arbitrarily assigned
values.
\item The only entries of the Gram matrix
$G(\Phi)$ which are not yet defined are those given by the edges of
the frame graph $\Gamma$ which are not in the spanning tree $\cT$.
Since $\cT$ is a spanning tree, adding each such edge to $\cT$ gives
an $n$-cycle. The corresponding nonzero $n$-product has all inner
products already determined, except the one corresponding to the
added edge, which is therefore uniquely determined by the
$n$-product.
\end{itemize}
This completes the proof.
\end{proof}

The following result gives a complete characterization of projective
unitary invariant properties of a tuple of $N$ pure states in terms
of Bargmann invariants.

\begin{thrm}[\cite{Oszmaniec2024NJP}]
Let $\Psi=(\proj{\psi_1},\ldots,\proj{\psi_N})$ be an $N$-tuple of
pure quantum states on $\bbC^d$. Then the unitary orbit (i.e., the
joint unitary similar class) of $\Psi$ is uniquely specified by
values of at most $(N-1)^2$ Bargmann invariants. The invariants are
of degree $n\leqslant N$ and their choice depends on $\Psi$.
\end{thrm}

\begin{proof}
Our strategy is based on encoding complete joint projective unitary
invariants in a single Gram matrix in a way that depends on
orthogonality relations of states in $\Psi$. We start with the
connection between joint unitary similarity of two tuples of pure
states $\Psi=(\proj{\psi_1},\ldots,\proj{\psi_N})$ and
$\Phi=(\proj{\phi_1},\ldots,\proj{\phi_N})$ and unitary equivalence
between the associated tuples of wave-functions. Namely, $\Psi$ is
joint unitary similar to $\Phi$ if and only if it is possible to
find representing wave functions $\tilde
\Psi=(\ket{\tilde\psi_1},\ldots,\ket{\tilde\psi_N}),\tilde\Phi=(\ket{\tilde\phi_1},\ldots,\ket{\tilde\phi_N})$
that are joint unitary equivalent, where
\begin{eqnarray}
\proj{\psi_j}=\proj{\tilde\psi_j}\quad\text{and}\quad\proj{\phi_j}=\proj{\tilde\phi_j},\quad\forall
j=1,\ldots,N.
\end{eqnarray}
That is, there exists an unitary operator $\bsU\in\sfU(d)$ such that
$\ket{\tilde\phi_i}=\bsU\ket{\tilde\psi_i}$ for $i=1,\ldots,N$. The
problem of joint unitary equivalence of tuples of vectors is
equivalent to equality of the corresponding Gram matrices, i.e.,
$$
\tilde\Phi=\bsU\tilde\Psi\Longleftrightarrow
G(\tilde\Psi)=G(\tilde\Phi),
$$
where
$[G(\tilde\Psi)]_{ij}=\inner{\tilde\psi_i}{\tilde\psi_j}=\inner{\tilde\phi_i}{\tilde\phi_j}=[G(\tilde\Phi)]_{ij}$.
In summary,
\begin{eqnarray}
\Phi=\bsU\Psi\bsU^\dagger \Longleftrightarrow
\tilde\Phi=\bsU\tilde\Psi\Longleftrightarrow
G(\tilde\Psi)=G(\tilde\Phi).
\end{eqnarray}
In what follows, we construct $(\tilde\Psi,\tilde\Phi)$ from
$(\Psi,\Phi)$. Since the phase of individual wave function is not an
observable, therefore the Gram matrix of a collection of pure states
$\Psi$ is uniquely defined only up to conjugation via a diagonal
unitary matrix $\bsT\in\sfU(1)^{\times N}$. Assume now that for
every tuple of quantum states $\Psi$, we have a construction of a
valid Gram matrix $G(\tilde \Psi)$ (to be specified later) whose
entries can be expressed solely in terms of projective
unitary-invariants of states from $\Psi$. It then follows from the
above considerations that $\Psi$ is projective unitary equivalent to
$\Phi$ if and only if $G(\tilde \Psi)=G(\tilde \Phi)$.

The construction of $G(\tilde \Psi)$ proceed as follows. Without
loss of generality, we assume that the frame graph $\Gamma(\Psi)$ is
connected, i.e., every pair of vertices in $\Gamma(\Psi)$ can be
connected via a path in $\Gamma(\Psi)$. We can choose a spanning
tree $\cT(\Psi)$ of $\Gamma(\Psi)$.

We now choose vector representative $\ket{\tilde\psi_i}$ of states
in $\Psi$ in such a way that, for $\set{i,j}$ an edge in
$\cT(\Psi)$,
$\Inner{\tilde\psi_i}{\tilde\psi_j}=\abs{\Inner{\tilde\psi_i}{\tilde\psi_j}}>0$.
Every other inner product $\Inner{\tilde\psi_i}{\tilde\psi_j}$ will
be either $0$, or its phase will be fixed as follows.
\begin{enumerate}[(i)]
\item Because $\cT(\Psi)$ is a spanning tree, there exists a \emph{unique} path from $j$ to $i$ within $\cT(\Psi)$.
Suppose this path has $k$ vertices
$(\alpha_1=j,\alpha_2,\ldots,\alpha_{k-1},\alpha_k=i)$. Consider now
the $k$-cycle that would be formed by adding the vertex $j$ at the
end of this path and denote it by $C_{ij}$. By construction, every
edge in $C_{ij}$ except for $\set{i,j}$ is in $\cT(\Psi)$, and
therefore all the inner products
$\inner{\tilde\psi_{\alpha_l}}{\tilde\psi_{\alpha_{l+1}}}>0$ for
$l\in\set{1,2,\ldots,k-1}$.
\item Hence, if we denote the $k$th-order Bargmann invariant
associated to $C_{ij}$ as
$\Delta(C_{ij}):=\Delta_{\alpha_1\alpha_2\ldots\alpha_{k-1}\alpha_k}$,
we can write
$$
\Delta(C_{ij})
=\iinner{\tilde\psi_{\alpha_k}}{\tilde\psi_{\alpha_1}}\prod^{k-1}_{l=1}\iinner{\tilde\psi_{\alpha_l}}{\tilde\psi_{\alpha_{l+1}}}=
\iinner{\tilde\psi_i}{\tilde\psi_j}\prod^{k-1}_{l=1}\iinner{\tilde\psi_{\alpha_l}}{\tilde\psi_{\alpha_{l+1}}}.
$$
Therefore, we can fix the phase of every nonzero inner product that
is not in $\cT(\Psi)$ as
\begin{eqnarray}
\iinner{\tilde\psi_i}{\tilde\psi_j} =
\frac{\Delta(C_{ij})}{\prod^{k-1}_{l=1}\iinner{\tilde\psi_{\alpha_l}}{\tilde\psi_{\alpha_{l+1}}}}=\kappa\cdot
\Delta(C_{ij})\quad(\kappa>0).
\end{eqnarray}
\item Thus, all matrix elements of the so-constructed Gram matrix $
[G(\tilde\Psi)]_{ij}=\Inner{\tilde\psi_i}{\tilde\psi_j}$ are
expressed via Bargmann invariants of degree at most $N$. Since this
Gram matrix $G(\tilde
\Psi)=(\Inner{\tilde\psi_i}{\tilde\psi_j})_{N\times N}$ is positive
semidefinite, it suffices to determine $\binom{N}{2}$ matrix
elements above the main diagonal (Here the diagonal part is $\I_N$).
For each element
$\Inner{\tilde\psi_i}{\tilde\psi_j}=e^{\mathrm{i}\theta_{ij}}\abs{\Inner{\tilde\psi_i}{\tilde\psi_j}}$,
where $1\leqslant i<j\leqslant N$. Every inner product
$\Inner{\tilde\psi_i}{\tilde\psi_j}$ can be determined by measure
$\abs{\Inner{\tilde\psi_i}{\tilde\psi_j}}$ and then phase factor
$e^{\mathrm{i}\theta_{ij}}$:
\begin{enumerate}[(1)]
\item First, perform measurements for the $\binom{N}{2}$ second-order Bargmann invariants
$$
\Delta_{ij}(\tilde\Psi)=\abs{\Inner{\tilde\psi_i}{\tilde\psi_j}}^2\quad(1\leqslant
i<j\leqslant N),
$$
which yields the moduli of the matrix elements:
$\abs{\Inner{\tilde\psi_i}{\tilde\psi_j}}=\Delta^\frac12_{ij}(\tilde\Psi)$.
\item To fully specify the matrix element
$\Inner{\tilde\psi_i}{\tilde\psi_j}$, the phase factor
$e^{\mathrm{i}\theta_{ij}}$ in
$\Inner{\tilde\psi_i}{\tilde\psi_j}=e^{\mathrm{i}\theta_{ij}}\abs{\Inner{\tilde\psi_i}{\tilde\psi_j}}$
must also be determined. Observe that a spanning tree $\cT(\Psi)$
constructed from the $N$ vectors
$(\ket{\tilde\psi_1},\ldots,\ket{\tilde\psi_N})$ contains \emph{at
least} $N-1$ edges (to ensure connectivity). The corresponding inner
products can be chosen to be positive and thus carry the trivial
phase factor $1$. Consequently, among the $\binom{N}{2}$ matrix
elements, at least $N-1$ ones are already real and positive, leaving
\emph{at most} $\binom{N}{2}-(N-1)=\binom{N-1}{2}$ matrix elements
with nontrivial phase factors. For an edge
$\set{i,j}\in\Gamma(\Psi)\backslash\cT(\Psi)$ with
$\Inner{\tilde\psi_i}{\tilde\psi_j}\neq0$, the phase factor
$e^{\mathrm{i}\theta_{ij}}$ is identified with
\begin{eqnarray}
e^{\mathrm{i}\theta_{ij}} =
\frac{\Inner{\tilde\psi_i}{\tilde\psi_j}}{\abs{\Inner{\tilde\psi_i}{\tilde\psi_j}}}
= \kappa
\frac{\Delta_{\alpha_1,\ldots,\alpha_k}(\tilde\Psi)}{\Delta^{\frac12}_{ij}(\tilde\Psi)}
\end{eqnarray}
is itself a Bargmann invariant (where the numerator involves
higher-order invariants along the path connecting $i$ and $j$ in the
tree). Determining these phases therefore requires at most
$\binom{N-1}{2}$ additional Bargmann invariants.
\end{enumerate}
\end{enumerate}
In summary, at most $\binom{N}{2} + \binom{N-1}{2} = (N-1)^2$
invariants are needed to determine the unitary orbit of $\Psi$. We
have done the proof.
\end{proof}

The Cayley-Hamilton Theorem \cite{Horn2012} for a single matrix
$\bsX \in \bbC^{d\times d}$ provides the starting point for
understanding trace identities. It states that every matrix
satisfies its own characteristic polynomial:
\begin{eqnarray*}
\bsX^d+c_1(\bsX)\bsX^{d-1}+\cdots+c_d(\bsX)\I_d=\zero,
\end{eqnarray*}
where the coefficient $c_k(\bsX)$ is (up to sign) the $k$-th
elementary symmetric polynomial in the eigenvalues. Crucially, each
$c_k(\bsX)$ can be expressed as a polynomial in the power traces
$\tr{\bsX^k}$ for $k=1,\ldots,d$. Consequently, $\bsX^d$ can be
written as a linear combination of $\I_d, \bsX, \ldots, \bsX^{d-1}$
with coefficients in $\bbC[\tr{\bsX},\ldots,\tr{\bsX^d}]$. By
induction, any power $\bsX^n$ for $n \geqslant d$ can be expressed
in the same basis, with coefficients that are polynomials in these
$d$ traces. This result extends powerfully to several matrices
$\bsX_1, \ldots, \bsX_N$. A core idea is to consider a generic
linear combination $\bsY = \sum_{k=1}^N t_k \bsX_k$, where the
$t_k$'s are formal variables. Applying the Cayley-Hamilton theorem
to $\bsY$ yields a polynomial identity whose coefficients are
themselves polynomials in $\bbC[\Tr{\bsY},\ldots,\tr{\bsY^d}]$:
\begin{eqnarray*}
\bsY^d +
p_1(\Tr{\bsY},\ldots,\tr{\bsY^d})\bsY^{d-1}+\cdots+p_d(\Tr{\bsY},\ldots,\tr{\bsY^d})\I_d=\zero.
\end{eqnarray*}
Expanding these traces,
\begin{eqnarray*}
\tr{\bsY^k} = \sum^N_{i_1,\ldots,i_k=1} t_{i_1}\cdots
t_{i_k}\Tr{\bsX_{i_1}\cdots\bsX_{i_k}},
\end{eqnarray*}
expresses them in terms of traces of arbitrary words (monomials) in
the $\bsX_i$'s. An important combinatorial fact is that all
polynomial relations among these traces (trace identities) are
generated by equating coefficients of the various monomials in the
$t_i$'s obtained from the Cayley-Hamilton identity for $\bsY$.

This leads directly to the First Fundamental Theorem of matrix
invariants \cite{Procesi1976}. It states that the ring of polynomial
invariants for $N$ matrices under simultaneous conjugation, $\bsX_k
\mapsto \bsU \bsX_k \bsU^\dagger$ with $\bsU \in \sfU(d)$, is
generated by the traces of all words, $\Tr{\bsX_{i_1}\cdots
\bsX_{i_k}}$ for $k \geqslant 1$. In quantum theory, where states
are represented by Hermitian matrices, these invariant traces are
precisely the Bargmann invariants, which therefore generate the
invariant ring for Hermitian tuples.

A natural question is: what is a finite generating set for this
ring? Procesi's deep result \cite{Procesi1976} provides the answer:
traces of words of length at most $d^2$ suffice. The reasoning
involves the associative algebra $\cA$ generated by $\bsX_1, \ldots,
\bsX_N$ inside $\bbC^{d\times d}$. By Burnside's theorem, if these
matrices generate the full matrix algebra, then $\cA = \bbC^{d\times
d}$, which has dimension $d^2$. In this case, the Cayley-Hamilton
theorem applied to the regular representation of $\cA$ implies that
any word of length $\geqslant d^2$ can be expressed as a linear
combination of shorter words, with coefficients that are polynomials
in traces of words of length $\leqslant d^2$. Consequently, the
trace of any longer word can be reduced to a polynomial in traces of
shorter words. If the matrices do not generate the full algebra, the
dimension of $\cA$ is smaller, potentially leading to a lower bound,
but $d^2$ remains the universal worst-case bound.

\begin{thrm}[\cite{Oszmaniec2024NJP}]
Let $\Psi=(\rho_1,\ldots,\rho_N)$ be an $N$-tuple of mixed quantum
states on $\bbC^d$. Bargmann invariants of degree $n\leqslant d^2$
form a complete set of invariants characterizing the unitary
invariants of $\Psi$. Moreover, the number of independent invariants
can be chosen to be $(N-1)(d^2-1)$.
\end{thrm}

\begin{proof}
{\bf (1) Bargmann invariants of degree at most $d^2$ form a complete
set of unitary invariants:}
\begin{itemize}
\item \textbf{Sufficiency of degree} $\leqslant d^2$: By the Cayley-Hamilton theorem, any $d \times d$ matrix satisfies its characteristic polynomial of degree $d$. For several matrices, a theorem of Procesi and Razmyslov implies that all trace identities follow from the Cayley-Hamilton theorem, and traces of products of length greater than $d^2$ can be expressed as polynomials in traces of products of length at most $d^2$. Hence, Bargmann invariants of degree $n \leqslant d^2$ suffice to generate the full invariant ring.
\item \textbf{Separation of orbits:} If two tuples $\Psi$ and $\Psi'$
  have the same Bargmann invariants for all sequences of length up to $d^2$, then they have the same invariants for all lengths (by the sufficiency argument). By the aforementioned invariant theory result, this implies $\Psi$ and $\Psi'$
  are in the same orbit under $\GL(d, \bbC)$. Since the matrices are Hermitian, $\GL(d, \bbC)$ orbits intersect the Hermitian matrices precisely in $\sfU(d)$ orbits. Hence, $\Psi$ and $\Psi'$
  are unitarily equivalent.
\end{itemize}
{\bf (2) The number of algebraically independent invariants is
$(N-1)(d^2-1)$.} The space of $N$-tuples of density matrices has
real dimension $N(d^2-1)$. The effective group acting is
$\sfU(d)/\sfU(1)$, which has dimension $d^2-1$. For a generic tuple
(e.g., one with no nontrivial common stabilizer), the stabilizer in
$\sfU(d)/\sfU(1)$ is trivial, so the orbit dimension is $d^2-1$.
Thus, the quotient space has dimension
\begin{eqnarray}
N(d^2-1) - (d^2-1)=(N-1)(d^2-1).
\end{eqnarray}
This equals the transcendence degree of the field of rational
invariants, meaning there exist $(N-1)(d^2-1)$ algebraically
independent Bargmann invariants, and any other Bargmann invariant is
algebraically dependent on these.
\end{proof}

\subsection{Witnessing quantum imaginarity}

For a tuple of quantum states $\Psi = (\rho_1, \ldots, \rho_n)$, the
imaginary part of the $n$th-order Bargmann invariant
$\Delta_{12\ldots n}(\Psi) = \Tr{\rho_1 \cdots \rho_n}$ can witness
the \emph{set imaginarity} of $\Psi$, as defined in
\cite{Fernandes2024PRL}. For convenience, let us focus on the qubit
state case. Using Bloch representation of a qubit state,
$\rho_i=\frac12(\I_2+\bsr_i\cdot\bssigma)$, where $i=1,\ldots,n$.
Denote
\begin{eqnarray}
\bsP_n:=\rho_1\cdots\rho_n=2^{-n}\Pa{p^{(n)}_0\I_2+\bsp^{(n)}\cdot\bssigma},
\end{eqnarray}
where $p^{(1)}_0=1$ and $\bsp^{(1)}=\bsr_1$. Moreover,
$\bsP_{n+1}=\bsP_n\rho_{n+1}$. We have the following update rules:
\begin{eqnarray}
p^{(n+1)}_0 &=& p^{(n)}_0 + \Inner{\bsp^{(n)}}{\bsr_{n+1}},\\
\bsp^{(n+1)}&=& p^{(n)}_0\bsr_{n+1} +\bsp^{(n)}
+\mathrm{i}\bsp^{(n)}\times \bsr_{n+1}.
\end{eqnarray}
Let $p^{(n)}_0=a^{(n)}_0+\mathrm{i}b^{(n)}_0$ and
$\bsp^{(n)}=\bsa^{(n)}+\mathrm{i}\bsb^{(n)}$. Thus
\begin{eqnarray}
\Delta_{1\ldots
n}=\Tr{\bsP_n}=\Tr{\rho_1\cdots\rho_n}=2^{1-n}p^{(n)}_0.
\end{eqnarray}
In addition,
\begin{eqnarray}
\sR_n&:=&\spn_{\bbZ}\Set{\Inner{\bsr_i}{\bsr_j}:1\leqslant
i,j\leqslant
n},\\
\sI_n&:=&\spn_{\sR_n}\Set{\det(\bsr_i,\bsr_j,\bsr_k):1\leqslant
i<j<k\leqslant n},\\
\sA_n&:=& \spn_{\sR_n}\Set{\bsr_k:1\leqslant k\leqslant n} +
\spn_{\sI_n}\Set{\bsr_i\times \bsr_j:1\leqslant i<j\leqslant n},\\
\sB_n&:=& \spn_{\sI_n}\Set{\bsr_k:1\leqslant k\leqslant n} +
\spn_{\sR_n}\Set{\bsr_i\times \bsr_j:1\leqslant i<j\leqslant n}.
\end{eqnarray}

Denote $\Delta_{ij}=\Tr{\rho_i\rho_j}$. We know from
\cite{Zhang2025PRA1} that
\begin{itemize}
\item for $n=2$, $\Delta_{12}=\Tr{\rho_1\rho_2}=\frac{1+\Inner{\bsr_1}{\bsr_2}}2\in\sR_2$.
\item for $n=3$, $\Tr{\rho_1\rho_2\rho_3}=\frac14(a^{(3)}_0+\mathrm{i}b^{(3)}_0)$, where
\begin{eqnarray*}
\begin{cases}
a^{(3)}_0= 1+\sum_{1\leqslant i<j\leqslant 3}\Inner{\bsr_i}{\bsr_j}\in\sR_3,\\
b^{(3)}_0= \det(\bsr_1,\bsr_2,\bsr_3)\in\sI_3.
\end{cases}
\end{eqnarray*}
Based on this observation, via second-order Bargmann invariants, we
get that
\begin{eqnarray}
\begin{cases}
a^{(3)}_0&=2(\sum_{1\leqslant i<j\leqslant 3}\Delta_{ij}-1),\\
\Pa{b^{(3)}_0}^2&= \det\Pa{\begin{array}{ccc}
                         2\Delta_{11}-1 & 2\Delta_{12}-1 & 2\Delta_{13}-1 \\
                         2\Delta_{12}-1 & 2\Delta_{22}-1 & 2\Delta_{23}-1 \\
                         2\Delta_{13}-1 & 2\Delta_{23}-1 & 2\Delta_{33}-1
                       \end{array}
}.
\end{cases}
\end{eqnarray}
\item for $n=4$, $\Tr{\rho_1\rho_2\rho_3\rho_4}=\frac18(a^{(4)}_0+\mathrm{i}b^{(4)}_0)$, where
\begin{eqnarray*}
\begin{cases}
a^{(4)}_0= (1+\Inner{\bsr_1}{\bsr_2})(1+\Inner{\bsr_3}{\bsr_4})-(1-\Inner{\bsr_1}{\bsr_3})(1-\Inner{\bsr_2}{\bsr_4})+(1+\Inner{\bsr_1}{\bsr_4})(1+\Inner{\bsr_2}{\bsr_3})\in\sR_4,\\
b^{(4)}_0= \det(\bsr_1+\bsr_2,\bsr_2+\bsr_3,\bsr_3+\bsr_4)\in\sI_4.
\end{cases}
\end{eqnarray*}
Thus
\begin{eqnarray}
\begin{cases}
a^{(4)}_0&=4\Pa{\Delta_{12}\Delta_{34}+\Delta_{14}\Delta_{23}-\Delta_{13}\Delta_{24}+\Delta_{13}+\Delta_{24}-1},\\
\Pa{b^{(4)}_0}^2&=8\det \scriptsize{\Pa{\begin{array}{ccc}
\Delta_{11}+2\Delta_{12}+\Delta_{22}-2 &
\Delta_{12}+\Delta_{13}+\Delta_{22}+\Delta_{23}-2 &
\Delta_{13}+\Delta_{14}+\Delta_{23}+\Delta_{24}-2 \\
\Delta_{12}+\Delta_{13}+\Delta_{22}+\Delta_{23}-2 &
\Delta_{22}+2\Delta_{23}+\Delta_{33}-2 &
\Delta_{23}+\Delta_{24}+\Delta_{33}+\Delta_{34}-2 \\
\Delta_{13}+\Delta_{14}+\Delta_{23}+\Delta_{24}-2 &
\Delta_{23}+\Delta_{24}+\Delta_{33}+\Delta_{34}-2 &
\Delta_{33}+2\Delta_{34}+\Delta_{44}-2
                      \end{array}
}}.
\end{cases}
\end{eqnarray}
\item for $n=5$,
$\Tr{\rho_1\rho_2\rho_3\rho_4\rho_5}=\frac1{16}(a^{(5)}_0+\mathrm{i}b^{(5)}_0)$,
where
\begin{eqnarray}
\begin{cases}
a^{(5)}_0=& 1+\sum_{1\leqslant i<j\leqslant 5}\Inner{\bsr_i}{\bsr_j}+\Inner{\bsr_1}{\bsr_2}\Inner{\bsr_3}{\bsr_4}-\Inner{\bsr_1}{\bsr_3}\Inner{\bsr_2}{\bsr_4}+\Inner{\bsr_1}{\bsr_4}\Inner{\bsr_2}{\bsr_3}\\
  & +(\Inner{\bsr_2}{\bsr_3}+\Inner{\bsr_2}{\bsr_4}+\Inner{\bsr_3}{\bsr_4})\Inner{\bsr_1}{\bsr_5} + (-\Inner{\bsr_1}{\bsr_3}-\Inner{\bsr_1}{\bsr_4}+\Inner{\bsr_3}{\bsr_4})\Inner{\bsr_2}{\bsr_5}\\
  &
  +(\Inner{\bsr_1}{\bsr_2}-\Inner{\bsr_1}{\bsr_4}-\Inner{\bsr_2}{\bsr_4})\Inner{\bsr_3}{\bsr_5}\in\sR_5,\\
b^{(5)}_0=& \sum_{1\leqslant i<j<k\leqslant
5}\det(\bsr_i,\bsr_j,\bsr_k) +
\Inner{\bsr_2}{\bsr_3}\det(\bsr_1,\bsr_4,\bsr_5) -
\Inner{\bsr_1}{\bsr_3}\det(\bsr_2,\bsr_4,\bsr_5) \\
& + \Inner{\bsr_1}{\bsr_2}\det(\bsr_3,\bsr_4,\bsr_5) +
\Inner{\bsr_4}{\bsr_5}\det(\bsr_1,\bsr_2,\bsr_3)\in\sI_5.
\end{cases}
\end{eqnarray}
\end{itemize}
By induction, it holds that
\begin{eqnarray}
a^{(n)}_0\in \sR_n,\quad b^{(n)}_0\in \sI_n,\quad
\bsa^{(n)}\in\sA_n,\quad\bsb^{(n)}\in\sB_n,
\end{eqnarray}
for all $n\in\bbN$.

\begin{thrm}[\cite{Li2026}]
Let $\rho_k\in\density{\bbC^2}$ for $k=1,\ldots,n$. The $n$th-order
Bargmann invariant $\Tr{\rho_1\cdots\rho_n}$ is completely
identified by all the second-order Bargmann invariants
$\set{\Delta_{ij}\equiv\Tr{\rho_i\rho_j}:1\leqslant i,j\leqslant n}$
up to complex conjugate. In fact, there exist polynomials $\tilde
p_n,\tilde q_n\in\bbQ[\Delta_{11},\Delta_{12},\ldots,\Delta_{nn}]$,
the set of all polynomials with rational coefficients in arguments
$\Delta_{11},\Delta_{12},\ldots,\Delta_{nn}$, such that the
$n$th-order Bargmann invariant $z=\Tr{\rho_1\cdots\rho_n}$ satisfies
the following quadratic equation:
\begin{eqnarray}
z^2-2\tilde p_nz + \tilde q_n=0.
\end{eqnarray}
\end{thrm}

This result shows that the real part and the absolute value of the
imaginary part of the $n$th-order Bargmann invariant are both
determined by measurements of all second-order Bargmann invariants;
however, the sign of the imaginary part remains indeterminate by
this method.

\begin{proof}
Note that, for any complex number $z=x+\mathrm{i}y\in\bbC$ for
$x,y\in\bbR$, as unique two roots $z$ and $\bar z$, they are the
roots of the following quadratic equation with real coefficients:
\begin{eqnarray*}
t^2-(z+\bar z)t+z\bar z=t^2-2p t+q=0
\end{eqnarray*}
where $p=\re(z)$ and $q=\abs{z}^2=x^2+y^2$, we get that
\begin{eqnarray*}
z^2 -2p z+q &=& (x+\mathrm{i}y)^2-2 x(x+\mathrm{i}y) + (x^2+y^2)\\
&=& x^2-y^2+2\mathrm{i}xy - (2x^2+2\mathrm{i}xy)+ (x^2+y^2)=0.
\end{eqnarray*}
From the previous discussion, $\re\Tr{\rho_1\cdots\rho_n}$ and
$(\im\Tr{\rho_1\cdots\rho_n})^2$ are determined by
$\Inner{\bsr_i}{\bsr_j}=2\Tr{\rho_i\rho_j}-1=2\Delta_{ij}-1$ for
$1\leqslant i,j\leqslant n$. Let
\begin{eqnarray*}
p_n &=& \re\Tr{\rho_1\cdots\rho_n} = 2^{1-n}\re \Br{p^{(n)}_0} = 2^{1-n}a^{(n)}_0,\\
q_n &=& \abs{\Tr{\rho_1\cdots\rho_n}}^2=
4^{1-n}\Abs{p^{(n)}_0}^2=4^{1-n}\Pa{\Br{a^{(n)}_0}^2+\Br{b^{(n)}_0}^2}.
\end{eqnarray*}
Denote by $\tilde p_n$ and $\tilde q_n$ the expressions obtained
from $p_n$ and $q_n$, respectively, by replacing each inner product
$\Inner{\bsr_i}{\bsr_j}$ with $2\Delta_{ij}-1$. Apparently $\tilde
p_n$ and $\tilde q_n$ are in
$\bbQ[\Delta_{11},\Delta_{12},\ldots,\Delta_{nn}]$. In particular,
for $z=\Tr{\rho_1\cdots\rho_n}$, $p_n=\re\Tr{\rho_1\cdots\rho_n}$
and $q_n=\abs{\Tr{\rho_1\cdots\rho_n}}^2$. Thus
$z=\Tr{\rho_1\cdots\rho_n}$ satisfies $z^2-2\tilde p_n z+\tilde
q_n=0$. But there is a caution: We cannot identify uniquely $z$ from
the equation $z^2-2\tilde p_n z+\tilde q_n=0$.
\end{proof}
We should remark here that a similar problem can be posed in high
dimensional space. But the answer to this problem is unclear at
present.

\subsection{Discriminating locally unitary orbits via Bargmann
invariant}

Typical example of nonlocal effect is entanglement. One of
approaches towards understanding multipartite entanglement is to
study the local unitary (LU) equivalence of multipartite states.
\begin{definition}
For any two multipartite states $\rho$ and $\sigma$ acting on
underlying space $\bbC^{d_1}\ot\cdots\ot\bbC^{d_N}$, the so-called
\emph{locally unitary (LU) equivalence} between $\rho$ and $\sigma$
means that there are unitaries $\bsU_k\in\sfU(d_k)(k=1,\ldots,N)$
such that
\begin{eqnarray}
\sigma=(\bsU_1\ot\cdots\ot\bsU_N)\rho(\bsU_1\ot\cdots\ot\bsU_N)^\dagger.
\end{eqnarray}
\end{definition}
Thus entanglement can be classified by LU equivalence relation:
$\density{\bbC^{d_1}\ot\cdots\ot\bbC^{d_N}}/\mathsf{LU}$.

In the following, we focus on the two-qubit system because in this
system, we can get a finer result.
\begin{thrm}[\cite{Zhang2025PRA2}]\label{th:LUa}
For any two-qubit state $\rho_{AB}\in\density{\bbC^2\ot\bbC^2}$ and
denoting $\bsX_0=\rho_{AB}, \bsX_1=\rho_A\ot\I_B$, and
$\bsX_2=\I_A\ot\rho_B$. The set comprising of 18 local unitary
Bargmann invariants $B_k(k=1,\ldots,18)$ can completely discriminate
locally unitary orbits of the two-qubit state $\rho_{AB}$, where
$B_k$'s are defined as:
\begin{eqnarray*}
&&B_1=\Tr{\bsX_0\bsX_1},B_2=\Tr{\bsX_0\bsX_2},B_3=\Tr{\bsX_0\bsX_1\bsX_2},
B_4=\Tr{\bsX^2_0},\\
&&B_5=\Tr{\bsX^2_0\bsX_1\bsX_2},B_6=\Tr{\bsX^3_0},B_7=\Tr{\bsX^3_0\bsX_1},B_8=\Tr{\bsX^3_0\bsX_2},\\
&&B_9=\Tr{\bsX^3_0\bsX_1\bsX_2},B_{10}=\Tr{\bsX^4_0},B_{11}=\Tr{\bsX^2_0\bsX_1\bsX^2_0\bsX_1},B_{12}=\Tr{\bsX^2_0\bsX_2\bsX^2_0\bsX_2},\\
&&B_{13}=\Tr{\bsX_0\bsX_1\bsX_2\bsX^2_0\bsX_1},B_{14}=\Tr{\bsX_0\bsX_1\bsX_2\bsX^2_0\bsX_2},
B_{15}=\Tr{\bsX_0\bsX_1\bsX_2\bsX^3_0\bsX_1},\\
&&B_{16}=\Tr{\bsX_0\bsX_1\bsX_2\bsX^3_0\bsX_2},
B_{17}=\Tr{\bsX_0\bsX_1\bsX^2_0\bsX_1\bsX^3_0\bsX_1},B_{18}=\Tr{\bsX_0\bsX_2\bsX^2_0\bsX_2\bsX^3_0\bsX_2}.
\end{eqnarray*}
In other words, two states of a two-qubit system are LU equivalent
if and only if both states have equal values of all 18 LU Bargmann
invariants.
\end{thrm}

\begin{proof}
The proof relies on a simple observation: the 18 generators of the
LU Bargmann invariants produce the same subalgebra as the 18 Makhlin
invariants. Verifying this equivalence, however, requires a detailed
algebraic computation, which is carried out in Ref.
\cite{Zhang2025PRA2}.
\end{proof}

\begin{thrm}\label{th:LUb}
For two-qubits $\rho_{AB}$ and $\tau_{AB}$, if a global unitary
$\bsW\in\SU(4)$ such that
\begin{eqnarray}\label{eq:UvsLU}
\begin{cases}
\tau_{AB}&=\bsW\rho_{AB}\bsW^\dagger ,\\
\tau_A\ot\I_B &=\bsW(\rho_A\ot\I_B)\bsW^\dagger,\\
\I_A\ot\tau_B &=\bsW(\I_A\ot\rho_B)\bsW^\dagger,
\end{cases}
\end{eqnarray}
then it holds that $\tau_{AB}$ and $\rho_{AB}$ are LU equivalent.
\end{thrm}

Denote $\Psi=(\rho_{AB},\rho_A\ot\I_B,\I_A\ot\rho_B)$ and
$\Phi=(\tau_{AB},\tau_A\ot\I_B,\I_A\ot\tau_B)$. It is easily seen
that both $\rho_{AB}$ and $\tau_{AB}$ are LU equivalent if and only
if both $\Psi$ and $\Phi$ are joint LU similarity. Apparently
Eq.~\eqref{eq:UvsLU} means that both $\Psi$ and $\Phi$ are unitary
similar. This implies that both states have equal values of all 18
LU Bargmann invariants. Therefore both $\rho_{AB}$ and $\tau_{AB}$
are LU equivalent by Theorem~\ref{th:LUa}.

In what follows, we present a direct proof of this result.

\begin{proof}
From the second equality, both $\rho_A$ and $\tau_A$ have the same
spectrum. Similarly, both $\rho_B$ and $\tau_B$ have the same
spectrum. Thus there exist $\bsL_A$ and $\bsL_B$ in $\SU(2)$ such
that $\rho_X=\bsL_X\tau_X\bsL^\dagger_X$, where $X=A,B$. Define
\begin{eqnarray*}
\tilde\tau_{AB} &=& (\bsL_A\ot \bsL_B)\tau_{AB}(\bsL_A\ot
\bsL_B)^\dagger,\\
\bsU&=& (\bsL_A\ot \bsL_B)\bsW\in \SU(4).
\end{eqnarray*}
Then we see that
\begin{eqnarray*}
\begin{cases}
\tilde\tau_{AB}&=\bsU\rho_{AB}\bsU^\dagger,\\
\rho_A\ot\I_B&=\bsU(\rho_A\ot\I_B)\bsU^\dagger,\\
\I_A\ot\rho_B&=\bsU(\I_A\ot\rho_B)\bsU^\dagger.
\end{cases}
\end{eqnarray*}
Thus $[\bsU,\rho_A\ot\I_B]=0=[\bsU,\I_A\ot\rho_B]$. Moreover,
$\tilde\tau_A=\rho_A$ and $\tilde\tau_B=\rho_B$. Now it boils down
to showing: \emph{If $\bsU$ commutes simultaneously with
$\rho_A\ot\I_B$ and $\I_A\ot\rho_B$, and both
$\bsU\rho_{AB}\bsU^\dagger$ and $\rho_{AB}$ have the same marginal
states, then $\bsU\rho_{AB}\bsU^\dagger$ and $\rho_{AB}$ are locally
unitary equivalent.} The proof is completed by considering the
following three cases:
\begin{itemize}
\item \textbf{Case 1: Neither of the two marginal states is maximally
mixed.} Through locally special unitary transformations, we can
assume that
$$
\rho_A=\frac12(\I+a\sigma_3),\quad \rho_B=\frac12(\I+b\sigma_3),
$$
where $\sigma_3:=\diag(1,-1)$ and $ab\neq0$. It must be that $\bsU$
commutes simultaneously with $\sigma_3\ot\I$ and $\I\ot\sigma_3$.
This indicates that $\bsU$ is a diagonal special unitary matrix in
the computational basis. Any special diagonal unitary can be written
as
\begin{eqnarray*}
\bsU&=&\exp\Pa{\mathrm{i}(\alpha
\sigma_3\ot\I+\beta\I\ot\sigma_3+\gamma\sigma_3\ot\sigma_3)}\\
&=&\exp\Pa{\mathrm{i}\alpha \sigma_3\ot\I}\exp\Pa{\mathrm{i}\beta
\I\ot\bsZ}\exp\Pa{\mathrm{i}\gamma \sigma_3\ot\sigma_3}\\
&=&(e^{\mathrm{i}\alpha\sigma_3}\ot
e^{\mathrm{i}\beta\sigma_3})e^{\mathrm{i}\gamma\sigma_3\ot\sigma_3}
=
e^{\mathrm{i}\gamma\sigma_3\ot\sigma_3}(e^{\mathrm{i}\alpha\sigma_3}\ot
e^{\mathrm{i}\beta\sigma_3}),
\end{eqnarray*}
where $\alpha,\beta$, and $\gamma$ are in $\bbR$. Hence
\begin{eqnarray*}
\bsU\rho_{AB}\bsU^\dagger &=&
e^{\mathrm{i}\gamma\sigma_3\ot\sigma_3}(e^{\mathrm{i}\alpha\sigma_3}\ot
e^{\mathrm{i}\beta\sigma_3})\rho_{AB}(e^{-\mathrm{i}\alpha\sigma_3}\ot
e^{-\mathrm{i}\beta\sigma_3})e^{-\mathrm{i}\gamma\sigma_3\ot\sigma_3}\\
&=&e^{\mathrm{i}\gamma\sigma_3\ot\sigma_3}\tilde\rho_{AB}e^{-\mathrm{i}\gamma\sigma_3\ot\sigma_3}
\end{eqnarray*}
where $\tilde\rho_{AB}:=(e^{\mathrm{i}\alpha\sigma_3}\ot
e^{\mathrm{i}\beta\sigma_3})\rho_{AB}(e^{-\mathrm{i}\alpha\sigma_3}\ot
e^{-\mathrm{i}\beta\sigma_3})$ which is LU equivalent to
$\rho_{AB}$; moreover $\tilde\rho_A=\rho_A$ and
$\tilde\rho_B=\rho_B$. Since $e^{\mathrm{i}\alpha\sigma_3}\ot
e^{\mathrm{i}\beta\sigma_3}$ is locally diagonal special unitary
transformation, without loss of generality, assume that
$\alpha=\beta=0$. We consider the only nonlocal part
$\bsU=e^{\mathrm{i}\gamma\sigma_3\ot\sigma_3}$. According to the
previous assumptions,
$$
\rho_{AB} = \frac14\Pa{\I\ot\I + a\sigma_3\ot\I+b\I\ot\sigma_3+
\sum^3_{i,j=1}c_{ij}\sigma_i\ot\sigma_j},
$$
where $\sigma_k$'s are Pauli matrices. Let us compute
\begin{eqnarray*}
&&\bsU\rho_{AB}\bsU^\dagger =
e^{\mathrm{i}\gamma\sigma_3\ot\sigma_3}\rho_{AB}e^{-\mathrm{i}\gamma\sigma_3\ot\sigma_3}\\
&&= \frac14\Pa{\I\ot\I + a\sigma_3\ot\I+b\I\ot\sigma_3+
\sum^3_{i,j=1}c_{ij}e^{\mathrm{i}\gamma\sigma_3\ot\sigma_3}(\sigma_i\ot\sigma_j)e^{-\mathrm{i}\gamma\sigma_3\ot\sigma_3}}.
\end{eqnarray*}
In all 9 terms $\sigma_i\ot\sigma_j$, only 4 terms
$\sigma_i\ot\sigma_3$ and $\sigma_3\ot\sigma_j$, where
$i,j\in\set{1,2}$, are anti-commuting with $\sigma_3\ot\sigma_3$.
Next, note that
\begin{eqnarray*}
e^{\mathrm{i}\gamma\sigma_3\ot\sigma_3}\bsX
e^{\mathrm{i}\gamma\sigma_3\ot\sigma_3} =\begin{cases}
\bsX,&\text{if
}[\bsX,\sigma_3\ot\sigma_3]=\zero;\\
\cos(2\gamma)\bsX+
\mathrm{i}\sin(2\gamma)(\sigma_3\ot\sigma_3)\bsX,&\text{if
}\set{\bsX,\sigma_3\ot\sigma_3}=\zero.
\end{cases}
\end{eqnarray*}
due to the fact that
$e^{\mathrm{i}\gamma\sigma_3\ot\sigma_3}=\cos\gamma\I+\mathrm{i}\sin\gamma
\sigma_3\ot\sigma_3$. Thus
\begin{eqnarray*}
\sum^3_{i,j=1}c_{ij}e^{\mathrm{i}\gamma\sigma_3\ot\sigma_3}(\sigma_i\ot\sigma_j)e^{-\mathrm{i}\gamma\sigma_3\ot\sigma_3}&=&\sum^3_{k=1}c_{kk}\sigma_k\ot\sigma_k
+ c_{12}\sigma_1\ot\sigma_2 + c_{21}\sigma_2\ot\sigma_1\\
&&+\sum^2_{i=1}c_{i3}[\cos(2\gamma)\sigma_i\ot\sigma_3+\mathrm{i}\sin(2\gamma)(\sigma_3\ot\sigma_3)\sigma_i\ot\sigma_3]\\
&&+\sum^2_{i=1}c_{3i}[\cos(2\gamma)\sigma_3\ot\sigma_i+\mathrm{i}\sin(2\gamma)(\sigma_3\ot\sigma_3)\sigma_3\ot\sigma_i].
\end{eqnarray*}
Note that $\sigma^2_3=\I$ and $\sigma_3\sigma_1=\mathrm{i}\sigma_2$
and $\sigma_3\sigma_2=-\mathrm{i}\sigma_1$,
\begin{eqnarray*}
\sum^3_{i,j=1}c_{ij}e^{\mathrm{i}\gamma\sigma_3\ot\sigma_3}(\sigma_i\ot\sigma_j)e^{-\mathrm{i}\gamma\sigma_3\ot\sigma_3}&=&\sum^3_{k=1}c_{kk}\sigma_k\ot\sigma_k
+ c_{12}\sigma_1\ot\sigma_2 + c_{21}\sigma_2\ot\sigma_1\\
&&+\sum^2_{i=1}c_{i3}[\cos(2\gamma)\sigma_i\ot\sigma_3+\mathrm{i}\sin(2\gamma)(\sigma_3\sigma_i)\ot\I]\\
&&+\sum^2_{j=1}c_{3j}[\cos(2\gamma)\sigma_3\ot\sigma_j+\mathrm{i}\sin(2\gamma)\I\ot(\sigma_3\sigma_j)].
\end{eqnarray*}
\begin{enumerate}[(1)]
\item If $\sin(2\gamma)\neq0$, then
$\pm\mathrm{i}\sigma_i\ot\I(i=1,2)$ will appear in the transformed
state. This will lead to the different marginal state when
$c_{i3}\neq0$ or $c_{3j}\neq0$. Since both
$e^{\mathrm{i}\gamma\sigma_3\ot\sigma_3}\rho_{AB}e^{-\mathrm{i}\gamma\sigma_3\ot\sigma_3}$
and $\rho_{AB}$ have the same marginal states, it must be that
$c_{i3}=c_{3j}=0$ for $i,j\in\set{1,2}$. So
$$
\rho_{AB}=\frac14\Pa{\I\ot\I + a\sigma_3\ot\I+b\I\ot\sigma_3+
\sum^3_{k=1}c_{kk}\sigma_k\ot\sigma_k + c_{12}\sigma_1\ot\sigma_2 +
c_{21}\sigma_2\ot\sigma_1},
$$
implying that $[\rho_{AB},\sigma_3\ot\sigma_3]=\zero$, and thus
$e^{\mathrm{i}\gamma\sigma_3\ot\sigma_3}\rho_{AB}e^{-\mathrm{i}\gamma\sigma_3\ot\sigma_3}=(\I\ot\I)\rho_{AB}(\I\ot\I)^\dagger$.
\item If $\sin(2\gamma)=0$, then $\gamma=\frac{k\pi}2$ for an integer $k$, then
\begin{eqnarray*}
&&e^{\mathrm{i}\gamma\sigma_3\ot\sigma_3}=\cos\gamma\I+\mathrm{i}\sin\gamma\sigma_3\ot\sigma_3\\
&&= (-1)^{\lfloor \frac k2\rfloor}\Pa{\frac{1+(-1)^k}2\I\ot\I +
\mathrm{i}\frac{1-(-1)^k}2\sigma_3\ot\sigma_3},
\end{eqnarray*}
which is a product of a global phase and $\I\ot \I$ or
$\sigma_3\ot\sigma_3$, and thus a local unitary.
\end{enumerate}
Therefore, both $\bsU\rho\bsU^\dagger$ and $\rho$ are locally
unitary equivalent.
\item \textbf{Case 2: Exactly one of the two marginal states is maximally mixed; the other is
not.} Without loss of generality, we assume that
$$
\rho_A=\frac12\I,\quad \rho_B=\frac12(\I+b\sigma_3)(b\neq0).
$$
Since $[\bsU,\I\ot\rho_B]=\zero$, it follows that
$[\bsU,\I\ot\sigma_3]=\zero$ and thus $\bsU$ must be of the form
$$
\bsU=\bsU_0\ot\proj{0} + \bsU_1\ot \proj{1},
$$
where $\bsU_k\in\sfU(2)$, where $k=0,1$. Thus
$$
\bsU = (\bsU_0\ot\I)(\I\ot\proj{0}+\bsV\ot\proj{1}),\quad
\bsV:=\bsU^{-1}_0\bsU_1\in \sfU(2).
$$
Any unitary matrix $\bsV\in\sfU(2)$ can be represented as
$\bsV=e^{\mathrm{i}\phi}e^{\mathrm{i}\theta
\bsn\cdot\boldsymbol{\sigma}}$, where $\bsn\in\bbR^3$ with
$\abs{\bsn}=1$ and $\phi\in\bbR$. Note that
\begin{eqnarray*}
&&\I\ot\proj{0}+\bsV\ot\proj{1} =
\I\ot\proj{0}+e^{\mathrm{i}\phi}e^{\mathrm{i}\theta
\bsn\cdot\boldsymbol{\sigma}}\ot\proj{1}\\
&&=\Br{e^{\mathrm{i}\frac\theta2\bsn\cdot\boldsymbol{\sigma}}\ot\Pa{\begin{array}{cc}
         1 & 0 \\
         0 & e^{\mathrm{i}\phi}
       \end{array}
}}e^{-\mathrm{i}\frac\theta2\bsn\cdot\boldsymbol{\sigma}\ot\sigma_3}.
\end{eqnarray*}
Indeed,
\begin{eqnarray*}
&&\Br{e^{\mathrm{i}\frac\theta2\bsn\cdot\boldsymbol{\sigma}}\ot\Pa{\begin{array}{cc}
         1 & 0 \\
         0 & e^{\mathrm{i}\phi}
       \end{array}
}}e^{-\mathrm{i}\frac\theta2\bsn\cdot\boldsymbol{\sigma}\ot\sigma_3}
=\Br{e^{\mathrm{i}\frac\theta2\bsn\cdot\boldsymbol{\sigma}}\ot
(\proj{0}+e^{\mathrm{i}\phi}\proj{1})}e^{-\mathrm{i}\frac\theta2\bsn\cdot\boldsymbol{\sigma}\ot\sigma_3}\\
&&=\Pa{e^{\mathrm{i}\frac\theta2\bsn\cdot\boldsymbol{\sigma}}\ot\proj{0}+
e^{\mathrm{i}\phi}e^{\mathrm{i}\frac\theta2\bsn\cdot\boldsymbol{\sigma}}\ot\proj{1}}\Pa{e^{-\mathrm{i}\frac\theta2\bsn\cdot\boldsymbol{\sigma}}\ot\proj{0}+e^{\mathrm{i}\frac\theta2\bsn\cdot\boldsymbol{\sigma}}\ot\proj{1}}\\
&&=e^{\mathrm{i}\frac\theta2\bsn\cdot\boldsymbol{\sigma}}e^{-\mathrm{i}\frac\theta2\bsn\cdot\boldsymbol{\sigma}}\ot\proj{0}+e^{\mathrm{i}\phi}e^{\mathrm{i}\frac\theta2\bsn\cdot\boldsymbol{\sigma}}e^{\mathrm{i}\frac\theta2\bsn\cdot\boldsymbol{\sigma}}\ot\proj{1}\\
&&=\I\ot\proj{0}+e^{\mathrm{i}\phi}e^{\mathrm{i}\theta
\bsn\cdot\boldsymbol{\sigma}}\ot\proj{1}.
\end{eqnarray*}
Thus
\begin{eqnarray*}
\bsU &=&
\Pa{\bsU_0e^{\mathrm{i}\frac\theta2\bsn\cdot\boldsymbol{\sigma}}\ot\Pa{\begin{array}{cc}
         1 & 0 \\
         0 & e^{\mathrm{i}\phi}
       \end{array}
}}e^{-\mathrm{i}\frac\theta2\bsn\cdot\boldsymbol{\sigma}\ot\sigma_3}\\
&=&(\bsL_A\ot\bsL_B)\bsT_\gamma,
\end{eqnarray*}
where
$\bsL_A=\bsU_0e^{\mathrm{i}\frac\theta2\bsn\cdot\boldsymbol{\sigma}},
\bsL_B=\proj{0}+e^{\mathrm{i}\phi}\proj{1}$, and
$\bsT_\gamma:=e^{\mathrm{i}\gamma\bsn\cdot\boldsymbol{\sigma}\ot\sigma_3}$
for $\gamma=-\frac\theta2$. In order to show that
$\bsU\rho_{AB}\bsU^\dagger=(\bsL_A\ot\bsL_B)\bsT_\gamma\rho_{AB}\bsT^\dagger_\gamma
(\bsL_A\ot\bsL_B)^\dagger$ is LU equivalent to $\rho_{AB}$, it
suffices to show that $\bsT_\gamma\rho_{AB}\bsT^\dagger_\gamma$ is
LU equivalent to $\rho_{AB}$. Clearly
$$
(\bsT_\gamma\rho_{AB}\bsT^\dagger_\gamma)_A=\frac12\I,\quad
(\bsT_\gamma\rho_{AB}\bsT^\dagger_\gamma)_B=\rho_B.
$$
We continue the simplify the above problem. In fact, there is a
unitary $\bsQ\in \SU(2)$ such that
$\bsQ(\bsn\cdot\boldsymbol{\sigma})\bsQ^\dagger=\abs{\bsn}\sigma_3=\sigma_3$.
Define $\rho'_{AB}:=(\bsQ\ot\I)\rho_{AB}(\bsQ\ot\I)^\dagger$.
Clearly $\rho'_A=\frac12\I$ and $\rho'_B=\rho_B$. At the same time,
$$
\bsT'_\gamma:=(\bsQ\ot\I)\bsT_\gamma(\bsQ\ot\I)^\dagger=e^{\mathrm{i}\gamma\sigma_3\ot\sigma_3}.
$$
It suffices to show the desired result for
$(\bsT'_\gamma,\rho'_{AB})$. To simplify the notations, we still use
$(\bsT_\gamma,\rho_{AB})$ instead of $(\bsT',\rho'_{AB})$. Now
\begin{eqnarray*}
\rho_{AB} = \frac14\Pa{\I\ot\I + \I\ot b\sigma_3 +
\sum^3_{i,j=1}c_{ij}\sigma_i\ot\sigma_j}.
\end{eqnarray*}
Note again that
\begin{eqnarray*}
\bsT_\gamma\bsX\bsT^\dagger_\gamma =\begin{cases} \bsX,&\text{if
}[\bsX,\sigma_3\ot\sigma_3]=\zero;\\
\cos(2\gamma)\bsX+
\mathrm{i}\sin(2\gamma)(\sigma_3\ot\sigma_3)\bsX,&\text{if
}\set{\bsX,\sigma_3\ot\sigma_3}=\zero.
\end{cases}
\end{eqnarray*}
Based on this, we infer that
\begin{eqnarray*}
\bsT_\gamma\rho_{AB}\bsT^\dagger_\gamma = \frac14\Pa{\I\ot\I + \I\ot
b\sigma_3 +
\sum^3_{i,j=1}c_{ij}\bsT_\gamma(\sigma_i\ot\sigma_j)\bsT^\dagger_\gamma},
\end{eqnarray*}
where
\begin{eqnarray*}
\sum^3_{i,j=1}c_{ij}\bsT_\gamma(\sigma_i\ot\sigma_j)\bsT^\dagger_\gamma&=&\sum^3_{k=1}c_{kk}\sigma_k\ot\sigma_k
+ c_{12}\sigma_1\ot\sigma_2 + c_{21}\sigma_2\ot\sigma_1\\
&&+\sum^2_{i=1}c_{i3}[\cos(2\gamma)\sigma_i\ot\sigma_3+\mathrm{i}\sin(2\gamma)(\sigma_3\sigma_i)\ot\I]\\
&&+\sum^2_{j=1}c_{3j}[\cos(2\gamma)\sigma_3\ot\sigma_j+\mathrm{i}\sin(2\gamma)\I\ot(\sigma_3\sigma_j)].
\end{eqnarray*}
By a similar analysis performed in the Case 1, if
$\sin(2\gamma)\neq0$, it must be that $c_{i3}=c_{3j}=0$ for
$i,j\in\set{1,2}$, and thus
$$
\rho_{AB}=\frac14\Pa{\I\ot\I + b\I\ot\sigma_3+
\sum^3_{k=1}c_{kk}\sigma_k\ot\sigma_k + c_{12}\sigma_1\ot\sigma_2 +
c_{21}\sigma_2\ot\sigma_1},
$$
implying that $[\rho_{AB},\sigma_3\ot\sigma_3]=\zero$ and
$\bsT_\gamma\rho_{AB}\bsT^\dagger_\gamma=\rho_{AB}$; if
$\sin(2\gamma)=0$, then $\gamma=\frac{k\pi}2$ for an integer $k$,
then
\begin{eqnarray*}
\bsT_\gamma=\bsT_{\frac
k2\pi}\in\set{\pm\I\ot\I,\pm\mathrm{i}\sigma_3\ot\sigma_3},
\end{eqnarray*}
which is a product of a global phase and $\I\ot \I$ or
$\sigma_3\ot\sigma_3$, and thus a local unitary. The case in which
$\rho_A=\frac12(\I+a\sigma_3)(a\neq0),\rho_B=\frac12\I$ can be
treated similarly.
\item \textbf{Case 3: Both marginal states are maximally mixed.}
Assume that $\rho_A=\rho_B=\frac12\I$. Then $\rho_{AB}$ must be the
form
$$
\rho_{AB} = \frac14\Pa{\I\ot\I +
\sum^3_{i,j=1}c_{ij}\sigma_i\ot\sigma_j},
$$
where the $c_{ij}$ form a real matrix $\bsC=(c_{ij})$. By the
singular value decomposition (SVD) of real matrices, any real matrix
$\bsC$ can be written as
$$
\bsC=\tilde\bsO_A\diag(s_1,s_2,s_3)\tilde\bsO^\t_B,\quad\tilde\bsO_A,\tilde\bsO_B\in\sfO(3),
$$
and $s_1\geqslant s_2\geqslant s_3\geqslant0$ the singular values.
Let $\varepsilon_{\bsC}:=\sign(\det(\bsC))$ if $\det(\bsC)\neq0$.
Then we can get that
$$
\bsC=\bsO_A\diag(s_1,s_2,\varepsilon_{\bsC}s_3)\bsO^\t_B,\quad
\bsO_A,\bsO_B\in\SO(3).
$$
However, since we are restricted to $\SO(3)$ instead of the full
$\sfO(3)$, the equivalence class of $\bsC$ under
$\SO(3)\times\SO(3)$ is determined by two pieces of data:
\begin{itemize}
\item Three singular values $\set{s_1,s_2,s_3}$ (in non-increasing
order).
\item The sign of the determinant, $\varepsilon_{\bsC}=\sign(\det(\bsC))\in\set{\pm1}$.
\end{itemize}
We now show that the four eigenvalues of $\rho_{AB}$ uniquely
determine $\set{s_1,s_2,s_3}$ and $\sign(\det(\bsC))$. For the above
$\bsO_A,\bsO_B\in\SO(3)$, recall that $\SU(2)$ is the double cover
of $\SO(3)$, there must be two special unitaries
$\bsV_A,\bsV_B\in\SU(2)$ such that
\begin{eqnarray}\label{eq:normalform}
&&(\bsV_A\ot\bsV_B)\rho_{AB}(\bsV_A\ot\bsV_B)^\dagger =
\frac14\Pa{\I\ot\I + \sum^2_{k=1}s_k\sigma_k\ot\sigma_k +
\varepsilon_{\bsC}s_3\sigma_3\ot\sigma_3},
\end{eqnarray}
which is a Bell-diagonal state whose eigenvalues are given by
\begin{eqnarray}
\begin{cases}
\lambda_1 = \frac14(1+s_1-s_2+\varepsilon_{\bsC}s_3),\\
\lambda_2 = \frac14(1-s_1+s_2+\varepsilon_{\bsC}s_3),\\
\lambda_3 = \frac14(1+s_1+s_2-\varepsilon_{\bsC}s_3),\\
\lambda_4 = \frac14(1-s_1-s_2-\varepsilon_{\bsC}s_3).
\end{cases}
\end{eqnarray}
Its inverse are given by
\begin{eqnarray}
\begin{cases}
s_1=\frac14(\lambda_1-\lambda_2+\lambda_3-\lambda_4),\\
s_2=\frac14(-\lambda_1+\lambda_2+\lambda_3-\lambda_4),\\
\varepsilon_{\bsC}s_3=\frac14(\lambda_1+\lambda_2-\lambda_3-\lambda_4).
\end{cases}
\end{eqnarray}
This means that the four eigenvalues $\lambda_k$ of $\rho_{AB}$
uniquely determine $\set{s_1,s_2,s_3}$ and $\sign(\det(\bsC))$. Now
$\rho'_{AB}:=\bsU\rho_{AB}\bsU^\dagger$ and $\rho_{AB}$ have the
same marginal states being maximally mixed, so it possesses a
correlation $\bsC'$. Since $\rho'_{AB}$ has the same spectrum as
$\rho_{AB}$, from the above discussion, $\bsC$ and $\bsC'$ have the
same singular values and the same determinant sign. Therefore both
$\rho'_{AB}$ and $\rho_{AB}$ are LU equivalent to the right hand
side of Eq.~\eqref{eq:normalform}, and thus both $\rho'_{AB}$ and
$\rho_{AB}$ are equivalent.
\end{itemize}
In summary, if both triples
$(\rho_{AB},\rho_A\ot\I_B,\I_A\ot\rho_B)$ and
$(\tau_{AB},\tau_A\ot\I_B,\I_A\ot\tau_B)$ are unitary equivalent,
then $\rho_{AB}$ is LU equivalent to $\tau_{AB}$. This completes the
proof.
\end{proof}

We are interested in the extended problem of above
Theorem~\ref{th:LUb}: For two-qudits $\rho_{AB}$ and $\sigma_{AB}$
in $\density{\bbC^m\ot\bbC^n}$, if a global unitary $\bsW\in\SU(mn)$
such that
\begin{eqnarray}
\begin{cases}
\sigma_{AB}&=\bsW\rho_{AB}\bsW^\dagger ,\\
\sigma_A\ot\I_B &=\bsW(\rho_A\ot\I_B)\bsW^\dagger,\\
\I_A\ot\sigma_B &=\bsW(\I_A\ot\rho_B)\bsW^\dagger,
\end{cases}
\end{eqnarray}
does it hold that $\sigma_{AB}$ and $\rho_{AB}$ are LU equivalent?
In fact, Theorem~\ref{th:LUa} and Theorem~\ref{th:LUb} can be
reformulated as: $\Psi$ is LU similar to $\Phi$ if and only if
$\Tr{\rho_{i_1}\cdots\rho_{i_n}}=\Tr{\sigma_{i_1}\cdots\sigma_{i_n}}$,
where $\rho_{i_k}\in\Psi$ and $\sigma_{i_k}\in\Phi$ for $1\leqslant
k\leqslant N$. In fact, for any multipartite state
$\rho\in\density{\bbC^{d_1}\ot\cdots\ot\bbC^{d_N}}$, denote
$\rho_S:=\Ptr{\bar S}{\rho}$, where $\bar
S:=\set{1,\ldots,N}\backslash S$. We have the following
\emph{conjecture}: For $\rho$ and $\sigma$ are in
$\density{\bbC^{d_1}\ot\cdots\ot\bbC^{d_N}}$,
\begin{center}
both $\rho$ and $\sigma$ are LU similar if and only if
$\Tr{\rho_{i_1}\cdots\rho_{i_n}}=\Tr{\sigma_{i_1}\cdots\sigma_{i_n}}$,
where $\rho_{i_k}\in\Set{\rho_S\ot\I_{\bar
S}:S\subset\set{1,\ldots,N}}$ and
$\sigma_{i_k}\in\Set{\sigma_S\ot\I_{\bar
S}:S\subset\set{1,\ldots,N}}$ for $1\leqslant k\leqslant N$.
\end{center}
From the following result, we will see that the ring of polynomial
invariants can also be generated by LU Bargmann invariants in the
two-qubit system.
\begin{thrm}
The ring of all LU polynomial invariants of a two-qubit state
$\rho_{AB}$ can be generated by a complete set of fundamental
invariants $\set{B_k:k=0,1,\ldots,20}$, where
$\set{B_k:k=1,\ldots,18}$ are taken from Theorem~\ref{th:LUa}, and
\begin{eqnarray}
B_0=\Tr{\rho_{AB}}=1, B_{19} =
\Tr{\bsX^2_0\bsX_1\bsX_2\bsX^3_0\bsX_1},B_{20}=
\Tr{\bsX^2_0\bsX_1\bsX_2\bsX^3_0\bsX_2}.
\end{eqnarray}
\end{thrm}

\begin{proof}
For $k = 1, \ldots, 18$, let both $I_k$ and $L_k$ be defined as in
\cite{Zhang2025PRA2}. Denote
\begin{eqnarray*}
I_{19}= \Inner{\bsa}{\bsC\bsC^\t\bsa\times
\bsC\bsC^\t\bsC\bsb}\quad\text{and}\quad I_{20}=
\Inner{\bsb}{\bsC^\t\bsC\bsb\times \bsC^\t\bsC\bsC^\t\bsa}
\end{eqnarray*}
and
\begin{eqnarray*}
L_{19}= \Inner{\bsC^\t\bsC\bsC^\t\bsa}{\widehat\bsC^\t\bsa\times
\bsb}\quad\text{and}\quad L_{20}=
\Inner{\bsC\bsC^\t\bsC\bsb}{\widehat\bsC\bsb\times \bsa}.
\end{eqnarray*}
We can infer that $I_{19}=L_{19}-L_2L_{15}$ and
$I_{20}=L_{20}-L_2L_{16}$. Thus we can summarize the relationships
between $\set{I_k:=1,\ldots,20}$ and $\set{L_k:k=1,\ldots,20}$ as
below:
\begin{enumerate}[(1)]
\item Let $I_k=L_k$ for $k\in\set{1, 2, 4, 5, 7, 8, 10, 11, 12, 13,
14, 17, 18}$;
\item $I_3=L^2_2-2L_3$;
\item $I_6=L_6+L_2L_5-L_3L_4$;
\item $I_9 = L_9 + L_2L_8 - L_3L_7$;
\item $I_k=-L_k$ for $k\in\set{15,16}$;
\item $I_{19}=L_{19}-L_2L_{15}$;
\item $I_{20}=L_{20}-L_2L_{16}$.
\end{enumerate}
The analytical relationships between $\set{L_k:k=1,\ldots,18}$ and
$\set{B_k:k=1,\ldots,18}$ have already been obtained in
\cite{Zhang2025PRA2}. The invariants $L_{19}$ and $L_{20}$ can be
also expressed by using Bargmann invariants $B_k$: {\scriptsize
\begin{eqnarray*}
L_{19}&=&\frac{4}{3} \mathrm{i} \Big(12 B_1^3+12 B_2 B_1^2-24 B_3 B_1^2-48 B_4 B_1^2-24 B_5 B_1^2-8 B_6 B_1^2+24 B_8 B_1^2-16 B_1^2-18 B_2^2 B_1+48 B_2 B_4^2 B_1-96 B_4^2 B_1\\
&&~~~~~~-17 B_2 B_1+12 B_2 B_3 B_1+38 B_3 B_1-12 B_2^2 B_4 B_1+42
B_2 B_4 B_1+84 B_3 B_4 B_1+68 B_4 B_1+72 B_4 B_5 B_1+84 B_5 B_1\\
&&~~~~~~-16 B_2 B_6 B_1+16 B_3 B_6 B_1+40 B_4 B_6 B_1+12 B_6 B_1-24
B_2 B_7 B_1-24 B_7 B_1+24 B_2 B_8 B_1-72 B_4 B_8 B_1-60 B_8 B_1\\
&&~~~~~~-96 B_9 B_1-24 B_2 B_{10} B_1+48 B_{10} B_1+48 B_{13} B_1-48
B_{14} B_1+3 B_1-12 B_2^2-24 B_2 B_3^2-60 B_3^2-54 B_2 B_4^2+36 B_3
B_4^2\\
&&~~~~~~+42 B_4^2+16 B_6^2-11 B_2+12 B_2^2 B_3+66 B_2 B_3+8 B_3+18
B_2^2 B_4+72 B_3^2 B_4+42 B_2 B_4-72 B_2 B_3 B_4-150 B_3 B_4-20
B_4\\
&&~~~~~~-60 B_2 B_5-48 B_4 B_5+14 B_5+44 B_2 B_6-32 B_3 B_6-112 B_4
B_6+16 B_5 B_6+14 B_6-96 B_3 B_7+144 B_4 B_7+96 B_5 B_7-32 B_6
B_7\\
&&~~~~~~-40 B_7-24 B_2 B_8+48 B_3 B_8+96 B_4 B_8+12 B_8+96 B_3
B_9+48 B_4 B_9+12 B_2 B_{10}-24 B_{10}+96 B_{11}-24 B_{12}+48 B_2
B_{13}\\
&&~~~~~~-144 B_4 B_{13}-168 B_{13}+144 B_{14}+192 B_{15}-96
B_{16}-192 B_{19}+3\Big)
\end{eqnarray*}}
and {\scriptsize
\begin{eqnarray*}
L_{20}&=&\frac{4}{3}\mathrm{i} \Big(12 B_2^3+12 B_1 B_2^2-24 B_3 B_2^2-48 B_4 B_2^2-24 B_5 B_2^2-8 B_6 B_2^2+24 B_7 B_2^2-16 B_2^2-18 B_1^2 B_2+48 B_1 B_4^2 B_2-96 B_4^2 B_2\\
&&~~~~~~-17 B_1 B_2+12 B_1 B_3 B_2+38 B_3 B_2-12 B_1^2 B_4 B_2+42
B_1 B_4 B_2+84 B_3 B_4 B_2+68 B_4 B_2+72 B_4 B_5 B_2+84 B_5 B_2\\
&&~~~~~~-16 B_1 B_6 B_2+16 B_3 B_6 B_2+40 B_4 B_6 B_2+12 B_6 B_2+24
B_1 B_7 B_2-72 B_4 B_7 B_2-60 B_7 B_2-24 B_1 B_8 B_2-24 B_8 B_2\\
&&~~~~~~-96 B_9 B_2-24 B_1 B_{10} B_2+48 B_{10} B_2-48 B_{13} B_2+48
B_{14} B_2+3 B_2-12 B_1^2-24 B_1 B_3^2-60 B_3^2-54 B_1 B_4^2+36 B_3
B_4^2\\
&&~~~~~~+42 B_4^2+16 B_6^2-11 B_1+12 B_1^2 B_3+66 B_1 B_3+8 B_3+18
B_1^2 B_4+72 B_3^2 B_4+42 B_1 B_4-72 B_1 B_3 B_4-150 B_3 B_4-20
B_4\\
&&~~~~~~-60 B_1 B_5-48 B_4 B_5+14 B_5+44 B_1 B_6-32 B_3 B_6-112 B_4
B_6+16 B_5 B_6+14 B_6-96 B_3 B_8+144 B_4 B_8+96 B_4 B_7-32 B_6
B_8\\
&&~~~~~~-40 B_8-24 B_1 B_7+48 B_3 B_7+96 B_5 B_8+12 B_7+96 B_3
B_9+48 B_4 B_9+12 B_1 B_{10}-24 B_{10}+96 B_{12}-24 B_{11}+48 B_1
B_{14}\\
&&~~~~~~-144 B_4 B_{14}-168 B_{14}+144 B_{13}+192 B_{16}-96
B_{15}-192 B_{20}+3\Big).
\end{eqnarray*}}
Moreover, {\scriptsize
\begin{eqnarray*}
2048B_{19}&=& 5 L_7 L_4^3+7 L_4^3+10 L_7^2 L_4^2-12 L_1 L_4^2+10 L_2
L_4^2+6 L_2 L_7 L_4^2+125 L_7 L_4^2+50 L_{12} L_4^2-4 L_{14}
L_4^2+35 L_4^2+L_7^3 L_4+3 L_2^2 L_4\\
&&+2 L_2 L_7^2 L_4+75 L_7^2 L_4+20 L_{12}^2 L_4-88 L_1 L_4-4 L_1 L_2
L_4+28 L_2 L_4-4 L_3 L_4+40 L_5 L_4+L_2^2 L_7 L_4-4 L_1 L_7 L_4+66
L_2 L_7 L_4\\
&&-4 L_3 L_7 L_4+4 L_5 L_7 L_4+175 L_7 L_4-4 L_7 L_8 L_4-8 L_8 L_4-8
L_1 L_{12} L_4+48 L_2 L_{12} L_4+140 L_7 L_{12} L_4+300 L_{12}
L_4-12 L_{13} L_4\\
&&-4 L_2 L_{14} L_4-12 L_7 L_{14} L_4-104 L_{14} L_4-4 i L_{15}
L_4-8 i L_{16} L_4+21 L_4+L_7^3+5 L_2^2+2 L_2 L_7^2+15 L_7^2+12 L_2
L_{12}^2+12 L_7 L_{12}^2\\
&&+188 L_{12}^2+8 L_{14}^2-60 L_1-20 L_1 L_2+10 L_2+20 L_3+8 L_1
L_5+12 L_2 L_5+72 L_5+16 L_6+L_2^2 L_7-60 L_1 L_7+20 L_2 L_7\\
&&+4 L_3 L_7+24 L_5 L_7+15 L_7+4 L_2 L_8+8 L_7 L_8+40 L_8+6 L_2^2
L_{12}+10 L_7^2 L_{12}-72 L_1 L_{12}+104 L_2 L_{12}+24 L_5 L_{12}+16
L_2 L_7 L_{12}\\
&&+180 L_7 L_{12}+8 L_8 L_{12}+130 L_{12}-4 L_2 L_{13}-4 L_7
L_{13}-8 L_{12} L_{13}-4 L_{13}+24 L_1 L_{14}-36 L_2 L_{14}-8 L_5
L_{14}-52 L_7 L_{14}\\
&&-112 L_{12} L_{14}-116 L_{14}+20 i L_{15}-12 i L_2 L_{15}-4 i L_7
L_{15}-8 i L_{16}+16 i L_{17}-8 i L_{18}+8 i L_{19}+1
\end{eqnarray*}}
and {\scriptsize
\begin{eqnarray*}
2048B_{20}&=& 5 L_7^3 L_4+7 L_7^3+10 L_7^2 L_4^2-12 L_1 L_7^2+10 L_2
L_7^2+6 L_2 L_7^2 L_4+125 L_7^2 L_4+50 L_{12} L_7^2-4L_{14} L_7^2
+35 L_7^2+L_7 L_4^3+3 L_2^2 L_7\\
&&+2 L_2 L_7 L_4^2+75 L_7 L_4^2+20 L_{12}^2 L_7-88 L_1 L_7-4 L_1 L_2
L_7+28 L_2 L_7-4 L_3 L_7+40 L_8 L_7+L_2^2 L_7 L_4-4 L_1 L_7 L_4+66
L_2 L_7 L_4\\
&&-4 L_3 L_7 L_4+4 L_8 L_7 L_4+175 L_7 L_4-4 L_7 L_5 L_4-8 L_5 L_7-8
L_1 L_{12} L_7+48 L_2 L_{12} L_7+140 L_7 L_{12} L_4+300 L_{12}
L_7-12 L_{13} L_7\\
&&-4 L_2 L_{14} L_7-12 L_7 L_{14} L_4-104 L_{14} L_7-4 i L_{16}
L_4-8 i L_{15} L_7+21 L_7+L_4^3+5 L_2^2+2 L_2 L_4^2+15 L_4^2+12 L_2
L_{12}^2+12 L_4 L_{12}^2\\
&&+188 L_{12}^2+8 L_{14}^2-60 L_1-20 L_1 L_2+10 L_2+20 L_3+8 L_1
L_8+12 L_2 L_8+72 L_8+16 L_9+L_2^2 L_4-60 L_1 L_4+20 L_2 L_4\\
&&+4 L_3 L_4+24 L_8 L_4+15 L_4+4 L_2 L_5+8 L_4 L_5+40 L_5+6 L_2^2
L_{12}+10 L_4^2 L_{12}-72 L_1 L_{12}+104 L_2 L_{12}+24 L_8 L_{12}+16
L_2 L_4 L_{12}\\
&&+180 L_4 L_{12}+8 L_5 L_{12}+130 L_{12}-4 L_2 L_{13}-4 L_4
L_{13}-8 L_{12} L_{13}-4 L_{13}+24 L_1 L_{14}-36 L_2 L_{14}-8 L_8
L_{14}-52 L_4 L_{14}\\
&&-112 L_{12} L_{14}-116 L_{14}+20 i L_{16}-12 i L_2 L_{16}-4 i L_7
L_{16}-8 i L_{15}+16 i L_{18}-8 i L_{17}+8 i L_{20}+1.
\end{eqnarray*}}
The above computations indicate that the same ring of invariant
polynomials can be generated by two sets of generators,
respectively: $\set{L_k:k=0,1,\ldots,20}$ and
$\set{B_k:k=0,1,\ldots,20}$.
\end{proof}

\subsection{Entanglement detection via Bargmann invariant}

LU Bargmann invariants are useful for entanglement detection because
the partial-transposed moments (PT-moments) of various orders can be
expressed in terms of them. In two-qubit systems, where entanglement
is completely determined by these PT-moments, this leads to the
following physical and operational criterion.

\begin{thrm}[\cite{Zhang2025PRA2}]
A two-qubit state $\rho_{AB}$ is entangled if and only if the
following subset of 7 LU Bargmann invariants
$\set{B_k:k=1,2,3,4,5,6,10}$ satisfies the following inequality:
\begin{eqnarray}
6(B_1+B_2-B_1B_2-B_4-B_{10})+12(B_5-B_3)+3B^2_4+4B_6<1.
\end{eqnarray}
Equivalently,
\begin{eqnarray}
&&6\Br{\Tr{\rho^2_A}+\Tr{\rho^2_B}-\Tr{\rho^2_A}\Tr{\rho^2_B}-\Tr{\rho^2_{AB}}-\Tr{\rho^4_{AB}}}\notag\\
&&+12\Br{\Tr{\rho^2_{AB}(\rho_A\ot\rho_B)}-\Tr{\rho_{AB}(\rho_A\ot\rho_B)}}\notag\\
&&+3\Br{\Tr{\rho^2_{AB}}}^2+4\Tr{\rho^3_{AB}}<1.
\end{eqnarray}
\end{thrm}

\begin{proof}
A two-qubit state $\rho_{AB}$ is entangled if and only if
$\det(\rho^\Gamma_{AB}) < 0$, where $\rho^\Gamma_{AB}$ denotes its
partial transpose with respect to either one subsystem. This
condition can be verified by expressing $\det(\rho^\Gamma_{AB})$ in
terms of generators of the LU Bargmann invariants. Remarkably, only
seven of the 18 generators ${B_k}$ are needed for this expression.
The specific details are provided in \cite{Zhang2025PRA2}.
\end{proof}
The application of LU Bargmann invariants to entanglement detection
in higher-dimensional systems remains an active and developing area
of research.

\section{Concluding remarks}

In this survey, we have provided a comprehensive overview of
Bargmann invariants, tracing their evolution from a foundational
concept in the theory of ray spaces and geometric phases to a
powerful and versatile toolkit in modern quantum information
science. Central to our exposition is the recognition that these
gauge-invariant quantities encode the intrinsic relational and
geometric structure of quantum states, offering a unified language
for addressing problems ranging from the classification of unitary
orbits to the operational detection of quantum entanglement.

A key contribution of this review is the systematic characterization
of the admissible set of $n$-th order Bargmann invariants. We have
established that $\cB_n = \cB_n|_{\text{circ}}$, demonstrating that
every valid set of invariants admits a circulant Gram matrix
representation. This result, together with the identification $\cB_n
= \cB_n(2)$ and the explicit description of the boundary curve
$\partial \cB_n$, resolves a longstanding open problem regarding the
convexity of the set $\cB_n$. Furthermore, we have presented an
envelope-based geometric approach for $n = 3, 4$, which provides an
alternative route to characterizing these sets and highlights the
deep connection between Bargmann invariants and the numerical range
of cyclic operators.

The practical relevance of Bargmann invariants is further
highlighted by their diverse applications in quantum information
processing. In particular, they enable the direct estimation of
relational information using simple quantum circuits---such as the
cycle test---without the need for full state tomography. Their
utility also extends to witnessing genuine quantum features,
including quantum imaginary and coherence, as well as to
distinguishing local unitary equivalence and detecting entanglement
via partial-transpose moments. For two-qubit systems, we have
explicitly constructed a set of 18 LU Bargmann invariants that fully
characterize the local unitary orbit and yield a necessary and
sufficient criterion for entanglement. Moreover, we have established
that the entire polynomial invariant ring is generated by these 18
LU Bargmann invariants together with three additional generators.

Despite the significant progress reported here, several important
questions remain open. These include a full characterization of
$\cB_n$ for $n \geqslant 5$ via the envelope method, the operational
interpretation of the condition $z \in \cB_{n+1} \setminus \cB_n$,
and the extension of the present framework to higher-dimensional
multipartite systems, where the conjecture regarding local unitary
equivalence through LU Bargmann invariants awaits further
validation. Moreover, the role of Bargmann invariants in resource
theories, such as coherence and contextuality, and their potential
in quantum algorithm design, represent promising directions for
future research.

In summary, Bargmann invariants are far more than mathematical
curiosities. They are fundamental probes of the non-commutative,
relational, and geometric fabric of quantum mechanics. By bridging
foundational geometry with practical quantum information tasks, they
offer a unifying perspective that is both conceptually deep and
experimentally relevant. We hope this survey will serve as a useful
reference and inspire further theoretical and experimental
investigations into the rich structure of quantum state space.

\subsection*{Acknowledgement}
This research is supported by Zhejiang Provincial Natural Science
Foundation of China under Grants No. LZ23A010005.


\end{document}